\documentclass[twoside,11pt]{article}

\usepackage[preprint]{jmlr2e}
\usepackage{mathtools}
\usepackage{enumitem, bm,placeins, verbatim, subcaption}
\usepackage{amsmath,amsfonts}
\usepackage{algorithm, algpseudocode}
\usepackage{tikz}
\usepackage{multicol}
\usepackage{multirow, array}
\usetikzlibrary{shapes,arrows,positioning}

\newcommand{\E}{{\mathbb E}}       
\newcommand{\pa}{{\rm pa}}       
\newcommand{\ch}{{\rm ch}}       
\newcommand{\sib}{{\rm sib}}	   
\newcommand{\an}{{\rm an}}	   
\newcommand{\An}{{\rm An}}	   
\newcommand{\de}{{\rm de}}	   
\newcommand{\pspa}{{{\rm pa}_D}}      
\newcommand{\psan}{{{\rm an}_D}}       
\newcommand{\psAn}{{{\rm An}_D}}       
\newcommand{\Ktop}{K_{\rm top}}
\newcommand{\Kbttm}{K_{\rm bttm}}
\newcommand{\irr}{{\rm irr}}
\newcommand{\Irr}{{\rm Irr}}

\newcommand{\independent}{\protect\mathpalette{\protect\independenT}{\perp}}
\def\independenT#1#2{\mathrel{\rlap{$#1#2$}\mkern2mu{#1#2}}}
\newcommand\numberthis{\addtocounter{equation}{1}\tag{\theequation}}

\jmlrheading{1}{2000}{1-48}{4/00}{10/00}{}{Y. Samuel Wang and Mathias Drton}

\ShortHeadings{Causal Discovery with Unobserved Confounding}{Wang and Drton}
\firstpageno{1}

\begin{document}

\title{Causal Discovery with Unobserved Confounding and Non-Gaussian Data}

\author{\name Y. Samuel Wang \email ysw7@cornell.edu \\
       \addr Department of Statistics and Data Science\\
Cornell University\\
       Ithaca, NY 14853, USA
       \AND
       \name Mathias Drton \email mathias.drton@tum.de \\
       \addr Department of Mathematics\\
       Technical University of Munich\\
       85748 Garching bei M\"{u}nchen, Germany}

\editor{}

\maketitle

\begin{abstract}
We consider 
recovering causal structure from multivariate observational data. We assume the data arise from a linear structural equation model (SEM) in which the idiosyncratic errors are allowed to be dependent in order to capture possible latent confounding. Each SEM can be represented by a graph where vertices represent observed variables, directed edges represent direct causal effects, and bidirected edges represent dependence among error terms. Specifically, we assume that the true model corresponds to a bow-free acyclic path diagram; i.e., a graph that has at most one edge between any pair of nodes and is acyclic in the directed part. We show that when the errors are non-Gaussian, the exact causal structure encoded by such a graph, and not merely an equivalence class, can be recovered from observational data. The 
method we propose for this purpose uses estimates of suitable moments, but, in contrast to previous results, does not require specifying the number of latent variables a priori. We also characterize the output of our procedure when the assumptions are violated and the true graph is acyclic, but not bow-free.  
We illustrate the effectiveness of our procedure in simulations and an application to an ecology data set.
\end{abstract}

\section{Introduction}
We consider the problem of discovering causal structure from multivariate data when only observational data is available, but latent confounding may exist between the observed variables. In our main result we show that if the data are generated under a recursive linear structural equation model with non-Gaussian idiosyncratic errors, the exact causal structure can be recovered provided the confounding is limited to pairs of variables which
do not have a direct effect on each other. These models correspond to
\emph{bow-free acyclic path diagrams} (BAPs). 
\subsection{Linear structural equation models and graphs}
Structural equation models (SEMs) are multivariate statistical models
that encode causal relationships and are popular in the social and
biological sciences \citep{bollen2014structural, shipley2016cause}. SEMs may be formulated to explicitly include latent unobserved variables, but in this article we consider a setup in which the latent
variables have been marginalized out and the models only explicitly refer to effects of the
unobserved variables.  This approach has, in particular, been fruitful for causal discovery \citep{handbook:evans}.


A linear SEM assumes that we observe a sample comprised of i.i.d. copies of a random vector $Y = (Y_{v} : v \in V)$ that solves the equation system
\begin{equation}
\label{eq:sem}
Y_{v} = \sum_{u \neq v}\beta_{v,u}Y_{u} + \varepsilon_{v}, \quad v\in V.
\end{equation}
The direct effect of $Y_u$ on $Y_v$ is encoded by $\beta_{v,u}$, and $\varepsilon_{v}$ is an idiosyncratic error term of mean zero.  Note that our setup assumes throughout, and without loss of generality,  that $Y$ is \emph{centered}.
We collect the effects $\beta_{v,u}$ into a $p \times p$ matrix $B =
\left(\beta_{v,u}\right)_{u, v \in V}$ and the error terms into a
vector $\varepsilon = (\varepsilon_v)_{v \in V}$.  
Because the matrix $B$ encodes the direct causal effect of $Y_{u}$ onto $Y_{v}$ for all $u, v \in V$, we will use the term \emph{direct effects} to refer to the matrix $B$. Each copy of the error vector 
$\varepsilon$ is drawn i.i.d.\ with expectation 0, but
we allow for unobserved confounding between different variables, say $Y_{v}$ and $Y_{u}$, by allowing the corresponding errors, $\varepsilon_{v}$ and $\varepsilon_{u}$, to be dependant. 
In vector form~\eqref{eq:sem} reads $Y = BY + \varepsilon$, which is uniquely solved by 
\begin{equation}\label{eq:linearTransformErrors}
Y = (I-B)^{-1}\varepsilon    
\end{equation} when
$I-B$ is invertible. Letting $\Omega := \mathbb{E}(\varepsilon
  \varepsilon^T ) =(\omega_{v,u})_{u,v \in V}$ be the covariance matrix of $\varepsilon$, we obtain that the covariance matrix of the observed variables in $Y$ is
  \begin{equation}\label{eq:sigmaDef}
\Sigma \coloneqq \E(YY^T) = (I-B)^{-1}\Omega(I-B)^{-T}.
\end{equation}

Throughout the article, we describe SEMs using the language of
graphical models or path diagrams \citep{handbook}.
We represent 
each SEM
by a mixed graph
$G = (V, E_\rightarrow, E_\leftrightarrow )$, where each vertex
$v \in V$ corresponds to an observed variable, and $E_\rightarrow$ and
$E_\leftrightarrow$ are sets of directed edges and bidirected edges,
respectively. Let $u,v\in V$ be two distinct vertices.  We represent a direct effect of $u$ on  $v$ by the
directed edge $u \rightarrow v \in E_\rightarrow$ and say that $u$ is
a \emph{parent} of \emph{child} $v$. So,
$\beta_{v,u} \neq 0$ only if $u$ is a parent of $v$. If there exists a
sequence of directed edges from $u$ to $v$, we say that $u$ is an
\emph{ancestor} of its \emph{descendant} $v$. Unobserved confounding
between $v$ and $u$ is represented by
$v \leftrightarrow u \in E_\leftrightarrow$, and we say that $v$ and
$u$ are \emph{siblings}. So, $\omega_{v,u} \neq 0$ only if $u$ and $v$ are
siblings.  The sibling relation is symmetric; i.e., $u$ being a sibling of $v$ implies that $v$ is a sibling of $u$.  We denote the sets of parents, children, ancestors, descendants, and siblings of $v$ as $\pa(v)$, $\ch(v)$, $\de(v)$, and $\sib(v)$, respectively.  We let $\An(v):=\an(v) \cup \{v\}$.  The problem of interest is then to infer the 
graph corresponding to a given data-generating SEM.

We assume that the true model is a recursive SEM.  
In graphical terms, this means that the mixed graph $G = (V, E_\rightarrow, E_\leftrightarrow )$ corresponding to the model is acyclic in the sense of not containing any directed cycles.  There then exists a (not necessarily unique) total ordering of $V$ under which $u \prec v$ implies $v \not \in \an(u)$.  If, in addition, the graph $G$ contains no bidirected edges, i.e., $E_\leftrightarrow = \emptyset$, then $G$ is a \emph{directed acyclic graph} (DAG).
A \emph{bow} in $G$ is a subgraph of two vertices $u$ and $v$ that contains both a directed and a bidirected
edge; i.e., $u \leftrightarrow v$ and either $u \rightarrow v$ or $v \rightarrow u$.  In this article, we primarily consider mixed graphs that are acyclic and do not contain bows, though in Section~\ref{sec:modelMiss} we also consider the case where the true graph may contain bows.   
Following \citet{drton2009computing}, we
refer to these graphs as \emph{bow-free acyclic path} diagrams (BAPs).
This class of graphs was also considered, e.g., by
\citet{brito:2002} and \cite{nowzohour2015structure}. 

\subsection{Previous work}

Most work on causal discovery with latent variables 
focuses on recovering causal structure in the form of an \emph{ancestral graph}. 
For settings without selection effects, as considered here, ancestral graphs are special cases of BAPs that satisfy the additional restriction that $\an(v)\cap\sib(v)=\emptyset$ for all nodes $v$. 
By adding bidirected edges to $E_\leftrightarrow$, every ancestral graph $G$ can be transformed into a \emph{maximal ancestral graph} (MAG) while preserving the conditional independence relations in $G$.
Gaussian MAG models can then be entirely characterized by conditional independence \citep{richardson:2002}.
However, for any MAG there are generally other MAGs that are Markov equivalent, i.e., encode the same set of conditional independence relations.  Markov equivalent MAGs have the same adjacencies
but the edges may be of different orientations or types.  The Markov equivalence class of any MAG can be compactly represented by a \emph{partial ancestral graph} (PAG) \citep{ali2009markov}. 


\citet{spirtes2000causation} 
proposed the Fast Causal Inference algorithm (FCI) to estimate the PAG corresponding to the underlying causal graph. \citet{zhang2008completeness} added additional orientation rules such that the output of FCI is complete. \citet{colombo:2012}, \citet{claassen2013learning}, and \cite{chen2021causal}  develop additional variants---RFCI, FCI+, and lFCI, respectively---which only require a polynomial number of conditional independence tests if the degree of the graph is bounded, or in the last case exploit possible local separation properties.  \citet{triantafillou2016score} select a MAG via a greedy search which maximizes a penalized Gaussian likelihood, and \citet{bernstein2020permutation} propose a greedy search over partial orderings of the variables.
Efforts are also underway to refine the picture provided by conditional independence by considering additional non-parametric constraints imposed by SEM \citep{verma1990equivalence,
shpitser2014introduction, evans2016margins}.  
In a different vein, 
\citet{nowzohour2015structure} propose a greedy search which assumes that the true model is a linear SEM with Gaussian errors which corresponds to a BAP.

The previously mentioned methods have enjoyed great success but operate in a regime in which only an equivalence class of graphs (e.g., via the PAG) can be discovered and different graphs in the equivalence class may have conflicting causal interpretations.
%
%
In contrast,  \citet{shimizu2006lingam} show that when the true model is a recursive linear SEM with \emph{non-Gaussian} errors, the exact graph---not just an equivalence class---can be identified from observational data using independent component analysis (ICA).  
Instead of ICA, the subsequent DirectLiNGAM \citep{shimizu2011direct} and Pairwise LiNGAM \citep{hyvarinen2013pairwise} methods use an iterative procedure to estimate a causal ordering. \citet{wang2020hdng} give a modified method that is also consistent in high-dimensional settings in which the number of variables $p$ exceeds the sample size $n$. However, all of the above methods for the linear non-Gaussian acyclic model (LiNGAM) do not allow for possible latent confounding.

\citet{hoyer2008estimation} consider the setting where the data is generated by a LiNGAM model, but some variables are unobserved.
Using existing results from overcomplete ICA, they show that the canonical DAG---roughly a DAG in which all unobserved variables have no parents and at least two children---can be identified when all parent-child pairs in the observed set are unconfounded. However, the result critically requires the number of latent variables in the canonical model to be known in advance and requires all unobserved confounding to be linear. For example, suppose $v \in \sib(u)$, and the confounding is caused by a hidden variable $Y_h$. Then a generative procedure where $\check\varepsilon_v \independent \check \varepsilon_u$, $\varepsilon_v = \check \varepsilon_v + \alpha_v Y_h$, and $\varepsilon_u = \check \varepsilon_u + \alpha_u Y_h$ would be allowed; however, $\varepsilon_u = \check \varepsilon_u + \alpha_u Y_h^2$ would be precluded. Furthermore, even when the model is correctly specified, \citet{shimizu2014bayesian} state ``current versions of the overcomplete ICA algorithms are not very computationally reliable since they often suffer from local optima,'' and indeed \citet{hoyer2008estimation} use a maximum likelihood procedure with mixtures of Gaussians instead of overcomplete ICA in their simulations.

To avoid using overcomplete ICA and improve practical performance, \citet{entner2010discovering} and \citet{tashiro2014parcelingam} both propose procedures which test subsets of the observed variables and seek to identify as many pairwise ancestral relationships as possible; i.e., either (1) $u \in \an(v)$, (2) $v \in \an(u)$, or (3) $v \not \in \an(u)$ and $u \not \in \an(v)$. \citet{entner2010discovering} apply ICA to all subsets of the observed variables which 
do not have latent confounding. \citet{tashiro2014parcelingam} apply an iterative procedure similar to DirectLiNGAM to each subset of variables.  They show that the procedure used for certifying ancestral relationships is sound in the presence of confounding, but do not characterize the class of graphs which can be identified. In the appendix, we show a simple ancestral graph that cannot be discovered using the method of \citet{entner2010discovering}.  For ParcelLiNGAM, we show in Section~\ref{sec:ancestral} that all ancestral relationships can indeed be discovered when the true causal graph itself is ancestral, but that the method will not identify all ancestral relationships for any non-ancestral BAP. 

The identifiability results of Section~\ref{sec:bap} and \ref{sec:algo} are based on Chapter 4 of \citet{wang2018thesis}, the Ph.D.~dissertation of the first author. As we were preparing the manuscript for submission, new work was published by \citet{maeda2020causal} and \cite{salehkaleybar2020learning}. \citet{maeda2020causal} propose Repetitive Causal Discovery (RCD) for discovering mixed graphs. RCD uses a constraint-based algorithm, and similar to our approach---but in constrast to \citet{tashiro2014parcelingam}---RCD iteratively uses previously discovered structure to inform later steps. However, in Section~\ref{sec:ancestral} we show that RCD is not able to identify all BAPs. Similar to \citet{hoyer2008estimation}, \citet{salehkaleybar2020learning} use overcomplete ICA and thus crucially require all confounding to be linear. They extend the results of \citet{hoyer2008estimation} by showing that under weak conditions, the total number of variables (unobserved and observed) in the system can be identified and a causal ordering can be identified from population quantities. However, in order to identify causal effects (i.e., determine the graph beyond just ancestral relations), they require a condition which precludes any non-ancestral graphs \citep[Assumption 2]{salehkaleybar2020learning}.


\subsection{Contribution}
In this work, we show that when the data are generated by a linear non-Gaussian SEM that is based on a BAP, then the exact BAP---not just an equivalence class---can be consistently recovered. This implies that the causal effects can also be identified.
Specifically, we show how to recover the BAP from low-order moments, avoiding the use of overcomplete ICA.
Our result does not require knowledge of the number of latent variables or knowledge about the distribution of the errors.  It also does not require linearity in how the observed variables depend on the unobserved variables.  It does, however, rely on a genericity assumption for the linear coefficients and error moments that, in particular, rules out Gaussian behavior of the considered moments.  

The \textbf{B}ow-free \textbf{A}cyclic \textbf{n}on-\textbf{G}aussian (BANG) method we propose for recovery of BAPs uses a series of independence tests between (suitably estimated) regressors and residuals to certify causal structure. 
When the maximum in-degree (both directed and bidirected edges) is bounded, the total number of tests performed is bounded by a polynomial in the number of variables considered. We also characterize what the BANG procedure will return---given population values---when the model is misspecified and the true generating procedure corresponds to a graph with bows.  In simulations, we confirm that the method reliably discovers exact causal structure when given a large enough sample and outperforms existing methods in some settings.

\subsection{Preliminaries}
Throughout, we often let a node $v \in V$ stand in for the variable $Y_v$. For a set $C \subset V$, we let $Y_C = (Y_c : c \in C)$ be the vector of variables indexed by an element of $C$.  Furthermore, for a matrix $B$ and index sets $R$ and $C$, let $B_{R, C}$ be the submatrix of $B$ corresponding to the $R$th rows and $C$th columns. For some positive integer $z$, we also let $[z]$ indicate the set $\{1, \ldots, z\}$.


The notions of \emph{sound} (i.e., correct) and \emph{complete} (i.e., maximally informative) are often used to describe the output of causal discovery methods based on conditional independence tests~\citep{spirtes2000causation}. We adapt these notions for discovery of ancestral relationships in a mixed graph $G$. Let $\prec_0$ be a (potentially partial) ordering of $V$. We say that the ordering $\prec_0$ is \emph{sound} with respect to ancestral relationships in $G$ if $u \prec_0 v$ implies that $v \not \in \an(u)$ holds. We say that the ordering $\prec_0$ is \emph{complete} with respect to ancestral relationships in $G$ if $u \in \an(v)$ implies $u \prec_0 v$.

\section{Ancestral graphs}\label{sec:ancestral}
Before discussing the main results, we first build intuition for causal discovery with non-Gaussian data by considering the simpler setting of ancestral graphs. We show that given population information, the previously proposed ParcelLiNGAM\footnote{\citet{tashiro2014parcelingam} give two variants of the ParcelLiNGAM algorithm which they label Algorithm 2 and 3. Algorithm 3 requires less computation and applies Algorithm 2 to a subset of the variables.} procedure \citep{tashiro2014parcelingam} is sound and complete for ancestral relationships when the true graph is ancestral. However, it is not complete for certifying ancestral relationships in non-ancestral BAPs. We also give an example of a BAP which the RCD method \citep{maeda2020causal} can not identify.

\subsection{Determining causal relationships in ancestral graphs}
Recall from~\eqref{eq:linearTransformErrors} that $Y = (I-B)^{-1} \varepsilon$ so that $Y_{v}$ is  a linear combination of $\varepsilon_{s}$ for all $s$ such that $\left[(I-B)^{-1}\right]_{v,s} \neq 0$. For generic linear coefficients, this set is equal to $\an(v)$.
Thus, for $c \not \in \sib(v) \cup \de(\sib(v)) \cup \de(v)$, the variable $Y_{c}$ is a linear combination of error terms which are independent of $\varepsilon_{v}$, i.e., $Y_{c} \independent \varepsilon_{v}$. 
%
Thus, for $v \in V$ and a set $C \subseteq V\setminus \{v\}$ such that $\pa(v) \subseteq C \subseteq V \setminus [\sib(v) \cup \de(\sib(v)) \cup \de(v)]$, the population regression coefficients for predicting $v$ from $C$ 
are
\begin{equation}\begin{aligned}\label{eq:ancestralSound}
D_{v, C}^T &= \left[\E\left(Y_{C} Y_{C}^T\right)\right]^{-1}\E\left(Y_{C} Y_{v}\right)\\
&= \left[\E\left(Y_{C} Y_{C}^T\right)\right]^{-1}\E\left(Y_{C} (Y_{C}^T B_{v,C}^T + \varepsilon_{v}) \right)\\
&= \left[\E\left(Y_{C} Y_{C}^T\right)\right]^{-1}\left[\E\left(Y_{C} Y_{C}^T\right) B_{v,C}^T + \E\left(Y_{C} \varepsilon_{v}\right) \right] =  B_{v, C}^T,
\end{aligned}\end{equation}
where $B_{v, C}=(\beta_{v,u})_{u\in C}$ is comprised of the direct effects of $C$ onto $v$. The last equality in~\eqref{eq:ancestralSound} crucially requires that $\E\left(Y_{C} \varepsilon_{v}\right) = 0$, as implied by the independences pointed out above. The regression residual $\eta_{v.C}$ obtained from the coefficients in $D_{v,C}$ then satisfies
\begin{equation}
\eta_{v.C} := Y_{v} - D_{v, C}  Y_{C}  = Y_{v} - B_{v, \pa(v)} Y_{\pa(v)}  = \varepsilon_{v}.
\end{equation}
As noted, $\eta_{v.C}$ is independent of the regressors $Y_C$. 

In contrast, if $C$ contains a descendant of $v$, a sibling of $v$, or a descendant of a sibling of $v$, then in general $\E\left(Y_{C} \varepsilon_{v}\right) \neq 0$, $D_{v,C} \neq B_{v,C}$, and  $\eta_{v.C} \neq \varepsilon_{v}$. It follows, in general, that there exists some $c \in C$ such that $\eta_{v.C} \not \independent Y_c$. Although the first order conditions of the least squares criterion ensure that regressors and residuals are uncorrelated, when the errors are non-Gaussian, dependence can still be detected by using a non-parametric independence test \citep{gretton2005measuring, szekely2009dist, bergsma2014consistent, pfister2018independence} or examining the higher order moments---e.g., $\E(Y_c^k \varepsilon_v)$ for $k > 1$ \citep{wang2020hdng}. Non-Gaussian errors are crucial because for a Gaussian random variable, uncorrelated and independent are equivalent so the residuals are independent of the regressors regardless of $C$. But when the errors are non-Gaussian, the independence of residuals and regressors can be used to certify that $C \subseteq V \setminus [\sib(v) \cup \de(\sib(v)) \cup \de(v)]$.  

This idea can be directly applied to discover a topological ordering of the variables by finding the largest set $C^{(\max)}_v$ such that the residual when regressing $v$ onto $C^{(\max)}_v$ is independent of $Y_{C^{(\max)}_v}$. When $G$ is ancestral then $\an(v) \subseteq C^{(\max)}_v = V \setminus [ \sib(v) \cup \de(\sib(v)) \cup \de(v)]$. Thus, to form $\hat \prec$, an initial estimate of a topological ordering, we can set $c \hat{\prec} v$ for all $c \in C^{(\max)}_v$. When there is no unique total ordering, there may be pairs $u, v$ such that $u \in C^{(\max)}_v \setminus \an(v)$ and $v \in C^{(\max)}_u \setminus \an(u)$. In this case, we can simply remove either (or both) $u \hat{\prec} v$ or $v \hat{\prec} u$ from the initial ordering to obtain a relation that is a valid ordering.

The basic intuition of certifying an ancestral relationship by testing independence of residuals and regressors motivates the DirectLiNGAM \citep{shimizu2011direct} and Pairwise LiNGAM \citep{hyvarinen2013pairwise} procedures. To begin, one can select a root node, one without any parents or latent confounding, by finding a variable which is independent of all the residuals formed by regressing another variable onto it. Once a root is identified, its effect on the remaining variables can be removed and the root finding procedure recurs on the sub-graph of the remaining variables. The sequence of selected roots forms a topological ordering of the variables. An ordering of the nodes can also be identified in the opposite direction by finding sinks---nodes which have no children or latent confounding---by testing whether the residuals of a variable, when regressed onto all other variables, is independent of all other variables. Once a sink is identified, we simply recur onto the sub-graph of the remaining variables.  We use \emph{top-down} to refer to a procedure which successively identifies roots, and we use \emph{bottom-up} to refer to a procedure which successively identifies sinks. 


When we allow for latent confounding, a (certifiable) root or sink may not exist in the graph or in one of subsequent sub-graphs considered, so the Pairwise lvLiNGAM \citep{entner2010discovering} and ParcelLiNGAM \citep{tashiro2014parcelingam} procedures aim to estimate ancestral relationships between pairs of variables rather than a total ordering. We show in the appendix that lvLiNGAM 
may fail to discover even simple ancestral graphs, so we focus our discussion primarily on ParcelLiNGAM. Roughly speaking, ParcelLiNGAM applies both the top-down and bottom-up procedure to all subsets of $V$ and certifies as many ancestral relationships as possible. \citet{tashiro2014parcelingam} show that the certification procedure is sound, but do not characterize a class of graphs for which the entire ParcelLiNGAM procedure is complete. In Lemma~\ref{thm:parceLCorrect}, we show that given population values, ParcelLiNGAM is indeed sound and complete for all ancestral graphs. Details are left for the appendix, but in short, we show that when a graph is ancestral, applying the bottom-up procedure to the subset $\An(v)$ will identify that $\an(v) \prec v$ for all $v$. 


\begin{lemma}\label{thm:parceLCorrect}
	 Suppose $Y$ is generated by a recursive linear SEM that corresponds to an ancestral graph $G$. With generic model parameters and population information (i.e., the distribution of $Y$), the ordering, $\hat{\prec}$, returned by Algorithm 2 of ParcelLiNGAM \citep{tashiro2014parcelingam} is sound and complete for ancestral relationships in $G$. 
\end{lemma}

\subsection{Non-ancestral graphs}
When $G$ is not ancestral, the set of ancestral relationships that are certified by the described approach 
is still sound, but in general it is not complete.  Indeed, in a non-ancestral graph, there exists some $v \in V$ such that $\sib(v) \cap \an(v) \neq \emptyset$. Thus, even if $c \in C = \pa(v)$, it is generally not true that $\varepsilon_v \independent Y_c$. This implies that $E\left(Y_{C} \varepsilon_{v}\right) \neq 0$ and the population regression coefficients $D_{v,C}$ no longer coincide with direct effects $B_{v,C}$. Indeed, in Lemma~\ref{thm:parcelNotComplete} we show that ParcelLiNGAM is no longer complete for non-ancestral BAPs; i.e., in every graph $G$ which is bow-free but not ancestral, there are ancestral relations which will not be identified. We also show in the appendix that RCD~\citep{maeda2020causal} cannot identify the BAP shown in Figure~\ref{fig:exampleBAP}.

\begin{lemma}\label{thm:parcelNotComplete}
	Suppose $Y$ is generated by a recursive linear SEM that corresponds to a graph $G$ which is bow-free but not ancestral. With generic parameters and population information, both Algorithm 2 and Algorithm 3 of ParcelLiNGAM \citep{tashiro2014parcelingam}  will return a partial ordering which is sound, but not complete for ancestral relationships in $G$.
\end{lemma}


As a preview of the work ahead, in Example~\ref{ex:nonAncest}, we exhibit some of the complexities of discovering a non-ancestral graph by testing independence of residuals and regressors. 

\begin{example}\label{ex:nonAncest}
\begin{figure}[t]
    \centering
\begin{tikzpicture}[->,>=triangle 45,shorten >=1pt,
		auto,
		main node/.style={ellipse,inner sep=0pt,fill=gray!20,draw,font=\sffamily,
			minimum width = .5cm, minimum height = .5cm}]
		
		\node[main node] (1) {1};
		\node[main node] (2) [right= .7cm of 1]  {2};
		\node[main node] (3) [right = .7cm of 2]  {3};
		\node[main node] (4) [right = .7cm of 3]  {4};
		
		\path[color=black!20!blue,style={->}]
		(1) edge node {} (2)
		(2) edge node {} (3)
		(3) edge node {} (4)
		;
		
		\path[color=black!20!red,style={<->}]
		(1) edge[bend left = 50] node {} (3)
		(2) edge[bend left = 50] node {} (4)
		(1) edge[bend right = 40] node {} (4);
	\end{tikzpicture}
	\caption{\label{fig:exampleBAP}A non-ancestral BAP which cannot be identified by Pairwise lvLiNGAM, ParcelLiNGAM, or RCD. However, it can be identified by BANG.}
\end{figure}
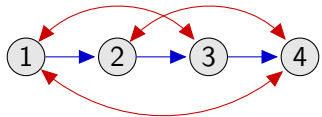

Consider discovering ancestral relationships in the BAP displayed in Figure~\ref{fig:exampleBAP}. 

\emph{Nodes 1 and 2}:  For this unconfounded pair, the 
direct approach of regressing $Y_2$ onto $Y_1$ yields the regression
coefficent $\E(Y_1Y_2)/\E(Y_1^2)=\beta_{21}$ and the residual
$Y_2-\beta_{21}Y_1=\varepsilon_2$ is independent of the regressor
$Y_1=\varepsilon_1$.  This independence certifies precedence of 1
before 2 in the graph, a relationship that would be discovered by ParcelLiNGAM and RCD. 

\emph{Nodes 2 and 3}: For general distributions in the model, there
will not exist $d_{32}\in\mathbb{R}$ such that
$Y_3 - d_{32} Y_2 \independent Y_2$ since $Y_2$ depends on
$\varepsilon_1$ and $1 \in \sib(3)$. However, we can consider
replacing $Y_2$ by an adjusted regressor.  Having established that
$1 \rightarrow 2$, we may take the adjusted regressor to be the
residual $Y_2-\beta_{21}Y_1=\varepsilon_2$ found when regressing $Y_2$ onto
$Y_1$ in the above first step.  The choice
$d_{32} = \beta_{32}$ then yields that
$Y_3 - d_{32} \varepsilon_2 = \varepsilon_3 +
\beta_{32}\beta_{21}\varepsilon_{1} \independent \varepsilon_2$.  This
independence establishes that $2$ precedes $3$. 



\emph{Nodes 3 and 4}: At this point the latent confounding is such
that the strategy considered so far fails.  Indeed, for a general
distribution in the model, there is no coefficient $d_{43}$ such that
the residual $Y_4-d_{43}Y_3'$ is independent of the adjusted regressor
$Y_3'=Y_3 - \beta_{32} \varepsilon_2 = \varepsilon_3 +
\beta_{32}\beta_{21}\varepsilon_{1}$ from the previous step (since
$1 \in \sib(4)$).  Other regressors such as $Y_3'=Y_3$ or
$Y_3'=\varepsilon_3$ also do not yield independence as $1 \in \sib(4)$
or $1 \in \sib(3)$, respectively.

Nevertheless, we can progress by giving up on the regression focus.
Observe instead that setting $d_{43} = \beta_{43}$ yields
$\varepsilon_3 \independent \varepsilon_4 = Y_4 - d_{43} Y_3$. In this
case, the independence test focuses on the error term and not the
regressor and is able to certify that 3 precedes 4. In the sequel, we
will show that this alternative type of independence certificate is in
fact sound and complete for parental relationships in BAPs.  In the
present example, since there are no other parental relationships that
can be certified, but dependencies still remain, we may conclude
(correctly) that
all other pairs are siblings.
\end{example}

\section{Bow-free acyclic path diagrams}\label{sec:bap}

We now turn to a larger class of graphs that are bow-free and acyclic but not necessarily ancestral.
We begin by presenting results that will be used to motivate the
discovery algorithm presented in Section~\ref{sec:algo}. Throughout,
we will consider higher order moments as a proxy for independence;
i.e., for random variables $X$ and $Z$ with $\E(X) = \E(Z) =0$, we
will use $\E(X^{K-1} Z) = 0$ as a stand-in for $X \independent Z$ and
$\E(X^{K-1} Z) \neq 0$ as a stand-in for $X \not \independent Z$. For
fixed $K > 2$, the vanishing of moments and independence are
equivalent for random variables derived using generic model parameters,
which are the matrix $B$ and the moments of $\varepsilon$. Crucially,
the moments we consider will be polynomials of the model parameters
which will allow us to leverage basic algebraic results to show
identifiability. In particular, as done above, we make statements which hold for \emph{generic} parameters. By generic, we mean that the parameters for which the statements do not hold are a null set with respect to Lebesgue measure. In the sequel, we will fix $K >2$, and consider moments of the form $\E(X^{K-1} Z)$ which depend only on the moments of $\varepsilon$ up to degree $K$; i.e., $\mathbb{E}\left(\prod_{v \in V} \varepsilon_{v}^{r_v}\right)$ for all $r \in \mathbb{Z}_{\geq 0}^{p}$ with $\|r\|_1 \leq K$. Thus, when we refer to generic error moments, we mean generic values for the error moments up to degree $K$. 
Most proofs in this section are deferred to the appendix.

\subsection{Setup}\label{sec:setup}
We first define and review some additional notions that will be helpful for the stated results. Consider the BAP $G=(V,E_\to,E_\leftrightarrow)$ and corresponding parameter $B = (\beta_{u,v})_{u,v \in V}$. For a directed path $l = v_1 \rightarrow \ldots \rightarrow v_s$, define the \emph{pathweight} of $l$ as $W(l) \coloneqq \prod_{j = 1}^{s-1} \beta_{v_{j+1}, v_j}$.
For a set $A \subset V$, let the \emph{marginal direct effects} be the
direct effects between $u, v \in A\subset V$ in the sub-model obtained
by marginalizing out all variables in $A^c \coloneqq V \setminus
A$. For convenience, let $\Lambda = I- B$.  Then the marginal direct
effects for $A$ are 
\begin{equation}\label{eq:margDirectEffect}
\begin{aligned}
\tilde B(A) = I - \left[\left(\Lambda^{-1}\right)_{A, A}\right]^{-1} &= \left[\left(\Lambda_{A,A} - \Lambda_{A, A^{c}} (\Lambda_{A^{c},A^{c}})^{-1} \Lambda_{A^{c}, A}\right)^{-1} \right]^{-1}\\
& = I - \Lambda_{A,A} - \Lambda_{A, A^{c}} (\Lambda_{A^{c},A^{c}})^{-1} \Lambda_{A^{c}, A}.
\end{aligned}
\end{equation}
For $u,v \in A$ with $u \neq v$, $\tilde B(A)_{v,u} = \beta_{v,u} + \sum_{s \in A^{c}} \beta_{v,s} \sum_{t \in A^{c}} \pi'_{s,t} \beta_{t,u}$
where
\[
  \pi'_{s,t} = ((\Lambda_{A^{c},A^{c}})^{-1})_{s,t}
\]
is the \emph{total effect} of $t$ on $s$ in the sub-graph of $G$
induced by $A^{c}$. Graphically, this total effect is the sum of
pathweights of all paths from $t$ to $s$ which only use nodes in
$A^{c}$. Observe that, for $u,v \in A$, $\tilde B(A)_{v,u} \neq 0$ only if $u \in \an(v)$.

Let $\mathcal{L}_{v,u}$ be the set of all directed paths from $u$ to
$v$ in $E_\rightarrow$. Given a set $C$ with $u \in C$, we can
partition $\mathcal{L}_{v,u}$ into disjoint sets as $
\mathcal{L}_{v,u} = \bigcup_{c \in C} \mathcal{L}^{(c)}_{v,u}(C)$, where $\mathcal{L}^{(c)}_{v,u}(C)$ is the subset of paths in $\mathcal{L}_{v,u}$ such that $c$ is the last node in $C$ to appear on the path.  Thus, $\mathcal{L}^{(u)}_{v,u}(C)$ is the set of paths from $u$ to $v$ which do not pass through any other node in $C$. This implies that 
\begin{equation}
\tilde B(C\cup \{v\})_{v,u} = \sum_{l \in \mathcal{L}^{(u)}_{v,u}(C)} W(l).
\end{equation} 


Let $D \in \mathbb{R}^{p\times p}$ be some estimate of the direct effects $B$; the support of $D$ and $B$ may differ. Let $E^{(D)}_\rightarrow$ be the set of directed edges defined by the support of $D$. Define the \emph{pseudo-parents} of $v$ given $D$, $\pspa(v)$, to be the set of parents of $v$ in $E^{(D)}_\rightarrow$ and define the \emph{pseudo-ancestors} of $v$ given $D$, $\psan(v)$, to be the ancestors of $v$ in $E^{(D)}_\rightarrow$ and $\psAn(v) = \psan(v) \cup \{v\}$. Similar to previous notation, when an argument is a set we mean the union of the function applied to each element; i.e., for the set $C$, $\pspa(C) = \bigcup_{c \in C}\pspa(c)$. 

Typically we will only consider matrices $D$ such that that $\pspa(v)
\subseteq \an(v)$; i.e, $D_{vu} \neq 0$ only if $u \in
\an(v)$. However, it will sometimes be useful to place an additional
restriction on $D$. Consider a family of sets $\mathcal{C} =
(C_v)_{v\in V}$ where $C_v \subseteq V \setminus \{v\}$. Such a family defines the
matrix-valued function $H_\mathcal{C}$ which maps $B
\in \mathbb{R}^{p \times p}$ to $D\in\mathbb{R}^{p\times p}$ given by
\begin{equation}\label{eq:Hdef}
D_{v,u} = \begin{cases}
\tilde{B}(C_v \cup \{v\})_{v,u} & \text{ if } u \in C_v,\\
0 & \text{ else}.
\end{cases}
\end{equation} 
Thus, for any $D$ which is the output of $H_\mathcal{C}$, the $v$th
row corresponds to the marginal direct effects of the sub-model
induced by $C_v \cup \{v\}$. Each element $D_{v,u}$ is the sum of
pathweights for a (not necessarily strict) subset of the paths from
$v$ to $u$, and thus is a polynomial in the elements of $B$. The
specific paths over which the sum is taken (and thus specific form of
the polynomial) depends on $\mathcal{C}$. Finally, let $\mathcal{D}$
be the set of functions $H_\mathcal{C}$ obtained from all families
$\mathcal{C} = (C_v)_{v\in V}$ with
$C_v \subseteq V \setminus \{v\}$.

\subsection{Certifying ancestral relationships in non-ancestral graphs}
In general, we use the symbol $\gamma$ to denote residuals.  Specifically, for $c \in V$, let $\gamma_c(D)$ denote the resulting residual of variable $c$ when positing $D$ to be the matrix of direct effects; i.e.
\begin{equation}
\gamma_c(D) = Y_{c} - D_{c,V}Y_{V}.
\end{equation}
For $v \in V$, we introduce the \emph{debiased direct effect} $\delta_v(C, A, S, D)$ as a function of sets $C$ and $A$ with $C \subseteq A \subseteq V \setminus \{v\}$ and matrices $S, D \in \mathbb{R}^{p  \times p}$, where $S$ is the (possibly sample) covariance of $Y$:
\begin{equation} \label{eq:deltaDef}
\delta_v(C, A, S, D) = \left\{\left[(I-D)_{C,A} S_{A,C}\right]^{-1}(I-D)_{C,A}S_{A,v}\right\}^T.
\end{equation}

Overloading the notation slightly, 
let $\gamma_v(C,S,D)$ denote the residual when using the debiased
direct effect in~\eqref{eq:deltaDef} calculated with $C$, $S$ and $D$,
and setting  $A = \psAn(C)$; i.e.,
\begin{equation}
\gamma_v(C, S, D) = Y_{v} - \delta_v(C,\psAn(C),S,D) Y_{C}.
\end{equation}
When the arguments for $\gamma_c(D)$ and $\gamma_v(C, S, D)$ are clear from the context, we will suppress the additional notation.

\begin{lemma}\label{thm:debiasedEst}
Suppose $Y$ is generated by a linear SEM with parameters $B$ and
$\Omega$ whose supports respect the sparsity imposed by the BAP $G$. For node $v \in V$ and sets of nodes $C\subseteq A \subseteq V\setminus \{v\}$, suppose 
\begin{enumerate}[label=(\roman*)]
	\item $pa(v) \subseteq C \subseteq \an(v) \setminus \sib(v)$,
	\item $A = \An(C)$,
	\item $D_{A,A} = B_{A,A}$, and
	\item $S_{A\cup\{v\},A\cup\{v\}} = \Sigma_{A\cup\{v\},A\cup\{v\}}$.
\end{enumerate}
Then $\delta_v(C, A, S, D) = B_{v,C}$.
\end{lemma}
\begin{proof}
For any $v \in V$ and set $C$ such that $\pa(v) \subseteq C \subseteq
\an(v)\setminus \sib(v)$, it follows from \eqref{eq:sigmaDef} that
\[\left[(I-B)\Sigma(I-B)^{T}\right]_{vu} = \omega_{v,u}= 0\]
 for
all $u \in C$ because $C \cap \sib(v) = \emptyset$. Since $B_{v, V \setminus C} = 0$ and $B_{C, V \setminus \an(v)} = 0$, this yields
\[B_{v,C} \Sigma_{\pa(v), \an(v)}(I-B)_{\an(v), C} = \Sigma_{v, C}(I-B)_{\an(v), C},\]
so that
\[B_{v,C}= \Sigma_{v, \an(v)}(I-B)_{\an(v), C}\left(\Sigma_{v, \an(v)}(I-B)_{\an(v), C}\right)^{-1}.\]
\end{proof}
Lemma~\ref{thm:debiasedEst} states that given the population
covariance and direct effects between vertices which are ``causally
upstream'' of $v$, selecting appropriate sets $C$ and $A$ such that
$ \pa(v)\subseteq C \subseteq \an(v)\setminus \sib(v)$ and $A =
\An(C)$ allows recovery of the direct effect of $C$ onto $v$. Since
$\delta_v$ only involves matrix inversion and multiplication, it is a
rational function of the elements of $S$ and $D$.  The specific form of the function is determined by the sets $C$ and $A$.

We use the name debiased direct effect because $\delta_v$ can be
calculated by the following alternative procedure. First form the
errors $\varepsilon_C$ and regress $Y_v$ onto $\varepsilon_C$. This
would yield the total effect of $C$ on $v$, but since $Y_v$ contains
terms involving $\varepsilon_A$, it will be biased by dependence
between $\varepsilon_C$ and $\varepsilon_A$. Given $B_{A,A}$, however,  $\Omega_{A,A}$ can be computed from $\Sigma$. Thus, the naive regression coefficients can be debiased to give the true direct effects. The assumption that $C\cap \sib(v) = \emptyset$ ensures that we do not need to also correct for dependence between $\varepsilon_C$ and $\varepsilon_v$ which we would not be able to calculate with the given information. 

Of course, in practice, we do not a priori know the relationships between candidate sets $C$ and $v$, but the following results show we can certify whether we have selected appropriate sets $C$ and $A$. Specifically, the algorithm proposed in Section~\ref{sec:algo} will certify that $C \subseteq \an(v)\setminus \sib(v)$ by testing if
\begin{equation}\label{eq:ancestralCert}
\E\left(\gamma_c(D)^{K-1}\gamma_v(C, S, D)\right) = 0 \qquad \forall c \in C.
\end{equation} 

\begin{corollary}\label{thm:suffInd}
Suppose the conditions in Lemma~\ref{thm:debiasedEst} hold, then for
every $c \in C$, we have $\gamma_c(D) \independent \gamma_v(C,S,D)$ and $\E(\gamma_c^{K-1}(D)\gamma_v(C,S,D)) = 0$.
\end{corollary}
\begin{proof}
By assumption $D_{C,A} = B_{C,A}$ so that $\gamma_C(D) =
\varepsilon_C$. In addition, Lemma~\ref{thm:debiasedEst} implies that
$\delta(C,A,S,D) = B_{v,C}$, so $\gamma_v(C,S,D) =
\varepsilon_v$. Since we assume $C \cap \sib(v) = \emptyset$, it holds
that $\gamma_c \independent \gamma_v$ for all $c \in C$.
\end{proof}

In Algorithm~\ref{alg:bang}, we use~\eqref{eq:ancestralCert} as a certificate that $C \subseteq \an(v) \setminus \sib(v)$, and indeed
Corollary~\ref{thm:suffInd} shows that $C \subseteq \an(v) \setminus \sib(v)$ is part of a set of sufficient conditions for \eqref{eq:ancestralCert} to hold. However, it is not necessary and more care is needed to ensure that 
\eqref{eq:ancestralCert} will not mistakenly certify a set $C$ if $C
\not \subseteq \an(v)\setminus \sib(v)$. 
We first
state Lemma~\ref{thm:neccIndGeneral} which gives a necessary condition
for~\eqref{eq:ancestralCert} that will be useful in deriving subsequent results. 

\begin{figure}[t]
	\centering
	\begin{tikzpicture}[->,>=triangle 45,shorten >=1pt,
	auto,
	main node/.style={ellipse,inner sep=0pt,fill=gray!20,draw,font=\sffamily,
		minimum width = .5cm, minimum height = .5cm}]
	
	\node[main node] (1) {1};
	\node[main node] (2) [left = .75cm of 1]  {2};
	\node[main node] (3) [right = .75cm of 1]  {3};
	\node[main node] (4) [right = .75cm of 3]  {4};
	\node[main node] (5) [right = .75cm of 4]  {5};
	
	\path[color=black!20!blue,style={->}]
	(1) edge node {} (2)
	(1) edge node {} (3)
	(3) edge node {} (4)
	(4) edge node {} (5)
	;
	\end{tikzpicture}

\caption{\label{fig:pruningExample}When $D_{2,1} = \beta_{2,1}$ and $D_{5,4} = \beta_{5,4}$, naively testing \eqref{eq:ancestralCert} would mistakenly certify $2$ and $5$ as ancestors of $3$.}
\end{figure}
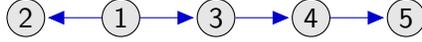

\begin{lemma}\label{thm:neccIndGeneral}
	Let $v \in V$ and $C \subseteq A \subseteq V \setminus \{v\}$. Suppose $D \in \mathbb{R}^{p \times p}$ such that $D_{s,t} \neq 0$ only if $t \in \an(s)$. Then, for generic $B$ and error moments, if $\delta_v(C, \psan(C), S, D) \neq \tilde B(C\cup\{v\})_{v,C}$,
	then $\E(\gamma_c(D)^{K-1}\gamma_v(C, S, D)) \neq 0$ for some $c \in C$.
\end{lemma}

Lemma~\ref{thm:neccIndSib}
shows that if $C \cap \sib(v) \neq \emptyset$, then there exists some $c
\in C$, such that $\E(\gamma_c^{K-1}(D)\gamma_v(C,S,D)) \neq 0$, 
Ensuring that we do not mistakenly certify non-ancestors of $v$ is a
bit more delicate because, depending on $D$, \eqref{eq:ancestralCert}
may actually hold for some set $C \not \subseteq \an(v)$ when $C \cap
\sib(v) =\emptyset$. In particular there are two cases of potential ways a set can be misscertified. First, $C$ may contain a node $c \not \in
\an(v) \cup \de(v)$. Alternatively, $C$ may contain a descendant of $v$ which
already has the effect of $v$ removed; for instance, for some $c$ such that $v \in \an(c) \setminus \pa(c)$, if $D_{c, V} = B_{c, V}$ then $\gamma_c(D)$ will not contain any term with $\varepsilon_v$. Consider the example in
Figure~\ref{fig:pruningExample}, and let $D_{21} = \beta_{21}$ and
$D_{54} = \beta_{54}$. The set $C = \{2,5\}$ would satisfy \eqref{eq:ancestralCert} for 
$v = 3$ because $2$ is neither an ancestor or descendant
of $3$ and $5$ is an ancestor of $2$, but adjusting $5$ by $4$ removes
the effect of $3$.

More generally, suppose that $C \cap \sib(v) = \emptyset$ but $C \not \subseteq \an(v)$, and let $C_1 = C \cap \an(v)$ and $C_2 = C \setminus \an(v)$. Thus, $C_1$ should rightfully be certified, but $C_2$ should not. However, if $D$ is such that $\E(\gamma_c(D)^{K-1}\gamma_v(C, S,D)) = 0$ for all $c \in C$, then $C$ (including $C_2$) could be mistakenly certified as a subset of $\an(v) \setminus \sib(v)$. Fortunately, Lemma~\ref{thm:neccIndDes} implies that instead of testing all possible sets $C \subseteq V \setminus v$, we can use an additional pre-screening procedure to filter out problematic sets, $C \not \subseteq \an(v)$ which would otherwise satisfy \eqref{eq:ancestralCert}. 

Specifically, Lemma~\ref{thm:neccIndDes} implies that
\begin{equation}
\E(\gamma_c(D)^{K-1}\gamma_v(C_1, S,D)) = 0 \qquad \forall c \in C_1.
\end{equation}
Thus, we still would have certified that $C_1 \subseteq \an(v) \setminus \sib(v)$. Furthermore, after adjusting $v$ for $C_1$, the resulting residuals of $v$---$\gamma_v(C_1, S, D)$---would also be independent of $\gamma_c$ for all $c \in C_2$; i.e.,
\begin{equation}
\E(\gamma_c(D)^{K-1}\gamma_v(C_1, S,D)) = 0 \qquad \forall c \in C_2.
\end{equation} 
Thus, we can screen out non-ancestors of $v$ which might otherwise be miscertified, by removing any $c \in C$ such that for some $C' \subseteq C \setminus \{c\}$,
\begin{equation}\label{eq:preScreen}
\E(\gamma_c(D)^{K-1}\gamma_v(C', S,D)) = 0.
\end{equation} 
This is implemented in Algorithm~\ref{alg:checkInd}. A concern is then the question whether the pre-screening
procedure implemented in Algorithm~\ref{alg:checkInd} may
mistakenly rule out a parent or sibling of $v$.  To show that this does not happen for
generic parameters, we derive
Lemma~\ref{thm:noMissPruneParents} below.

Lemma~\ref{thm:noMissPruneParents} and \ref{thm:neccIndSib} both require that $D = H_\mathcal{C}(B)$; i.e., for each $v \in V$, $D_{v,C}$ is the marginal direct effect of $C$ on $v$ where $C$ is the set of non-zero entries in $D_{v, V}$. In Algorithm~\ref{alg:certPar}, we only update $D_{v,C}$ to $\delta_v(C, A, S, D)$ when $\E(\gamma_c^{K-1}(D)\gamma_v(C, S, D) = 0$ for all $c\in C$. As shown by Lemma~\ref{thm:neccIndGeneral}, this implies that $\delta_v(C, \psAn(C), S, D)$ is the marginal direct effect of $C$ on $v$ so that the updated $D = H_\mathcal{C}(B)$ for some $H_\mathcal{C} \in \mathcal{D}$.



\begin{lemma}\label{thm:neccIndDes}
Consider $v \in V$ and $C \subseteq V\setminus \{v\}$. Let $D \in
\mathbb{R}^{p \times p}$ such that $D_{s,t} \neq 0$ only if $t \in
\an(s)$.
Suppose $C \not \subseteq \an(v)$, but that
$\E(\gamma_c(D)^{K-1}\gamma_v(C, S,D)) = 0$
for all $c \in C$. Then for generic $B$ and error moments, $C_1 = C\cap [\an(v)\setminus \sib(v)]$,
\begin{equation*} \E(\gamma_c(D)^{K-1}\gamma_v(C_1, S,D)) = 0
  \qquad\forall c \in C.
  \end{equation*}
\end{lemma}

\begin{lemma}\label{thm:noMissPruneParents}
Suppose $D = H_\mathcal{C}(B)$ for some $H_\mathcal{C} \in
\mathcal{D}$ with $\mathcal{C}=(C_s)_{s\in V}$ such that $C_s \subseteq \an(s) \setminus \sib(s)$.  Let $v \in V$ be such that we have 
$
\E(\gamma_c(D)^{K-1}\gamma_v(\pspa(v), S,D)) = 0
$ for all $c \in \pspa(v)$.
If $q \in \{\pa(v)\setminus \pspa(v)\} \cup \sib(v)$, then for generic $B$ and error moments,  $\E \left(\gamma_q(D)^{K-1}\gamma_v(D)\right) \neq 0$.
\end{lemma}

\begin{lemma}\label{thm:neccIndSib}
	Consider $v \in V$ and sets $A, C$ such that $C \subseteq A \subseteq V \setminus \{v\}$. Suppose $D = H_\mathcal{C}(B)$ for some $H_\mathcal{C} \in \mathcal{D}$ with $\mathcal{C}=(C_s)_{s\in V}$ such that $C_s \subseteq \an(s) \setminus \sib(s)$ for all $v \in V$. If $C \cap \sib(v) \neq \emptyset$, then for generic $B$ and error moments, there exists some $q \in C$ such that $\E\left(\gamma_q(D)^{K-1}\gamma_v(C, \Sigma, D )\right) \neq 0$.
\end{lemma}

Thus far we have been concerned with discovering sets which contain ancestors but not siblings of some node $v$. Corollary~\ref{thm:noMissPruneParentsFinal} shows that when we have identified such a set which is also a superset of the parents of $v$, we can prune away ancestors which are not parents. This motivates the pruning procedure described in Algorithm~\ref{alg:pruneAncest}.

\begin{corollary}\label{thm:noMissPruneParentsFinal}
	Suppose $D = B$. For $v \in V$ and generic $B$ and error moments, suppose $\pa(v) \subseteq C \subseteq \an(v)\setminus \sib(v)$. 
	If $q \in C \setminus \pa(v)$, then for all $c \in C$
	\begin{equation}
	\E(\gamma_c(D)^{K-1}\gamma_v(C \setminus \{q\}, \Sigma,D)) = 0.
	\end{equation}
	If $q \in \pa(v)$, then there exists some $c \in C$ such that
	\begin{equation}
	    \E(\gamma_c(D)^{K-1}\gamma_v(C \setminus \{q\}, \Sigma,D)) \neq 0.
	\end{equation}
\end{corollary}

\section{Graph estimation algorithm}\label{sec:algo}
Using the claims established above, we present the \textbf{B}ow-free \textbf{A}cyclic \textbf{n}on-\textbf{G}aussian (BANG) procedure in Algorithm~\ref{alg:bang} which completely identifies the underlying causal structural of the linear SEM when it corresponds to a BAP.

The algorithm starts with a complete bidirected graph so that the posited siblings for each node, $\widehat{\sib}(v)$, are initialized to $V \setminus v$ and the posited parents $\widehat{\pa}(v)$, are initialized to $ \emptyset$. The method then iteratively certifies ancestors which are not siblings by considering whether \eqref{eq:ancestralCert} holds for progressively larger sets. When we certify that $C \subseteq \an(v) \setminus \sib(v)$, $C$ is added to $\widehat{\pa}(v)$, $C$ is removed from $\widehat{\sib}(v)$, and $D$ is updated. This procedure is repeated until no additional ancestral relationships can be certified. Any remaining dependency between the residuals are then assumed to be due to a bidirected edge. In the algorithm, whenever we specify a test for $X \independent W$, we mean testing $\E(X^{K-1}W) = 0$ for some prespecified $K > 2$.

\begin{algorithm}[b]\scriptsize
	\caption{\label{alg:bang}BANG procedure}
	\begin{algorithmic}[1]
		\State Input: Data $Y \in \mathbb{R}^{p \times n}$ and $S \in \mathbb{R}^{p\times p}$ which is the (potentially sample) covariance of $Y$
		\State For all $v \in V$, set $\widehat{\pa}(v) = \emptyset$ and $\widehat{\sib}(v) = V \setminus \{v\}$
		\State Set all elements of $D \in \mathbb{R}^{p \times p}$ to be $0$ and $l = 1 $

		\While{$\max_v |\widehat{\sib}(v)| \geq l$}
		\For{$v \in V$}
		\State Prune $\widehat{\sib}(v)$ using Alg.~\ref{alg:checkInd}
		\State Certify pseudo-parents of $v$ and update $\widehat{ \pa}(v)$, $\widehat{\sib}(v)$, and $D$ using Alg.~\ref{alg:certPar}
		\EndFor
		\State \textbf{if} $D$ was updated, reset $l = 1$; \textbf{ else } set $l = l + 1$
		\EndWhile\label{alg:ancestorFind}
		\State Remove ancestors which are not parents from $\widehat{ \pa}(v)$ for all $v \in V$ using Alg.~\ref{alg:pruneAncest}

		\State Return: $\hat E_\rightarrow = \{(u,v): u \in \widehat{\pa}(v)\}$, $\hat E_\leftrightarrow = \{\{u,v\}: u \in \widehat{\sib}(v)\}$ 
	\end{algorithmic}
\end{algorithm}

\begin{algorithm}[t]\scriptsize
	\caption{\label{alg:checkInd}Prune $\widehat{\sib}(v)$}
	\begin{algorithmic}[1]
		\State Input: $v$, $\widehat{\sib}(v)$, $Y$, $D$
		\State Set $\gamma(D) = Y - DY$
		\For{$u \in \widehat{\sib}(v)$}
			\If{$\gamma_u(D) \independent \gamma_v(D)$}
		\State Remove $u$ from $\widehat{\sib}(v)$ and remove $v$ from  $\widehat{\sib}(u)$
	\EndIf
	\EndFor
		\State Return: $\widehat{\sib}(v)$ for all $v \in V$, 
	\end{algorithmic}
\end{algorithm}

\begin{algorithm}[t]\scriptsize
	\caption{\label{alg:certPar}Certify pseudo-parents}
	\begin{algorithmic}[1]
		\State Input: $v$, $\widehat{ \pa}(v)$, $\widehat{\sib}(v)$, $D$, $S$, $Y$, $l$
		\State Set $C^\star = \emptyset$
		\For {$C \in  \binom{\widehat{\sib}(v)}{l}$}
		\If {$\gamma_C(D) \independent \gamma_v(C \cup \widehat{\pa}(v), S, D)$}
		\State $C^\star = C^\star \cup C$
		\EndIf
		\EndFor
		\State $\widehat{\pa}(v) = \widehat{\pa}(v) \cup C^\star$
		\State $D_{v,\widehat{\pa}(v)} = \delta_v(\widehat{\pa}(v), S, D)$
		\State $\widehat{\sib}(v) = \widehat{\sib}(v) \setminus \widehat{\pa}(v)$
		\State $\widehat{\sib}(s) = \widehat{\sib}(s) \setminus \{v\} \qquad  \forall s \in \widehat{ \pa}(v)$
		\State Return: $D$, $\widehat{\sib}(v)$ and $\widehat{ \pa}(v)$ for all $v \in V$, 
		\end{algorithmic}
	\end{algorithm}
	
\begin{algorithm}[t]\scriptsize
	\caption{\label{alg:pruneAncest}Prune ancestors which are not parents}
	\begin{algorithmic}[1]
		\State Input: $\widehat{ \pa}(v)$ and $\widehat{\sib}(v)$ for all $v \in V$, $D$, $S$, $Y$, $l$
\State Form topological ordering $\sigma$ such that $\sigma(u) < \sigma(v)$ implies $v \not \in \psan(u)$  
\For{$v \in \sigma^{-1}([p])$}
\For{$s\in \widehat{\pa}(v)$}
\If{$\gamma_c \independent \gamma_v(\widehat{\pa(v)}\setminus \{s\}, S,D )$ for all $c \in \hat{\pa}(v)$}
\State $\widehat{\pa}(v) = \widehat{\pa}(v) \setminus \{s\}$
\State $D_{v, s} = 0$
\EndIf
\EndFor
\EndFor
	\end{algorithmic}
\end{algorithm}

\subsection{Graph identification}
We first show that when given population values, the BANG procedure will return the correct graph.
\begin{theorem}\label{thm:mainID}
Suppose $Y$ is generated by a linear SEM which corresponds to a BAP $G$. Then for generic choices of $B$ and error moments, Algorithm~\ref{alg:bang} will output $\hat G = G$ when given population moments of $Y$.
\end{theorem}
\begin{proof}
The lemmas in Section~\ref{sec:bap} make statements about different individual quantities being non-zero for generic $B$ and error moments.  Since we will only consider a finite set of these quantities, the union of the null sets to be avoided for each individual quantity is also a null set. Thus, in this proof, we may assume that quantities that are generically non-zero are all actually non-zero.  

Our proof proceeds by induction and shows that after Line \ref{alg:ancestorFind} of Alg.~\ref{alg:bang}, the procedure obtains $\widehat{\pa}(v)$ such that $\pa(v) \subseteq \widehat{\pa}(v) \subseteq \an(v)\setminus \sib(v)$ and $\widehat{\sib}(v) = \sib(v)$. Then the final step, using Alg.~\ref{alg:pruneAncest}, results in $\widehat{\pa}(v) = \pa(v)$. 

Let $\sigma$ be a topological ordering consistent with the directed portion of underlying graph $G$. Let the $z$th induction step be defined as an entire step of testing progressively larger sets $C$ until all parents of $v = \sigma^{-1}(z)$ have been discovered. Note that since we do not know the ordering a priori and simply cycle over all variables and progressively larger sets, it could be that $z$th induction step is actually completed (i.e., all the parents of $\sigma^{-1}(z)$ are discovered) chronologically before the $z-1$ step is done.   

As the induction hypothesis for step $z$, let $v = \sigma^{-1}(z)$ and suppose that:
\begin{enumerate}
	\item\label{itm:induction1} For $A = \sigma^{-1}([z-1])$, $D_{A, A} = B_{A,A}$ and $\widehat{\pa}(a) \supseteq \pa(a)\; \forall a \in A$; 
	\item\label{itm:induction2} $D = H_\mathcal{C}(B)$ for some $H_\mathcal{C} \in \mathcal{D}$ where $\forall s \in V$, $\widehat{\pa}(s)\subseteq \an(s)\setminus\sib(s)$ (i.e., $C_s \subseteq \an(s) \setminus \sib(s)$);
	\item\label{itm:induction3} For all $u \in V$,
          $\widehat{\sib}(u) \supseteq \sib(u)$ and $\pa(u) \subseteq
          \{ \widehat{\sib}(u)\cup \widehat{pa}(u)\}$.
\end{enumerate}
The first condition assumes all directed edges upstream of $v$ have
been identified. The second condition assumes that each row in the current value of $D$ corresponds to a marginal direct effect and that nothing has been misscertified into $\widehat{\pa}(v)$. The third condition assumes no siblings or parents have been incorrectly pruned. We now show that the $z$th step will be completed and the induction conditions will hold for step $z+1$. 

\textbf{Condition 3:}
By assumption, $D_{A, A} = B_{A,A}$ and $D = H_\mathcal{C}(B)$ for some $H_\mathcal{C} \in \mathcal{D}$, so $\psAn(\pa(v)) = \An(\pa(v))$. Lemma~\ref{thm:noMissPruneParents} implies that for all $u \in V$, Alg.~\ref{alg:checkInd} does not mistakenly remove any siblings or parents of $u$ from $\widehat{\sib}(v)$. Furthermore, Lemma \ref{thm:neccIndSib} implies that Alg.~\ref{alg:certPar} will not remove any siblings from $\widehat{\sib}(v)$. Thus $\pa(u) \subseteq \big(\widehat{\sib}(u) \cup \widehat{ \pa}(u)\big)$ and $\widehat{\sib}(u) \supseteq \sib(u)$ continue to hold, and Condition~\ref{itm:induction3} is satisfied for the next step.

\textbf{Condition 2:}
Alg.~\ref{alg:certPar} only adds $C$ to $\widehat{\pa}(v)$ if $C \cup \widehat{\pa}(v)$ satisfies \eqref{eq:ancestralCert}. Lemma~\ref{thm:neccIndSib} implies that any set $C$ such that $C \cap \sib(v) \neq \emptyset$ will not be added. 
We now show that any set $C \not\subseteq \an(v)$ will not be added to $\widehat{\pa}(v)$ because it either will not be considered by Alg.~\ref{alg:certPar} or will not satisfy \eqref{eq:ancestralCert}. 

For some $v \in V$, let $C  \subseteq V \setminus \widehat{\pa}(v)$ and $C_1 = C \cap \an(v)$. If $C \not \subseteq \an(v)$ and $C_1 \neq C$, the set $C \cup \widehat{\pa}(v)$ will only be considered by Alg.~\ref{alg:certPar} after the set $C_1 \cup \widehat{\pa}(v)$. If $C \cup \widehat{\pa}(v)$ satisfies \eqref{eq:ancestralCert}, then Lemma~\ref{thm:neccIndDes} implies that $C_1 \cup \widehat{\pa}(v) $ also satisfies \eqref{eq:ancestralCert}. Thus, $C_1$ would first be certified into $\widehat{\pa}(v)$. Lemma~\ref{thm:neccIndDes} further implies that for any $c \in C$, $\E(\gamma_c(D)^{K-1} \gamma_v(C_1 \cup \widehat{\pa}(v), S, D) = 0$, 
so that after  $C_1$ is placed in $\widehat{\pa}(v)$ and $D$ is
updated, we have $\E(\gamma_c(D)^{K-1} \gamma_v(D)) = 0$ for all  $c \in C \setminus C_1$. Hence, Alg.~\ref{alg:checkInd} will subsequently remove $C \setminus C_1$ from $\widehat{\sib}(v)$. Thus, $C$ will only be considered if $C \subseteq \an(v)$ or if $C \cup \widehat{\pa}(v)$ does not satisfy~\eqref{eq:ancestralCert}, and any updates preserve $\widehat{\pa}(v) \subseteq \an(v) \setminus \sib(v)$. Lemma~\ref{thm:neccIndGeneral} also implies that after updates are made to $\widehat{\pa}(v)$, the resulting update to the $v$th row of $D$ is a marginal direct effect so that $D = H_\mathcal{C}(B)$ for some $H_\mathcal{C} \in \mathcal{D}$. Thus, Condition~\ref{itm:induction2} continues to hold. 



\textbf{Condition 1:}
By the acyclicity assumption, $|\pa(v)| \leq  z - 1$. So by successively testing larger sets, and resetting the counter after each update, if we do not first certify all parents of $v$ as part of smaller sets, we will eventually consider $C = \pa(v)$. The induction hypothesis and Lemma~\ref{thm:suffInd} ensure that $\E(\gamma_c^{K-1}(D) \gamma_v(C, S, D)) = 0$ for all $c \in C$ so that $C$ will be certified into $\widehat{\pa}(v)$. Lemma~\ref{thm:neccIndGeneral} implies the resulting update which sets $D_{v, \widehat{\pa}(v)} = \delta_v(\widehat{\pa}(v),\psAn(\widehat{\pa}(v)), S, D)$ will result in $D_{v,V} = B_{v,V}$. Thus, step $z$ will be completed, and Condition~\ref{itm:induction1} continues to hold.

After $p$ steps, $D = B$, so $\gamma_v(D) = \varepsilon_{v}$ for all $v$ and $\E(\gamma_u(D)^{K-1}\gamma_v(D)) \neq 0$ if and only if $u \in \sib(v)$. If $\widehat{\sib}(v) \neq \emptyset$ for any $v$ then after the last update to $D$, there will be at least one more pass through Alg.~\ref{alg:checkInd} so any non-siblings will be removed and $\widehat{\sib}(v) = \sib(v)$ for all $v \in V$. If $\widehat{\sib}(v) = \emptyset$ for all $v$, then by the induction conditions, $\sib(v) \subseteq \widehat{\sib}(v) = \emptyset$, so again,  $\sib(v) = \widehat{\sib}(v)$.

By Condition~\ref{itm:induction2}, $\pa(v) \subseteq \widehat{\pa}(v) \subseteq \an(v) \setminus \sib(v)$, and Corollary~\ref{thm:noMissPruneParentsFinal} implies that Alg.~\ref{alg:pruneAncest} removes any ancestors from $\widehat{\pa}(v)$ which are not parents but does not remove any parents. Thus, $\widehat{\pa}(v) = \pa(v)$.
\end{proof}

Theorem~\ref{thm:mainID} shows that the graph is correctly identified
given population values by successively testing whether a quantity is
zero or non-zero. However, the quantities considered are non-linear
functions of the data so in finite samples, in addition to sample
variability, the sample quantities will typically be
biased. Nonetheless, the following corollary shows that there exists a
cut-off $\eta_1 > 0$ such that checking whether each sample statistic
is greater than or less than $\eta_1/2$ as a proxy for independence
will yield consistent estimates of $G$ as long as the sample moments
of $Y$ are consistent for the population moments. The value of
$\eta_1$ depends on the model parameters, but some $\eta_1 > 0$ must
exist for generic $B$ and error moments. This implies pointwise consistency of
BANG when the tests are ``appropriately'' tuned. Of course $\eta_1$
depends on quantities that are unknown in practice, so in applications we find ourselves in a similar position as for
other existing constraint-based algorithms in causal discovery (e.g.,
PC or FCI algorithm) where algorithm output delicately depends on
suitable specification of levels for statistical tests which act as
tuning parameters.

\begin{corollary}\label{thm:consistentEst}
	Suppose $Y$ is a sample comprised of i.i.d.~vectors $Y^{(1)}, \dots, Y^{(n)}$ generated by a linear SEM that corresponds to the BAP $G$. Then, for generic choices of $B$ and error moments, there exist $\eta_1, \eta_2 > 0$ such that when the sample moments are within an $\eta_2$-ball of the population moments of $Y$, Alg.~\ref{alg:bang} will output $\hat G = G$ when comparing the absolute value of the sample statistics to $\eta_1/2$ as a proxy for the independence tests.
\end{corollary}
\begin{proof}
In Theorem~\ref{thm:mainID}, we showed that BANG will correctly identify the true BAP as long as certain  expectations encoding absence of edges and paths are all $0$ and  further  expectations encoding presence of edges/paths are all non-zero (which holds for generic $B$ and error moments). For a fixed BAP $G$, let $S_0$ be the set of expectations which should be $0$, and let $S_1$ be the set of expectations which should be non-zero. Let $\eta_1 = \min_{S_1} |\E(\gamma_c^{k-1}\gamma_v)|$. For generic parameters, $\eta_1 > 0$.

We note that when $D = H_\mathcal{C}(B)$ for some $H_\mathcal{C} \in \mathcal{D}$, the maps which take moments of $Y$ to $E(\gamma_c(D)^{K-1}\gamma_v(C,\Sigma,D))$ are rational functions and are thus Lipschitz within a sufficiently small ball around the population moments of $Y$. Thus, there must exist some $\eta_2 > 0$ such that when the sample moments are within $\eta_2$ of the population moments, the sample quantities in $S_0$ and $S_1$ are with in $\eta_1/2$ of the population quantities. 

This implies that all estimates corresponding to quantities which are 0 are less than $\eta_1/2$ in absolute value, and all estimates that correspond to quantities which are generically non-zero are greater than $\eta_1/2$ in absolute value. Thus, comparing the absolute value of the sample quantities to $\eta_1/2$ accurately determines whether the parameters belong to $S_0$ or $S_1$ and thus yields a correct estimate $\hat G$. 
\end{proof}

\subsection{Practical concerns}
For any BAP, identification with population values holds for all but a
null set of $B$ and error moments. This set includes any of parameters
where the direct marginal effect of $u$ on $v$ might vanish for some
$u \in pa(v)$. Existing results on distributional equivalence of BAPS
with Gaussian errors \citep{nowzohour2017greedy} imply that error
moments which correspond to some Gaussian distribution also must be
avoided. For finite samples, accurate estimation would require an
analogue to the strong faithfulness conditions studied in
\citet{uhler2013geometry}. However, the typical Gaussian strong
faithfulness condition only regards the linear coefficients and error
covariances, while we additionally require that the higher order
moments of the errors are sufficiently non-Gaussian. As we see in the
simulations, when the errors come from a multivariate $T$ distribution
which are not too different from a Gaussian, performance may suffer considerably. 

Throughout the proof, we examine high order moments as a proxy for independence. Since these quantities are polynomials of the parameters, it allows us to make algebraic arguments that
facilitate the analysis. However, in practice, one could use any
non-parametric independence test instead  \citep{gretton2005measuring,
  szekely2009dist, bergsma2014consistent, pfister2018independence}. In
Section~\ref{sec:numericalResults}, we use
dHSIC~\citep{pfister2018independence} which performs well when the
sample size is small. However, simply calculating the statistic
requires $O(n^2)$ time rendering the permutation or bootstrap procedures required for calibrating a null distribution prohibitively expensive; we thus use the ``gamma approximation'' to the null distribution. However, even this becomes infeasible for $n > 2000$ and $p = 6$.  

When the sample size is large, we choose an implementation which is tied to the theoretical analysis and test whether moments are zero or non-zero. Specifically, we use empirical likelihood to test the joint hypothesis that $\E(\gamma_{c}^{K-1} \gamma_v) = 0$ for all $c \in C$. Empirical likelihood is useful as it does not require explicit estimation of the variances of $\gamma_{c}^{K-1} \gamma_v$ in order to form a well-calibrated test statistic, and the empirical likelihood ratio statistic converges to a known reference distribution under mild conditions. In addition, pooling together all the tests into one omnibus test helps mitigate multiple testing.  The empirical likelihood approach is typically less powerful than dHSIC at detecting dependence; however, the computation time required for the test is an order of magnitude smaller. When testing whether moments are zero or non-zero, a specific value of $K$ must be selected. This should correspond to a moment of the errors which is not consistent with the Gaussian distribution. If the data is skewed, $K=3$ could suffice since the third moment of the Gaussian is zero, but the third moment of the data is non-zero. If the data is either heavy or light tailed (relative to the Gaussian), $K=4$ should suffice. One can also combine the results of multiple values and test $\E(\gamma_c^{k-1}\gamma_v) = 0 $ for all $k =\{3, \ldots, K\}$ for some arbitrarily large $K$, but in practice, using larger values of $K$ requires more samples for accurate estimation and testing.

Given an oracle independence test, if the in-degree of each node
(counting both directed and bidirected edges) is bounded by some
constant $J$, then the total number of tests required is bounded by a
polynomial of the number of variables. Again, let $\sigma$ be a
topological ordering of the nodes consistent with $G$. As shown in the
proof of Theorem~\ref{thm:mainID}, at step $z$, once all the ancestors
of $\sigma^{-1}([z-1])$ have been identified, then we need only test
sets $C$ up to size $|\pa(\sigma^{-1}(z))| \leq J$ to certify the
parents of $\sigma^{-1}(z)$ and subsequently update $D$. Thus, $l$,
the size of sets considered, will never exceed $J$. In between any
update to $D$, for each node there will no more than $\sum_{k =
  1}^J\binom{p}{k}\leq p^{J}$ sets considered. In addition, each time
$l$ is incremented, for each node we screen no more than $p-1$
potential siblings using Alg.~\ref{alg:checkInd}. By the acyclicity
assumption, there are at most $p(p-1)/2 < p^2$ ancestral relationships
to discover, which would cause an update to $D$. Thus, to fully
discover $B$, there will be no more than $p^2 \times p(p^{J} + J p)$
independence tests. Once $D$ is fully updated so that $D = B$, then
$\widehat{\sib}(v) = \sib(v)$ for all $v \in V$. So that for each $v
\in V$ there will be at most an additional cycle through all sets with
size less than $J$ which will again result in $p(p^J + Jp)$ additional
tests.  We conclude that Alg.~\ref{alg:pruneAncest} will check at most $p(p-1)/2$ discovered ancestral relationships. Thus, there must be less than $O(p^{J+3}) $ total independence tests.

Finally, we note that after a graph has been estimated, the resulting $D$ and empirical covariance of $\gamma$ could be used as estimates for $B$ and $\Omega$. Alternatively, the empirical likelihood procedure of~\citet{wang2017empirical} could be used for both point estimates and simultaneous confidence intervals.

\section{Discovery in the presence of bows}\label{sec:modelMiss}
In Section~\ref{sec:algo}, we show that BANG recovers the true graph when the data is generated by a linear SEM corresponding to a BAP.  It is unclear, however, how to test this assumption directly in practice.  Hence, it is interesting to study what BANG will output if the true data generating process is a linear SEM which corresponds to an acyclic mixed graph $G = \left(V, E_\rightarrow, E_\leftrightarrow\right)$ which is not bow-free. 

In this section, we show that, given population values, BANG will
return a BAP $\bar G = \left( V, \bar E_\rightarrow, \bar
  E_\leftrightarrow \right)$ which we subsequently define. We use
$\bar \pa(v)$ and $\bar \sib(v)$ to denote the parents and siblings of
node $v$ in $\bar G$. Although $\bar G$ can be quite different from
$G$,  certain key properties important for interpretation are preserved. In particular, $\bar \pa(v) \subseteq \an(v)$, so that any directed edge in the output is not in the opposite direction of the true effect. 
Thus, roughly speaking, in the presence of bows, the procedure is sound---though potentially not complete---for identifying ancestral relations. In addition, any member of $\bar \sib(v)$ must be connected to $v$ by a bi-directed path in $G$, so in the presence of bows, the procedure is complete---though potentially not sound---for identifying confounded relationships. However, roughly speaking, a bidirected edge indicates less certainty about a causal relationship than a directed or absent edge, so the BANG procedure can be considered ``conservative'' in the sense of not being overconfident when positing causal relationships.      
In Example~\ref{ex:modelMisspecify}, we show a graph with bows and its corresponding ``projection'' $\bar G$.


For $v \in V$, we recursively define a set of nodes which we will call the \emph{irremovable} nodes. We denote this set by $\irr(v)$. Roughly speaking, $\irr(v)$ contains all nodes, whose total effect will never be fully removed from $v$ by the BANG procedure. 

\begin{definition}
Let $v\in V$, and define $\Irr(v)_0=\{v\}$.  For $k=1,\dots,p$, we define recursively $\Irr(v)_{k}= \left[\pa(\Irr(v)_{k-1})\cap\sib(\Irr(v)_{k-1}) \right] \cup \Irr(v)_{k-1}$. Then the set of irremovable nodes is defined as $\irr(v)= \Irr(v)_p$.
\end{definition}


Every $w \in \irr(v)$ is connected to $v$ by both a directed path
which only passes through $\irr(v)$ as well as a path of bidirected
edges which also only passes through nodes which are in
$\irr(v)$. 
In addition, $\irr(v) \subseteq \sib(\irr(v))$ and $\Irr(v)_1$ contains all nodes which form a bow which ends at $v$. Given $\irr(v)$ we now define the bidirected edges in $\bar G$.

\begin{definition}\label{def:barSib}
 The siblings of a node $v\in V$ in $\bar G$ are
\begin{equation} \label{eq:sibProjection}
\bar{\sib}(v) =  \{u \, : \, \sib(\irr(v)) \cap  \irr(u) \neq \emptyset\} \setminus \{v\}.
\end{equation}
In other words, the elements of $\bar{\sib}(v)$ are all nodes which are irremovable from $v$, the siblings of those nodes, and any other nodes whose irremovable nodes have a sibling in $\irr(v)$.
\end{definition}

\begin{definition}\label{def:barParent}
The parents of a node $v \in V$ in $\bar G$ are
\begin{equation}\label{eq:paProjection}
\bar{\pa}(v) =   \pa(\irr(v)) \setminus \bar \sib(v).
\end{equation}
Note that by definition $\bar{\pa}(v)\cap\bar \sib(v)=\emptyset$,
which prevents bows.
\end{definition}

Since every $u \in \irr(v)$ has a directed path to $v$ which only passes through $\irr(v)$, Definition~\ref{def:barParent} implies that for every $u \in \bar \pa(v)$, there exists a path $l = u \rightarrow s_1 \rightarrow \ldots \rightarrow s_{|l| - 2} \rightarrow v$ such that $\{s_1, \ldots, s_{|l|-2} \} \subseteq \irr(v)$.
By construction, every path from $w \in V$ to $v$ either is entirely contained in $\irr(v)$ or passes through $\bar \pa(v)$.

\begin{definition}\label{def:barG}
Let $G = \{V, E_\rightarrow, E_\leftrightarrow\}$ be an acyclic
directed mixed graph.  Let $\bar E_\rightarrow$ and $\bar
E_\leftrightarrow$ be given by Definitions~\ref{def:barSib} and \ref{def:barParent}.  We term $\bar G = \{V, \bar E_\rightarrow, \bar
E_\leftrightarrow\}$  the BAP projection of $G$.
\end{definition}

\begin{example}\label{ex:modelMisspecify}
\begin{figure}
    	\begin{subfigure}[t]{.45\textwidth}\centering
\begin{tikzpicture}[->,>=triangle 45,shorten >=1pt,
		auto,
		main node/.style={ellipse,inner sep=0pt,fill=gray!20,draw,font=\sffamily,
			minimum width = .5cm, minimum height = .5cm}]

\node[main node] (1) {1};
\node[main node] (2) [below = .9cm of 1] {2};
\node[main node] (3) [right = .9cm of 1]  {3};
\node[main node] (4) [right = .9cm of 2]{4};
\node[main node] (5) [right = .9cm of 3] {5};
\node[main node] (6) [right = .9cm of 4] {6};

\path[color=black!20!blue,style={->}]
(1) edge node {} (3)
(2) edge node {} (4)
(3) edge node {} (5)
(4) edge node {} (5)
(5) edge node {} (6)
(3) edge node {} (6)
;

\path[color=black!20!red,style={<->}]
(5.north) edge[bend right = 25] node {} (3.north)
(5.north) edge[bend right = 45] node {} (1.north)
(3) edge node {} (4)
;
\end{tikzpicture}
\caption{\label{fig:trueMod}True model, $G$.}
\end{subfigure}
~
\begin{subfigure}[t]{.45\textwidth}\centering
\begin{tikzpicture}[->,>=triangle 45,shorten >=1pt,
		auto,
		main node/.style={ellipse,inner sep=0pt,fill=gray!20,draw,font=\sffamily,
			minimum width = .5cm, minimum height = .5cm}]

\node[main node] (1) {1};
\node[main node] (2) [below =  .9 cm of 1] {2};
\node[main node] (3) [right = .9cm of 1]  {3};
\node[main node] (4) [right = .9cm of 2]{4};
\node[main node] (5) [right = .9cm of 3] {5};
\node[main node] (6) [right = .9cm of 4] {6};

\path[color=black!20!blue,style={->}]
(1) edge node {} (3)
(2) edge node {} (4)
(2) edge[densely dashed] node {} (5)
(5) edge node {} (6)
(3) edge node {} (6)
;

\path[color=black!20!red,style={<->}]
(5) edge[densely dashed] node {} (3)
(5.north) edge[bend right = 45] node {} (1.north)
(3) edge node {} (4)
(4) edge[densely dashed] node {} (5)
;
\end{tikzpicture}
\caption{\label{fig:projectedMod}Bow-free graph, $\bar G$.}
\end{subfigure}
\caption{\label{fig:exampleModelMiss}Graphs used in Example~\ref{ex:modelMisspecify}. The ``projection'' of (a) into a BAP is shown in (b) where dotted lines indicate edges that differ from the ``true model.''}  
\end{figure}
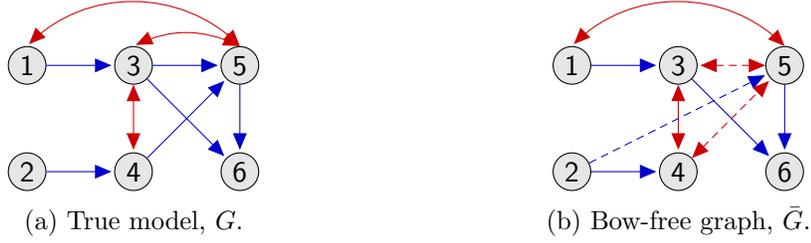
Consider the graph $G$ from panel (a) of Figure~\ref{fig:exampleModelMiss}. There is a bow between $3$ and $5$. Thus, $ \Irr(5)_1 = \{3\} \cup \{5\}$. 

Furthermore, $1 \in \pa(3) \cap \sib(5) \subseteq \pa(\Irr(5)_1) \cap \sib(\Irr(5)_1)$; thus, $1 \in \Irr(5)_2$. Similarly, $4 \in \pa(5) \cap \sib(3) \subseteq \pa(\Irr(5)_1) \cap \sib(\Irr(5)_1)$; thus $4 \in \Irr(5)_2$. No other nodes are in $\pa(\{3,5\})$ so $\Irr(5)_2 = \{1,4\} \cup \{3,5\}$. 

Next, note that $2 \in \pa(4)$ but $2 \not \in \sib(\Irr(5)_2)$ so it is not a member of $\Irr(5)_3$. Also, $6 \not \in \pa(\Irr(5)_2)$, so it is also not a member of $\Irr(5)_3$. Thus, for $s = 3, \ldots, 6$, $\Irr(5)_{s} = \Irr(5)_2 = \{1,3, 4, 5\}$ and $\irr(5) = \{1,3,4,5\}$.

There are no other bows in the graph, so $\irr(v) = \{v\}$ for all
other nodes and $\bar \sib(5) = \{1,3,4\}$. As $2 \in \pa(\irr(5))
\setminus \bar \sib(5)$, we have $2 \in \bar \pa(5)$. 

As we will show, when applied to population values corresponding to
the graph $G$ in (a), BANG will recover the BAP projection $\bar G$
displayed in panel (b). The dotted lines in (b) indicate edges in $\bar G$ which are different from the edges in $G$.
\end{example}

To develop our result on estimation of $\bar G$, we first show that the distribution of $Y$ is also in the model implied by BAP $\bar G$. 
\begin{lemma}\label{lem:bapProjection}
Let $G$ be an acyclic directed mixed graph, and suppose the random
vector $Y$ follows a distribution in the
linear SEM given by $G$. Let $\bar G = \left( V, \bar E_\rightarrow,
  \bar E_\leftrightarrow \right)$ be the projection of $G$ given by
Definition~\ref{def:barG}. Then $\bar G$ is a BAP. Furthermore, the
distribution of $Y$ is equal to the distribution in the linear SEM
given by $\bar G$ when defining edge weights and errors as follows:
\begin{equation}
\begin{aligned}\label{eq:projectionParam}
\bar    \beta_{v,u} &= \begin{cases}
\sum_{l \in \mathcal{L}^{(u)}_{v,u}(\bar \pa(v))}W(l) & \text{ if } u \in \bar \pa(v),\\
0 & \text{ else, }
\end{cases}\\
\quad & \quad \\
\bar \varepsilon_v &= \varepsilon_v + \sum_{w \in \irr(v)} \varepsilon_w \sum_{\substack{l \in \mathcal{L}_{v,w}\\
l\subseteq \irr(v)}} W(l). 
\end{aligned}
\end{equation}
\end{lemma}
\begin{proof}   
  Acyclicity of $G$ implies acyclicity of 
the projection $\bar G$ because $\bar{\pa}(v) \subseteq \an(v)$. Furthermore, by definition, $\bar \pa(v)$ does not include $\bar \sib(v)$, so $\bar G$ does not contain any bows. Thus, $\bar G$ is a BAP.

We now show that the distribution of $Y$ belongs to the SEM given by
$\bar G$. Recall the definition of $\mathcal{L}^{(w)}_{v,w}(\an(v))$
from~Section\ \ref{sec:setup}. Note that $\mathcal{L}^{(w)}_{v,w}(\an(v))= \emptyset$ for any $w \not \in \bar \pa(v) \cup \bar{\sib}(v)$ because if $w \not \in \bar{\sib}(v)$ and there was a path from $w$ to $v$ for which $w$ was the last node, then $w$ would be in $\bar \pa(v)$.
Hence,
\begin{equation}\allowdisplaybreaks\begin{aligned} \label{eq:yvExpansion}
    Y_v &= \varepsilon_v + \sum_{w \in \an(v)} \pi_{v,w}\varepsilon_w = \varepsilon_v + \sum_{w \in \an(v)} \varepsilon_w\sum_{l \in \mathcal{L}_{v,w}}W(l)  \\
    &= \varepsilon_v + \sum_{w \in \an(v)} \varepsilon_w \left(\sum_{\substack{l \in \mathcal{L}_{v,w}\\ \bar \pa(v) \cap l \neq \emptyset}}W(l) + \sum_{\substack{l \in \mathcal{L}_{v,w}\\ \bar \pa(v) \cap l = \emptyset}}W(l) \right)  \\
    &= \varepsilon_v + \sum_{w \in \an(v)} \varepsilon_w \left(\sum_{s \in \bar \pa(v)} \sum_{l \in \mathcal{L}^{(s)}_{v,w}(\bar \pa(v))}W(l) + \sum_{\substack{l \in \mathcal{L}_{v,w}\\ l \subseteq \irr(v)}}W(l) \right)  \\ 
    &= \varepsilon_v + \sum_{w \in \an(v)} \varepsilon_w \left(\sum_{s \in \bar \pa(v)} \pi_{s,w}\sum_{l \in \mathcal{L}^{(s)}_{v,s}(\bar \pa(v))}W(l) + \sum_{\substack{l \in \mathcal{L}_{v,w}\\ l \subseteq \irr(v)}}W(l) \right)  \\
    &= \varepsilon_v + \sum_{s \in \bar \pa(v)} \sum_{w \in \an(v)} \pi_{s,w} \varepsilon_w  \sum_{l \in \mathcal{L}^{(s)}_{v,s}(\bar \pa(v))}W(l) + \sum_{w \in \an(v)} \varepsilon_w\sum_{\substack{l \in \mathcal{L}_{v,w}\\ l \subseteq \irr(v)}}W(l).
    \end{aligned} 
 \end{equation}  
Because $\an(s) \subseteq \an(v)$ if $s \in \bar{\pa}(v)\subseteq \an(v)$, we have that
  \begin{equation}\allowdisplaybreaks\begin{aligned}  
    \eqref{eq:yvExpansion} &= \varepsilon_v + \sum_{s \in \bar \pa(v)} Y_s  \sum_{l \in \mathcal{L}^{(s)}_{v,s}(\bar \pa(v))}W(l) + \sum_{w \in \an(v)} \varepsilon_w \sum_{\substack{l \in \mathcal{L}_{v,w}\\ l \subseteq \irr(v)}}W(l)  \\
    &= \sum_{s \in \bar \pa(v)} Y_s  \underbrace{\sum_{l \in \mathcal{L}^{(s)}_{v,s}(\bar \pa(v))}W(l)}_{\bar \beta_{v,s}} + \underbrace{\varepsilon_v + \sum_{w \in \irr(v)} \varepsilon_w \sum_{\substack{l \in \mathcal{L}_{v,w}\\ \bar l \subseteq \irr(v)}}W(l)}_{\bar \varepsilon_v}.  \\
\end{aligned}
\end{equation}
Our claim now follows because the coefficients $\bar\beta_{v,u}$
respect the constraints given by $\bar G$.  Indeed, $\bar \beta_{v,u} = 0 $ if $ u \not \in \bar \pa(v)$. Furthermore, if $\irr(u) \cap \sib(\irr(v)) = \emptyset$, then $\bar \varepsilon_v$ only contains terms which are independent of the terms in $\bar \varepsilon_u$. Thus, $\varepsilon_v \independent \varepsilon_u$ if $u \not \in \bar{\sib}(v)$.  
\end{proof}

Though, as we subsequently show, it is true that given population values BANG returns $\bar G$, Lemma~\ref{lem:bapProjection} does not immediately imply that BANG will discover $\bar G$; it simply implies that the intersection of these models is non-trivial. It must further be shown that for generic parameters (of the full model), the distribution of $Y$ does not lie in any sub-model of $\bar G$.  

Replacing $B$ with $\bar B$, Lemma~\ref{thm:debiasedEst},
Corollary~\ref{thm:suffInd},  Lemma~\ref{thm:neccIndGeneral}, and
Lemma~\ref{thm:neccIndDes} still directly hold in this setting. We
restate these results below for completeness, but note that the proofs
follow in analogy to the previously proved lemmas. Thus, to prove an analogue to Theorem~\ref{thm:mainID}, it remains to show analogues to Lemma~\ref{thm:noMissPruneParents}, Lemma~\ref{thm:neccIndSib}, and Corollary~\ref{thm:noMissPruneParentsFinal}.

\begin{corollary}[Analogue of Lemma~\ref{thm:debiasedEst} and Corollary~\ref{thm:suffInd}]\label{lem:parentBarCertified}
Suppose $Y$ is generated by a linear SEM which corresponds to a acyclic mixed graph $G$. Let $\bar G$ be the projection of $G$ and $\bar B$ be the corresponding direct effects
defined in \eqref{eq:projectionParam}. For a node $v \in V$ and a set $C \subseteq A \subseteq V \setminus v$, suppose:
\begin{enumerate}
    \item $\bar{\pa}(v) \subseteq C \subseteq \bar{\an}(v) \setminus \bar{\sib}(v)$
    \item $A = \bar{\An}(C)$
    \item $D_{A,A} = \bar B_{A,A}$
    \item $S_{\{A, v\}, \{A, v\}} = \Sigma_{\{A, v\}, \{A, v\}}$
\end{enumerate}
Then $\delta_v(C, A, S, D) = \bar{B}_{v, C}$. Furthermore, for every $c \in C$, $\gamma_c(D) \independent \gamma_v(C,S,D)$ and $\E(\gamma_c^{K-1}(D) \gamma_v(C,S,D))$.
\end{corollary}
\begin{proof}
By Lemma~\ref{lem:bapProjection}, the distribution of $Y$ is equivalent to the BAP defined by $\bar G$, $\bar B$ and $\bar \varepsilon$. Thus, the result directly follows by applying Lemma~\ref{thm:debiasedEst} and Corollary~\ref{thm:suffInd} to $\bar G$, $\bar B$ and $\bar \varepsilon$.
\end{proof}

In the following statements, we will at times make statements about sets of nodes in $\bar G$; however, the requirement of generic parameters will always refer to parameters in the model which may have bows defined by $G$. Let $\check B(C\cup\{v\})_{v,C}$ be the marginal direct effect of $C\cup\{v\}$ in $\bar{G}$; i.e., the analogue of $\tilde B$, as defined in \eqref{eq:margDirectEffect}, but using $\bar{B}$ instead of $B$. 

\begin{corollary}[Analogue of Lemma~\ref{thm:neccIndGeneral}]\label{thm:neccIndGeneralMiss}
  Let $v \in V$, and consider any set $C \subseteq A \subseteq V
  \setminus \{v\}$. Suppose $D \in \mathbb{R}^{p \times p}$ with $D_{s,t} \neq 0$ only if $t \in \bar{\an}(s)$. Then, for generic $B$ and error moments, if $\delta_v(C, A, S, D) \neq \check B(C\cup v)_{v,C}$, then $\E(\gamma_c^{K-1}(D)\gamma_v(C, S, D)) \neq 0$ for some $c \in C$.
\end{corollary}

\begin{corollary}[Analogue of Lemma~\ref{thm:neccIndDes}]
\label{thm:neccIndDesMiss}
Consider $v \in V$ and set $C \subseteq V\setminus \{v\}$. Let $D \in \mathbb{R}^{p \times p}$ such that $D_{s,t} \neq 0$ only if $t \in \bar{\an}(s)$.
Suppose $C \not \subseteq \bar{\an}(v)$, but 
$\E(\gamma_c(D)^{K-1}\gamma_v(C, S,D)) = 0$
for all $c \in C$. Then for generic $B$ and error moments,  $C_1 = C\cap \left[\bar{\an}(v)\setminus \bar{\sib}(v)\right]$,
\begin{equation*} \E(\gamma_c(D)^{K-1}\gamma_v(C_1, S ,D)) = 0\end{equation*} for all $c \in C$.
\end{corollary}

\begin{lemma}[Analogue of Lemma~\ref{thm:noMissPruneParents}]\label{thm:noMissPruneParentsMiss}
Suppose $D = H_\mathcal{C}(\bar B)$ for some $H_\mathcal{C} \in \mathcal{D}$ with $\mathcal{C}=(C_s)_{s\in V}$ such that $C_s \subseteq \bar{\an}(s) \setminus \bar{\sib}(s)$ for all $s \in V$.  Let $v \in V$ be such that we have 
$
\E(\gamma_c(D)^{K-1}\gamma_v(D)) = 0$ for all $c \in \pspa(v)$.
If $q \in \big(\bar{\pa(v)}\setminus \pspa(v)\big) \cup \bar{\sib}(v)$, then for generic $B$ and error moments,  $\E \left(\gamma_q(D)^{K-1}\gamma_v(D)\right) \neq 0$.
\end{lemma}

\begin{lemma}[Analogue of Lemma~\ref{thm:neccIndSib}]\label{thm:neccIndSibMiss}
	Consider $v \in V$ and sets $A, C$ such that $C \subseteq A \subseteq V \setminus \{v\}$. Suppose $D = H_\mathcal{C}(\bar B)$ for some $H_\mathcal{C} \in \mathcal{D}$ with $\mathcal{C}=(C_s)_{s\in V}$ such that $C_s \subseteq \bar{\an}(s) \bar{\sib}(s)$ for all $s \in V$. Suppose $u \in C$ and $u \in  \bar{\sib}(v)$, then for generic $B$ and error moments, there exists some $q \in C$ such that $\E\left(\gamma_q(D)^{K-1}\gamma_v(C, \Sigma, D)\right) \neq 0$.
\end{lemma}

\begin{corollary}[Analogue of Corollary~\ref{thm:noMissPruneParentsFinal}]\label{thm:noMissPruneParentsFinalMis}
	Suppose $D = \bar B$. For $v \in V$ and generic $B$ and error moments, suppose $\bar{\pa}(v) \subseteq C \subseteq \bar{\an}(v)\setminus \bar{\sib}(v)$ and $\E(\gamma_c(D)^{K-1}\gamma_v(C, \Sigma,D)) = 0$ for all $c \in C$. If $q \in C \setminus \bar{\pa}(v)$, the for all $c \in C$
	\begin{equation}
	\E(\gamma_q(D)^{K-1}\gamma_v(C \setminus \{q\}, \Sigma,D)) = 0.
	\end{equation}
	If $q \in \pa(v)$, then there exists some $c \in C$ such that
	\begin{equation}
	    \E(\gamma_q(D)^{K-1}\gamma_v(C \setminus \{q\}, \Sigma,D)) \neq 0.
	\end{equation}
\end{corollary}

Lemma~\ref{thm:noMissPruneParentsMiss} and \ref{thm:neccIndSibMiss} require two intermediate results which we state here for completeness.

\begin{lemma}\label{lem:nonZeroZeta}
Let $D = H_\mathcal{C}(\bar B)$ for some $H_\mathcal{C} \in \mathcal{D}$ with $\mathcal{C} = (C_s)_{s \in V}$ such that $C_s \subseteq \bar{\an}(s) \setminus \bar{\sib}(s)$ for all $s \in V$. Suppose there exists some path $l = s_1 \rightarrow s_2 \rightarrow \ldots s_{|l|-1} \rightarrow v$ such that $ l \cap C_v = \emptyset$. Further suppose that $u \in \sib(s_1) \setminus l$ and $u \not \in C_v$. Then for generic parameters
\begin{equation}
    \E\left(\gamma_v(D)^{K-1} \gamma_{s_1}(D)\right) \neq 0 \qquad \text{ and } \qquad     \E\left(\gamma_v(D)^{K-1} \gamma_u(D)\right) \neq 0,
\end{equation}
so that $s_1$ and $u$ will not be pruned from $\widehat{\sib}(v)$ by Alg~\ref{alg:checkInd}. Furthermore, for any $C$ such that $u \in C$, for generic parameters there exists some $c \in C$ such that
\begin{equation}
\E\left(\gamma_c(D)^{K-1}\gamma_v(C, \Sigma, D)\right) \neq 0,
\end{equation}
so that $u$ will not be certified into $\widehat \pa(v)$.
\end{lemma}

\begin{lemma}\label{lem:cousins}
Suppose $\irr(u) \cap \sib(\irr(v)) \neq \emptyset$ and $D = H_{\mathcal{C}}(\bar{B})$ for some $\mathcal{C} = (C_s)_{s \in V}$ such that $C_s \subseteq \bar{\an}(s) \setminus \bar{\sib}(s)$ for all $s \in V$. Then, for generic parameters
\begin{equation}
     \E\left(\gamma_v(D)^{K-1} \gamma_u(D)\right) \neq 0
\end{equation}
so that $u$ will not be pruned from $\widehat{\sib}(v)$ by Alg~\ref{alg:checkInd}. Furthermore, for any $C \subseteq V \setminus v$ such that $u \in C$, for generic parameters, there exists some $c \in C$ such that
\begin{equation}
\E\left(\gamma_c(D)^{K-1}\gamma_v(C, \Sigma, D)\right) \neq 0,
\end{equation}
so that $u$ will not be certified into $\widehat \pa(v)$.
\end{lemma}

\begin{theorem}\label{lem:consistentMisspecified}
Suppose $Y$ is generated under the linear SEM given by an acyclic directed mixed graph $G$ with BAP projection $\bar G$ as defined by \eqref{eq:sibProjection} and \eqref{eq:paProjection}. Then for generic choices of $B$ and error moments, BANG will output $\hat G = \bar G$ when given population values of $Y$.
\end{theorem}
\begin{proof}
The proof exactly mirrors the proof of Theorem~\ref{thm:mainID}, but using the
lemmas
developed for the misspecified case.
\end{proof}

\section{Numerical results}\label{sec:numericalResults}
We consider two implementations of BANG\footnote{Available at \url{https://github.com/ysamwang/ngBap}}: one which uses empirical
likelihood to test whether moments are zero or non-zero and another
which uses dHSIC~\citep{pfister2018independence} to test whether
certain variables are independent or not. We compare these
implementations against ParcelLiNGAM~\citep{tashiro2014parcelingam} as
well as two methods for Gaussian data---FCI+
\citep{claassen2013learning} with Gaussian conditional independence
tests (i.e. vanishing partial correlations) and Greedy BAP Search
(GBS) \citep{nowzohour2015structure}. For ParcelLiNGAM we use the
Matlab implementation available from the author's
website\footnote{\url{https://sites.google.com/site/sshimizu06/Plingamcode}};
for FCI+, we use the \texttt{R} package \texttt{pcalg}
\citep{kalisch2012pcalg} and for GBS we use the \texttt{R} package
\texttt{greedyBaps} \citep{nowzohour2017greedy}.

Our experiments consider structure learning in two settings: ancestral graphs and
BAPs.  In each case, we simulate errors from gamma, log-normal, and uniform
distributions. In addition, we include a setting with errors drawn
from $T_{13}$ as a counter example to show how performance can suffer
when the errors are close to Gaussian. Finally, we show that when
applied to ecology data the BANG method recovers a model close to the
ground truth.


\subsection{Comparison with ParcelLiNGAM}
ParcelLiNGAM is designed to discover ancestral relationships, not graph structure. Also, as shown in Section~\ref{sec:ancestral} it is sound and complete for ancestral graphs, but not non-ancestral graphs. Thus, to give ParcelLiNGAM the most favorable comparison, we only consider ancestral graphs, and to compare performance, we measure the accuracy in identifying ancestral relationships; i.e., for each $(u, v)$, is $u \in \an(v)$, $u \not \in \an(v)$. The accuracy is defined as $(\text{True Positives} + \text{True Negatives}) / \text{Total Cases}$.

We let $p = 6$ and consider three settings with varying levels of sparsity with $d$ directed edges and $b$ bidirected edges: \emph{sparse} ($d = p/2$, $b = p/2$), \emph{medium} ($d = p$, $b = p$), and \emph{dense} ($d = 3p/2$, $b = p$). We let $n = 500, 1000, 1500$ since the BANG dHSIC and ParcelLiNGAM procedures become computationally infeasible when $n > 2000$.

To generate a random ancestral graphs, we first select $d$ directed edges uniformly from the set $\{(i,j) :  i < j\}$, and then select $b$ bidirected edges uniformly from the set $\{(i,j) :  i \not \in \an(j) \text{ and } j \not \in \an(i)\}$ when generating ancestral graphs if the set of possible bidirected edges is less than $b$, we select as many as possible. We then draw the directed edgeweights uniformly from $\pm(.6, 1)$. Note that the graphs may not be maximal.  

For the idiosyncratic errors, we first form the covariance $\Omega$ by setting $\omega_{ii} = 1$ for all $i \in V$, drawing the $\omega_{ij} = \omega_{ji}$ uniformly from $\pm (.3, .5)$ for all $(i,j) \in E_\leftrightarrow$, and setting all other elements to $0$. If $\Omega$ is not positive definite, we repeatedly scale the off-diagonal elements of $\Omega$ by $.97$ until the minimum eigenvalue is greater than $.01$. We consider five settings where the errors marginally follow uniform, gamma, lognormal, and $T$ with $d.f. = 13$. We center and scale the errors so that (in expectation) they marginally have mean 0 and variance 1. We draw the gamma errors using \texttt{lcmix} \citep{lcmix}, uniform using \texttt{MultiRNG} \citep{multiRNG}, $T_{13}$ using \texttt{mvtnorm} \citep{mvtnorm}, and the lognormal errors are formed by exponentiating multivariate normal draws with covariance $\Omega$ which are subsequently scaled and centered so that each element has (in expectation) mean $0$ and variance $1$.


We then set $Y^{(i)} = (I-B)^{-1}\epsilon^{(i)}$. Finally, because the output generally depends on the ordering of the variables in the data matrix, we also randomly permute the labeling of the variables so that $1, \ldots, p$ is generally not a valid ordering. This entire process is repeated 200 times for each setting.

The results are given in Table~\ref{tab:plComp}. For BANG with dHSIC or empirical likelihood, the value is bolded if the accuracy is significantly larger than the ParcelLiNGAM accuracy (measured using a two-sample paired T-test with with $\alpha = .05$). For ParcelLiNGAM, the value is bolded if the accuracy is significantly larger than both of the BANG implementations. In general, we see that ParcelLiNGAM tends to do better when the graph is dense; this is particularly drastic when the errors are uniform, but less so for the other error types. Under the ``medium'' and ``sparse'' graph regimes, BANG tends to outperform ParcelLiNGAM, particularly the dHSIC implementation. In the sparse regime, the difference is quite drastic for all error settings. In general, when the errors are $T$ and not too far from Gaussian, all three methods suffer.

\begin{table}[t]
    \centering
\caption{\label{tab:plComp}The average accuracy of each procedure in identifying ancestral relationships across 200 replications. DH: BANG using dHSIC independence tests, EL: BANG using empirical likelihood moment tests, PL: ParcelLiNGAM. All procedures use independence tests with $\alpha = .01$.}
\small
\begin{tabular}{l|l|c|c|c|c}
\hline
Regime & Errors & n & DH & EL & PL\\
\hline
 &  & 500 & 0.871 & 0.814 & \textbf{0.905}\\
\cline{3-6}
 &  & 1000 & 0.914 & 0.840 & \textbf{0.937}\\
\cline{3-6}
 & \multirow{-3}{*}{\raggedright\arraybackslash Gamma} & 1500 & 0.923 & 0.885 & \textbf{0.950}\\
\cline{2-6}
 &  & 500 & 0.902 & 0.790 & \textbf{0.925}\\
\cline{3-6}
 &  & 1000 & 0.914 & 0.811 & \textbf{0.934}\\
\cline{3-6}
 & \multirow{-3}{*}{\raggedright\arraybackslash Lognormal} & 1500 & 0.895 & 0.820 & \textbf{0.931}\\
\cline{2-6}
 &  & 500 & 0.524 & 0.526 & 0.504\\
\cline{3-6}
 &  & 1000 & 0.517 & 0.513 & 0.521\\
\cline{3-6}
 & \multirow{-3}{*}{\raggedright\arraybackslash T} & 1500 & 0.495 & 0.510 & 0.502\\
\cline{2-6}
 &  & 500 & 0.683 & 0.625 & \textbf{0.794}\\
\cline{3-6}
 &  & 1000 & 0.742 & 0.699 & \textbf{0.823}\\
\cline{3-6}
\multirow{-12}{*}{\raggedright\arraybackslash Dense} & \multirow{-3}{*}{\raggedright\arraybackslash Uniform} & 1500 & 0.733 & 0.687 & \textbf{0.851}\\
\cline{1-6}
 &  & 500 & \textbf{0.903} & 0.852 & 0.866\\
\cline{3-6}
 &  & 1000 & \textbf{0.936} & 0.888 & 0.911\\
\cline{3-6}
 & \multirow{-3}{*}{\raggedright\arraybackslash Gamma} & 1500 & \textbf{0.936} & 0.900 & 0.922\\
\cline{2-6}
 &  & 500 & \textbf{0.915} & 0.823 & 0.893\\
\cline{3-6}
 &  & 1000 & \textbf{0.931} & 0.845 & 0.912\\
\cline{3-6}
 & \multirow{-3}{*}{\raggedright\arraybackslash Lognormal} & 1500 & \textbf{0.951} & 0.870 & 0.921\\
\cline{2-6}
 &  & 500 & \textbf{0.555} & \textbf{0.545} & 0.493\\
\cline{3-6}
 &  & 1000 & \textbf{0.537} & \textbf{0.560} & 0.492\\
\cline{3-6}
 & \multirow{-3}{*}{\raggedright\arraybackslash T} & 1500 & 0.525 & \textbf{0.551} & 0.510\\
\cline{2-6}
 &  & 500 & 0.726 & 0.704 & 0.722\\
\cline{3-6}
 &  & 1000 & 0.753 & 0.729 & 0.756\\
\cline{3-6}
\multirow{-12}{*}{\raggedright\arraybackslash Medium} & \multirow{-3}{*}{\raggedright\arraybackslash Uniform} & 1500 & 0.784 & 0.747 & 0.789\\
\cline{1-6}
 &  & 500 & \textbf{0.953} & \textbf{0.931} & 0.688\\
\cline{3-6}
 &  & 1000 & \textbf{0.974} & \textbf{0.951} & 0.706\\
\cline{3-6}
 & \multirow{-3}{*}{\raggedright\arraybackslash Gamma} & 1500 & \textbf{0.982} & \textbf{0.970} & 0.713\\
\cline{2-6}
 &  & 500 & \textbf{0.962} & \textbf{0.905} & 0.698\\
\cline{3-6}
 &  & 1000 & \textbf{0.980} & \textbf{0.918} & 0.716\\
\cline{3-6}
 & \multirow{-3}{*}{\raggedright\arraybackslash Lognormal} & 1500 & \textbf{0.984} & \textbf{0.934} & 0.720\\
\cline{2-6}
 &  & 500 & \textbf{0.695} & \textbf{0.680} & 0.503\\
\cline{3-6}
 &  & 1000 & \textbf{0.686} & \textbf{0.700} & 0.500\\
\cline{3-6}
 & \multirow{-3}{*}{\raggedright\arraybackslash T} & 1500 & \textbf{0.664} & \textbf{0.695} & 0.496\\
\cline{2-6}
 &  & 500 & \textbf{0.862} & \textbf{0.847} & 0.632\\
\cline{3-6}
 &  & 1000 & \textbf{0.883} & \textbf{0.892} & 0.667\\
\cline{3-6}
\multirow{-12}{*}{\raggedright\arraybackslash Sparse} & \multirow{-3}{*}{\raggedright\arraybackslash Uniform} & 1500 & \textbf{0.891} & \textbf{0.909} & 0.676\\
\hline
\end{tabular}
\end{table}

\subsection{Comparison with FCI+ and Greedy BAP search}
We now compare BANG against FCI+ and GBS, which both identify an equivalence class of graphs. We compare FCI+, GBS, and BANG on ancestral graphs generated as described in the previous section. Furthermore, we compare GBS and BANG on possibly non-ancestral BAPs. The graphs are generated by the same procedure, except when generating BAPs we do not enforce the ancestral condition and instead draw bidirected edges from the set $\{(i,j) :  i \not \in \pa(j) \text{ and } j \not \in \pa(i)\}$. We let $n = 500, 1000, 2000, 10000, 25000, 50000$, but for computational reasons we only use the dHSIC implementation of BANG for $n \leq 2000$.

To compare performance, we record the proportion of the times each method recovers the equivalence class corresponding to the true graph (or a graph in the equivalence class of the true graph). For BANG, we also record the proportion of times it recovers the true exact graph. For BANG and FCI+, we set the nominal level of each hypothesis test performed to $\alpha = .05, .01$. For GBS, we allow 100 random restarts, the same number used in the simulations by \citet{nowzohour2017greedy}.  For BANG, we set $K = 3$ for the gamma and lognormal errors (since they are skewed) and let $K = 4$ for the uniform and $T_{13}$ (since they are symmetric). In the simulations with ancestral graphs, to check whether the equivalence class is recovered, we take the graph estimated by BANG and GBS and project it into a PAG. In the simulations with BAPs, since there is no known graphical characterization for the equivalence class of BAPs, we follow \citet{nowzohour2017greedy} and say that the estimated and true graph are in the same equivalence class if the score of the estimated graph is within $10^{-10}$ of score of the true graph. We repeat the experiment 500 times for each simulation setting. The results for ancestral graphs are shown in Figure~\ref{fig:magSimsFlip} and the results for BAPs are shown in Figure~\ref{fig:bapSimsFlip}.

We make several observations about the results. First, the performance of FCI+ and GBS does not seem to change drastically across the different error distributions. BANG does best when the errors are gamma which we posit occurs because we only estimate moments up to degree 3 (when compared to the uniform which requires estimation of $4$th moments) and of the gamma's relatively lighter tails (when compared to the lognormal). As expected, because the errors are too close to a Gaussian, BANG performs very poorly under the $T_{13}$ setting. 

We also see that GBS and FCI+ do well when the graph is sparse, and GBS tends to outperform both FCI and BANG. However, the performance of GBS and FCI+ begins to suffer when the true graph becomes more dense. Although the performance of BANG also decreases, it decreases less in comparison, and BANG outperforms the other methods in the medium and dense setting. There are even  many settings where BANG recovers the true graph more often than GBS and FCI+ recover the equivalence class. 

When comparing the two BANG implementations, we see that dHSIC tends to outperform empirical likelihood when the errors are gamma or lognormal, but empirical likelihood slightly outperforms the dHSIC implementation when the errors are uniform. This may be due to the fact that we use the ``gamma approximation'' for the null distribution of the dHSIC test statistic because using a permutation test is too computationally expensive.     

In the appendix, we also consider settings where we restrict the elements of $B$ and $\Delta$ to be positive which results in fewer ``faithfulness'' violations. In this setting, the performance of BANG tends to improve (particularly in the gamma and lognormal settings) while the performance of GBS and FCI+ do not change substantially. This suggests that BANG is much more dependant on ``strong faithfulness'' than the other methods. 


\begin{figure}[tb]
	\centering

    \includegraphics[scale=.82]{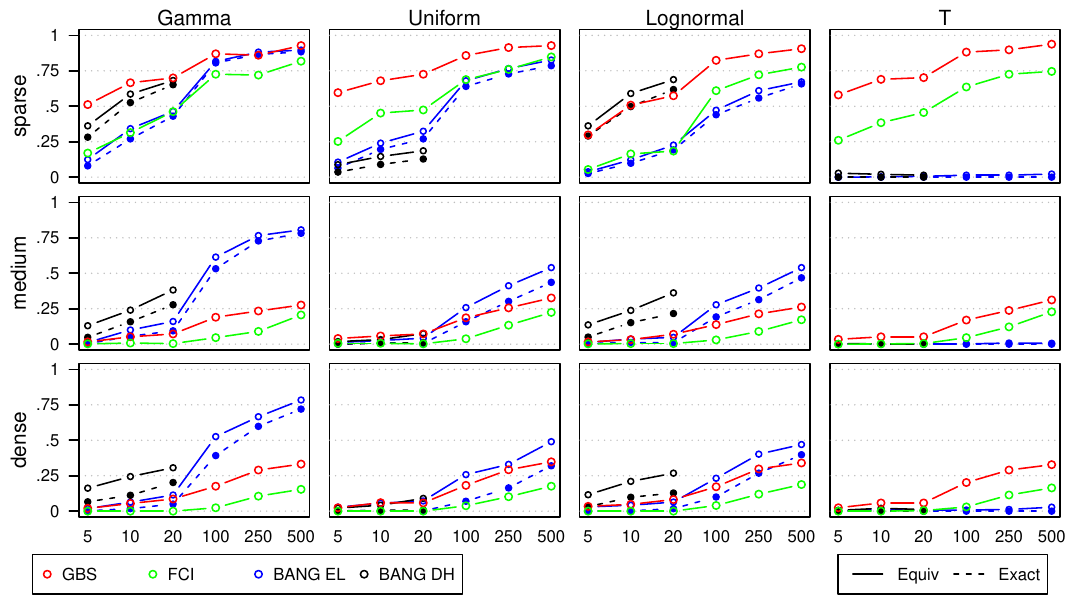}
    \caption{\label{fig:magSimsFlip}The performance of each method on random ancestral graphs with $p=6$ across 500 replications. The solid lines indicate the proportion of times the estimated graph corresponds to the PAG of the true graph; the dotted lines indicates the proportion of times BANG identifies the exact graph. The horizontal axis shows sample size in hundreds.}
\end{figure}

\begin{figure}[tb]
	\centering
    \includegraphics[scale=.82]{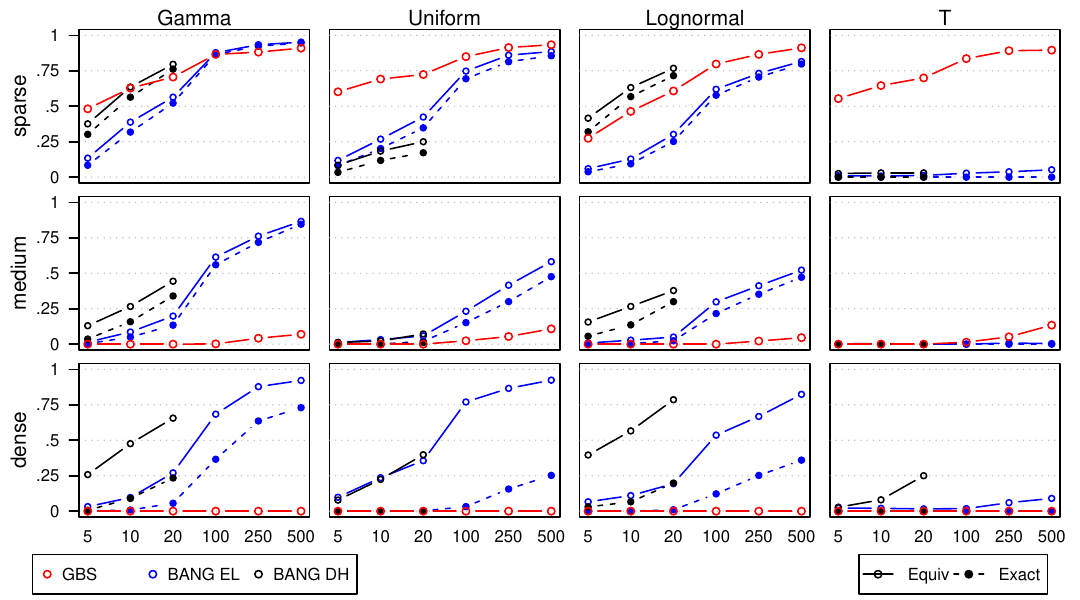}
    \caption{\label{fig:bapSimsFlip}The performance of each method on random bow-free acyclic graphs with $p=6$ across 500 replications. The solid lines indicate the proportion of times the estimated graph is in the equivalence class of the truth; the dotted lines indicates the proportion of times BANG identifies the exact graph. The horizontal axis shows sample size in hundreds.}
\end{figure}

\subsection{Data example}

\begin{figure}[htb]
	\centering
	\includegraphics[scale = .25]{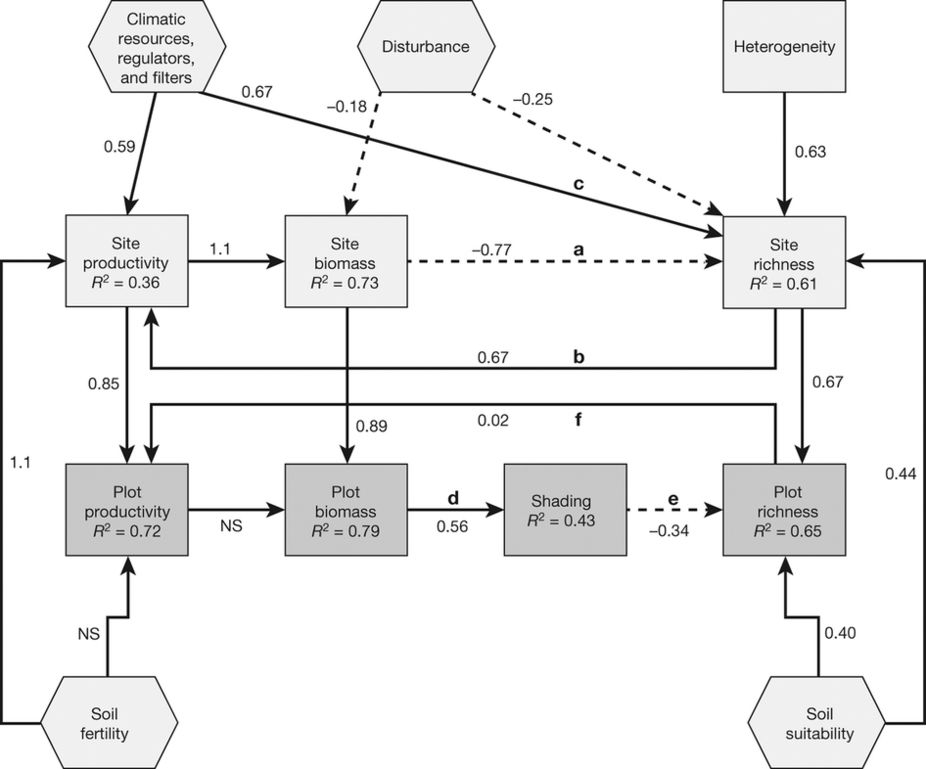}
	\caption[Data example: full model from Grace et al. (2016)]{\label{fig:grace2016orig}Full model from \citet{grace2016integrative}.}
\end{figure}

\citet{grace2016integrative} use a structural equation model to examine the relationships between land productivity and the richness of plant diversity. They consider measurements taken at 
1126 plots which are locations across 39 different sites. A graphical model from the original paper is in Figure~\ref{fig:grace2016orig}. We consider their plot level model which includes: plot productivity, plot biomass, plot shade, plot richness, plot soil suitability, site richness, site biomass, site productivity.



Beginning with the graphical model shown in Figure~\ref{fig:grace2016orig}, we first remove any edges which they found were not significant (denoted by NS in Figure~\ref{fig:grace2016orig}). Note that this removes the cycle in the plot specific measurements, but there is still a cycle between site productivity, biomass and richness. The nodes for climate, disturbance and suitability, actually represent multiple variables which are used in the SEMs. For climate and disturbance, the separate measures are both highly correlated, so it seems reasonable to use bidirected edges between site productivity, biomass and richness when marginalizing out those variables, despite the fact that they are actually separate measures. To keep the bow-free assumption, we do not include the directed edges between site productivity, site biomass and site richness. This results in ancestral relationships in the full model which are not otherwise captured in the marginalized model. Thus, we add directed edges from site productivity to plot biomass and plot richness; from site biomass to plot productivity and plot richness; from site richness to plot productivity and plot biomass. For suitability, there is both a site suitability, which is a parent of site richness, and a plot suitability which is a parent of plot richness. Although there is no explicit specification in their SEM of how site suitability relates into plot suitability, it seems reasonable to assume that site suitability has a direct effect on plot suitability, as is the case for all other site vs plot measures. Thus, we include a bidirected edge between plot suitability and site richness. This results in the BAP shown in Figure~\ref{fig:targetBAP}. We consider this model the ground truth.

For BANG, we use empirical likelihood and selected the nominal test level, $.01$, so that there are roughly the same number of directed edges in the estimated and ground truth graphs, 11 and 13 respectively. The discovered graph is shown in Figure~\ref{fig:discovered}. Of the 28 pairs of nodes, BANG correctly identifies the correct relation ($\rightarrow, \leftarrow, \leftrightarrow$ or no edge) for 16 of the pairs. Naively, letting the probability of guessing each relationship to be $1/4$, this results in a binomial probability of $P(X \geq 16) = .00029$. This probability does not account for the dependency between edges since there is an acyclic restriction, but it suggests that BANG is discovering reasonable structure. There are 7 bidirected edges in the estimated graph compared to 4 in the ground truth model. This behavior is somewhat expected since there is still likely to be uncontrolled confounding which is either not actually fully accounted for in the ground truth model or direct causes which cannot be fully explained by a linear relationship. For comparison, we also use the GBS procedure with 500 random restarts. In Figure~\ref{fig:gbsRes}, we plot the resulting score against the number of correct edges for each of the 500 runs. There seems to be a positive association between the score and the correct number of edges. Although one initialization resulted in a graph with 16 correct edges it did not have the highest score, and each of the resulting estimated graphs with maximum score (up to optimization error) only has 12 correct edges.

\begin{figure}[htb]
	\centering
	\begin{subfigure}[t]{1\textwidth}\centering
	\begin{tikzpicture}[->,>=triangle 45,shorten >=1pt,
auto,
main node/.style={rounded corners,inner sep=0pt,fill=gray!20,draw,font=\sffamily,
	minimum width = 1.5cm, minimum height = .5cm, scale=0.95}]

\node[main node] (Pprod) {Pl Prod};
\node[main node] (Pbio) [right = .95cm of Pprod]  {Pl Bio};
\node[main node] (Shade) [right = .95cm of Pbio]  {Pl Shade};
\node[main node] (Prich) [right = .95cm of Shade]  {Pl Rich};
\node[main node] (Sprod) [above = .95cm of Pprod]  {St Prod};
\node[main node] (Sbio) [above = .95cm of Pbio]  {St Bio};
\node[main node] (Srich) [above = .95cm of Prich]  {St Rich};
\node[main node] (Suit) [right = .95cm of Prich]  {Pl Suit};

\path[color=black!20!blue,style={->}]
(Sprod) edge node {} (Pprod)
(Sprod) edge node {} (Pbio)
(Sprod) edge node {} (Prich)
(Sbio) edge node {} (Prich)
(Sbio) edge node {} (Pbio)
(Sbio) edge node {} (Pprod)
(Srich) edge node {} (Pprod)
(Srich) edge node {} (Pbio)

(Srich) edge node {} (Prich)

(Pbio) edge node {} (Shade)
(Shade) edge node {} (Prich)
(Shade) edge node {} (Prich)
(Prich) edge[bend left = 20] node {} (Pprod)
(Suit) edge node {} (Prich);

\path[color=black!20!red,style={<->}]
(Suit) edge node {} (Srich)
(Sprod) edge[bend left = 15] node {} (Srich)
(Sbio) edge node {} (Srich)
(Sprod) edge node {} (Sbio);

\end{tikzpicture}
\caption[Data example: BAP respresentation from Grace et al. (2016)]{\label{fig:targetBAP}BAP representation of plot specific model from \citet{grace2016integrative}.}
\end{subfigure}
~
\begin{subfigure}[t]{1\textwidth}\centering
	\begin{tikzpicture}[->,>=triangle 45,shorten >=1pt,
	auto,
	main node/.style={rounded corners,inner sep=0pt,fill=gray!20,draw,font=\sffamily,
	minimum width = 1.5cm, minimum height = .5cm, scale=0.95}]

	\node[main node] (Pprod) {Pl Prod};
	\node[main node] (Pbio) [right = .95cm of Pprod]  {Pl Bio};
	\node[main node] (Shade) [right = .95cm of Pbio]  {Pl Shade};
	\node[main node] (Prich) [right = .95cm of Shade]  {Pl Rich};
	\node[main node] (Sprod) [above = .95cm of Pprod]  {St Prod};
	\node[main node] (Sbio) [above = .95cm of Pbio]  {St Bio};
	\node[main node] (Srich) [above = .95cm of Prich]  {St Rich};
	\node[main node] (Suit) [right = .95cm of Prich]  {Pl Suit};

	\path[color=black!20!blue,style={->}]
	(Suit) edge node {} (Srich)
	(Pprod) edge node {} (Srich)
	(Pbio) edge node {} (Shade)
	(Sbio) edge node {} (Pbio)
	(Sbio) edge node {} (Prich)
	(Srich) edge node {} (Prich)
	(Sprod) edge node {} (Pprod)
	(Sprod) edge node {} (Prich)
	(Sprod) edge node {} (Shade)
	(Shade) edge node {} (Prich)
	(Pbio) edge node {} (Pprod)
	;

	\path[color=black!20!red,style={<->}]
	(Suit) edge[bend left = 20] node {} (Pprod)
	(Suit) edge[bend right = 5] node {} (Sbio)
	(Suit) edge node {} (Sprod)
	(Srich) edge[bend right = 20] node {} (Sprod)
	(Srich) edge node {} (Sbio)
	(Sbio) edge node {} (Pprod)
	(Sprod) edge node {} (Sbio)
	;

	\end{tikzpicture}
		\caption[Data example: model discovered by BANG]{\label{fig:discovered}Discovered model (BANG).}
\end{subfigure}
\end{figure}

\begin{figure}[htb]
	\centering
	\includegraphics[scale=.45]{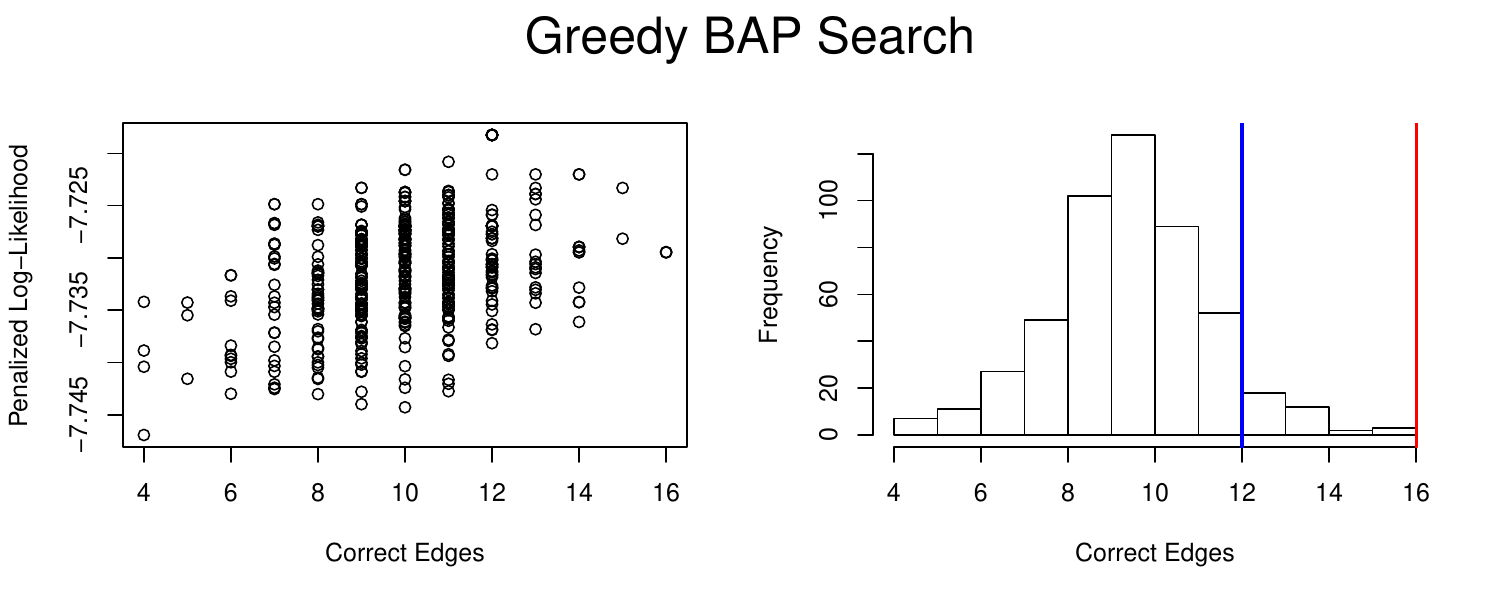}
	\caption{\label{fig:gbsRes}The left panel shows the score and number of correct edges for each of the 500 random initializations of GBS on the \citet{grace2016integrative} ecology data. The estimated graphs with the highest score has 12 correct edges. The right panel shows a histogram of the number of correct edges from the 500 random restarts. The blue line represents the graph with the highest score and the red line represents the number of correct edges for the BANG procedure.}
\end{figure}

\section{Discussion}
Borrowing intuition from the LiNGAM line of work \citep{shimizu2006lingam}, we show that when a SEM corresponds to a BAP and the errors are non-Gaussian, one can identify the exact causal structure from observational data. We propose the BANG algorithm and show that it consistently identifies the graph. This extends previous work on BAPs by \citet{nowzohour2015structure} by identifying an exact graph rather than a larger equivalence class. In addition, this extends the work on non-Gaussian SEMs with confounding by not requiring advance knowledge of the number of latent variables, not requiring the effect of confounders to be linear, or provably recovering a larger class of graphs. Finally, we also show that in the presence of bows, our proposed procedure is ``conservative'' in certifying causal relationships and explicitly characterize the returned output in the population setting.

Since the number of independence tests considered is a polynomial of the number of variables, under additional assumptions, future work might investigate conditions under which the graph might also be consistently recovered in a sparse high dimensional setting where the number of variables is larger than the number of samples. Theoretical results may be straightforward; however, considering the results in Section~\ref{sec:numericalResults} where very large sample sizes are needed for recovery with high probability, this may require significant methodological improvements. One such improvement is a pre-screening procedure. \citet{loh2014inverse} show for DAGs, even with non-Gaussian errors, the precision matrix encodes causal structure. A similar statement can be shown for BAPs, where a non-zero entry in the precision implies that two nodes are in the same \emph{mixed component}---roughly a set of nodes which are connected by bidirected edges plus the parents of those nodes; see \citet{tian2005identify, foygel2012halftrek} for a formal definition. Thus, starting with a sparse estimate of the precision could reduce the search space and improve empirical performance.

\FloatBarrier

\section*{Acknowledgements}
This project has received funding from the European Research Council
(ERC) under the European Union’s Horizon 2020 research and innovation
programme (grant agreement No 883818) as well as from the
U.S. National Science Foundation under Grant No.~DMS 1712535. The
authors also thank Thomas S. Richardson for helpful comments on an
early copy of the manuscript.

\bibliography{bang_arxiv}

\newpage

\appendix

\section{Additional simulation results}

\FloatBarrier

\subsection{Positive Parameters}
In Figures~\ref{fig:magSimsPlus} and \ref{fig:bapSimsPlus}, we generate random graphs and data using the same procedure described in Section~\ref{sec:numericalResults}; however we sample the elements of $B$ from $(.6, 1)$ (instead of $\pm (.6, 1)$) and the off diagonal elements of $\Omega $ from $(.3, .5)$ (instead of $\pm(.3, .5)$. This results in fewer faithfulness violations. BANG seems to improve quite a bit in the gamma and lognormal setting, but FCI and GBS do not improve substantially. Figure~\ref{fig:magSimsPlus} shows results when the random graphs are ancestral. Figure~\ref{fig:bapSimsPlus} shows results when the random graphs are bow-free, but may not be ancestral.

\begin{figure}[htb]
	\centering

    \includegraphics[scale=.8]{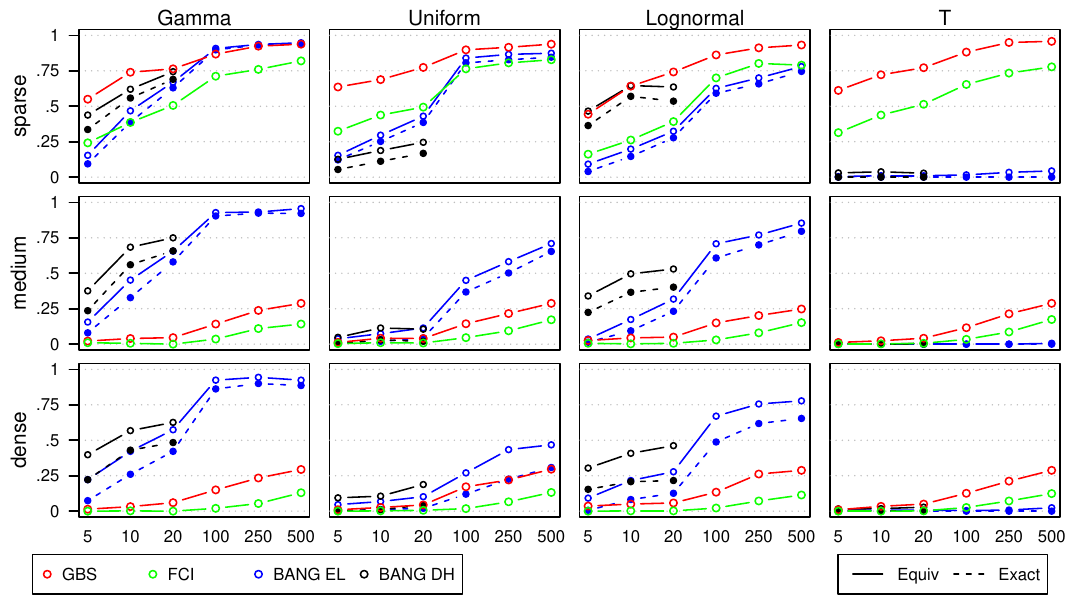}
    \caption{\label{fig:magSimsPlus}The performance across 500 replications of each method on random ancestral graphs with $p=6$ with strictly positive parameters. The solid lines indicate the proportion of times the estimated graph is in the equivalence class of the truth; the dotted lines indicates the proportion of times BANG identifies the exact graph. The horizontal axis shows sample size in hundreds.}
\end{figure}

\begin{figure}[htb]
	\centering
    \includegraphics[scale=.8]{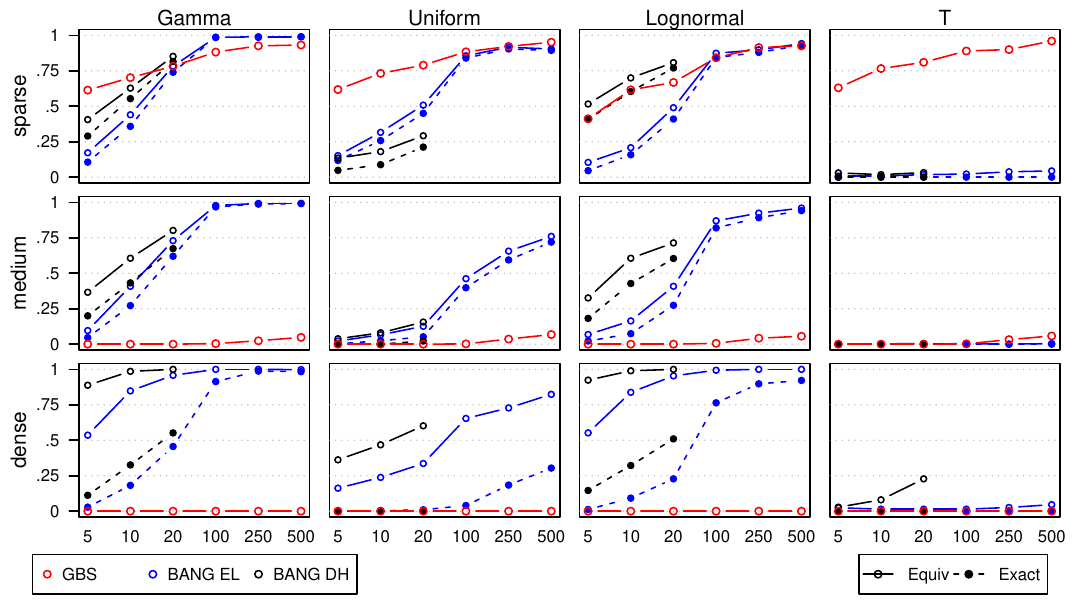}
    \caption{\label{fig:bapSimsPlus}The performance across 500 replications of each method on random bow-free graphs with $p=6$ with strictly positive parameters. The solid lines indicate the proportion of times the estimated graph is in the equivalence class of the truth; the dotted lines indicates the proportion of times BANG identifies the exact graph. The horizontal axis shows sample size in hundreds}
\end{figure}

\newpage 

\clearpage

\section{Proofs from Section~\lowercase{\ref{sec:ancestral}}}

We first restate two lemmas from \citet{tashiro2014parcelingam} which imply that the sink and source certification procedures used by ParcelLiNGAM are sound. Strictly speaking, the lemmas require linear confounders, but the results trivially generalize to our setting in which the effects of confounders are represented via correlated errors.

Also, as stated, the lemmas do not require faithfulness because they consider the full latent variable LiNGAM model where $\beta_{v,p} \neq 0$ for all $p \in \pa(v)$. Because the graph is acyclic, this implies that every non-source have at least one parent with a non-zero total effect. However, this does not hold when considering sub-models induced by marginalizing out subsets of the variables; i.e., $\beta_{v,p} \neq 0$ for all $p \in \pa(v)$ in the full model does not imply that all parents (or ancestors) in a sub-model induced by marginalization have non-zero total effect on their children (or descendants). A simple example is given in Figure~\ref{fig:plEx}. Thus, to show that ParcelLiNGAM is sound and complete, we require that the marginal direct effect of an ancestor on its descendants does not disappear for any model induced by marginalization. This is similar to the notion of parental faithfulness required in \citet{wang2020hdng} and is true for generic linear coefficients.  Hence, in the proofs of Lemma~\ref{thm:parceLCorrect} and \ref{thm:parcelNotComplete} we assume generic model parameters, and then apply Lemmas~\ref{thm:tashiroLemma1} and \ref{thm:tashiroLemma2} assuming that they hold for all sub-sets of the variables as well.

\begin{figure}[htb]
    \centering
\begin{tikzpicture}[->,>=triangle 45,shorten >=1pt,
		auto,
		main node/.style={ellipse,inner sep=0pt,fill=gray!20,draw,font=\sffamily,
			minimum width = .5cm, minimum height = .5cm}]
		
		\node[main node] (1) {1};
		\node[main node] (2) [right= 2cm of 1]  {2};
		\node[main node] (3) [right = 2cm of 2]  {3};
		
		\path[color=black!20!blue,style={->}]
		(1) edge node {$\beta_{21} = 1$} (2)
		(2) edge node {$\beta_{32} = -1$} (3)
		(1) edge[bend left = 40] node {$\beta_{31} = 1$} (3)
		;
	\end{tikzpicture}
	\caption{\label{fig:plEx}When considering the entire graph with nodes $\{1,2,3\}$, the lemmas holds since every non-source has at least one parent with a non-zero total effect. However, when only considering the sub-graph induced by $\{1, 3\}$, the marginal direct effect of $1$ on $3$ is $0$ so $3$ is a source in that sub-model despite the fact that it is a sink in the full model.}
\end{figure}
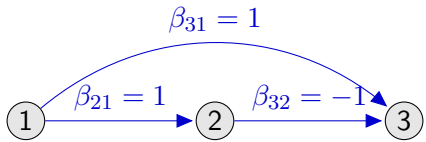

\begin{lemma}{{(Lemma 1 in \citet{tashiro2014parcelingam})}}\label{thm:tashiroLemma1}
	Assume all model assumptions of the latent variable LiNGAM are met. Denote by $r_i^{(j)}$ the population residuals when $Y_i$ are regressed onto $Y_j$. Then a variable $Y_j$ is exogenous in the sense that is has no parent observed variable or latent confounder if and only if $Y_j$ is independent of its residuals $r_i^{(j)}$ for all $i \neq j$.
\end{lemma}

\begin{lemma}{{(Lemma 2 in \citet{tashiro2014parcelingam})}} \label{thm:tashiroLemma2}
	Assume all model assumptions of the latent variable LiNGAM are met. Denote by $Y_{(-j)}$ a vector that contains all the variables other than $Y_j$. Denote by $r_j^{(-j)}$ the population residuals when $Y_j$ is regressed onto $Y_{(-j)}$.
	Then a variable $Y_j$ is a sink in the sense that is has no parent observed variable or latent confounder if and only if $Y_{(-j)}$ is independent of its residual $r_j^{(-j)}$.
\end{lemma}

\vspace{.3cm}
\subsection{Lemma~\ref{thm:parceLCorrect}}

Suppose $Y$ is generated by a recursive linear SEM that corresponds to an ancestral graph $G$. With generic model parameters and population information (i.e., the distribution of $Y$), the ordering, $\hat{\prec}$, returned by Algorithm 2 of ParcelLiNGAM is sound and complete for ancestral relationships in $G$.
	 
\begin{proof}	 
We first consider Algorithm 2 which applies Algorithm 1 to all sets in the powerset of $V$. Since Lemma~\ref{thm:tashiroLemma1} and \ref{thm:tashiroLemma2} explicitly concern the certificate used to place nodes into $\Ktop$ or $\Kbttm$, they trivially imply that any output $\Ktop$ and $\Kbttm$ method is sound.
It remains to be shown that the procedure is complete for ancestral relationships. 
By the ancestral assumption, every $v \in V$ is a sink in the set $\An(v)$ since it does not share a confounder with any ancestor, so when applying Algorithm 1 to $\An(v)$ either all of $\An(v)$ will be put into $\Ktop$ or $v$ will be put into $\Kbttm$. Regardless, it will be identified that $u \prec v$ for all $u \in \An(v) \setminus v = \an(v) $. Thus, Algorithm 2 is thus complete.

Now consider Algorithm 3 which first applies Algorithm 1 to $V$ and identifies $\Ktop$ and $\Kbttm$. It then applies Algorithm 2 to $U_{\rm res} := V \setminus \{\Ktop \cup \Kbttm\}$. $\Ktop$ and $\Kbttm$ are both total orderings so the orientation rules in Step 4 of Algorithm 1 will completely identify all ancestral relationships between any $u, v$ such that (1) $u, v \in \Ktop \cup \Kbttm$, (2) $u \in \Ktop$ and $v \in U_{\text{res}}\cup \Kbttm$, or (3) $u \in U_{\text{res}}$ and $v \in \Kbttm$. Thus, it remains to show that the remaining steps of Algorithm 3 completely discover all ancestral relationships between any pair $u, v$ such that $u, v \in U_{\text{res}}$.

By the soundness of the certification procedure, $\Ktop \cup \an(\Ktop) = \Ktop$ and $\Ktop \cap S = \emptyset$ where $S = \{v \in V: \sib(v) \neq \emptyset\}$. Similarly $\Kbttm \cup \de(\Kbttm) = \Kbttm$. It is well known that when $A \subset V$ is an ancestral set, the residuals when regressing $V\setminus A$ onto $A$ correspond to a model which can be represented by the sub-graph induced by $V\setminus A$ (e.g., \citet[Lemma 2]{chen2019causal}). In addition, removing a set which contains all of its descendants does not change the induced sub-graph; for instance see \citet[Section 5]{drton2018algebraic}. Thus, the residuals formed in Step 4, $R_{\text{res}}$, correspond to the sub-graph induced by $U_{\text{res}}$, which is also ancestral. Thus, applying the proof for Algorithm 2 implies that Step 5 of Algorithm 1 discovers all ancestral relations for $u, v \in U_{\text{res}}$. Thus Algorithm 3 is also complete.  
\end{proof}
\vspace{.3cm}

\subsection{Lemma~\ref{thm:parcelNotComplete}}

Suppose $Y$ is generated by a recursive linear SEM that corresponds to a graph $G$ which is bow-free but not ancestral. With generic parameters and population information, both Algorithm 2 and Algorithm 3 of ParcelLiNGAM will return a partial ordering which is sound, but not complete for ancestral relationships in $G$.
\begin{proof}
Algorithm 2 applies Algorithm 1 to the powerset of $V$, and we first consider the output of Algorithm 1 on a set $M \subseteq V$. 
	
Let $S = \{v \in V: \sib(v) \neq \emptyset\}$. In a graph which is not ancestral, there must exist some $u, v \in V$ such that $u \in \sib(v) \cap \an(v)$. Let 
	\begin{equation}
	Z(u, v) = \{z \,:\, \{u \cup \de(u)\} \setminus \de(v) \}.
	\end{equation}	
Now consider testing any set $M \subseteq V$ such that $v \in M$ and $M \cap Z(u,v) \neq \emptyset$. Let \[Z_{\text{top}} = \{z \in Z(u,v) \,:\, \an(z) \cap \{M \cap Z(u,v)\} = \emptyset\},\]
	so that $Z_{\text{top}}$ are nodes in $M\cap Z(u,v)$ which are not downstream of any other nodes in $M\cap Z(u,v)$. 
	Thus any $z \in Z_{\text{top}}$ will not be exogenous since it shares a latent confounder (acting through $u$) with $v$ and similarly $v$ will not be a sink. Thus, Lemma~\ref{thm:tashiroLemma1} implies that no $z\in Z(u,v)$ will be placed into $K_{\text{top}}$ which further implies no $\de(Z(u,v))$ will be placed into $K_{\text{top}}$. Similarly, Lemma~\ref{thm:tashiroLemma2} implies that $v$ will not be placed into $K_{\text{bttm}}$ which further implies no ancestor of $v$ will be put into $K_{\text{bttm}}$. Together, this implies that $M \cap Z(u,v) \subseteq U_{\rm res}$ so that running Algorithm 1 on $M$ will return inconclusive ancestral relationships between all $z \in Z(u,v)$. Since this holds for any $M \subseteq V$ such that $v \in M$ and $M \cap Z(u,v) \neq \emptyset $, Algorithm 2 will not discover that $z \prec v$ for any $z \in Z(u,v)$. Since $Z(u,v) \cap \an(v) \neq \emptyset$ Algorithm 2 is not complete.
	
	Algorithm 3 uses additional steps (Steps 2-4) before applying Algorithm 2. We show that these additional steps do not rectify the problem. First note that when applying Algorithm 1 to $V$ (Step 2), $\Ktop \subseteq V \setminus \{S\cup \de(S)\}$ and $\Kbttm \subseteq V \setminus \An(S)$. This is true because, by definition, any $s \in S$ is not exogenous since it shares a common confounder with some other $s' \in S$. Thus, no $s \in S$ will be put into $\Ktop$ and subsequently no $\de(S)$ will be put into $\Ktop$. For the same reason, no $s \in S$ will be put into $K_{\text{bttm}}$ since it is not a sink and subsequently no $\an(S)$ will be put into $\Kbttm$.

	Since $Z(u,v) \subseteq \{u \cup \de(u)\}$ and $u \in S$, then $Z(u,v) \cap K_{\text{top}} = \emptyset$. Thus, Step 4 will not remove from any $z \in Z(u,v)$ the effect of $u$ or the effect of the latent confounder shared by $u$ and $v$. Thus, as shown above, Step 5 of Algorithm 3 (applying Algorithm 2 to $R_{\rm res}$) will still fail to identify that $z \in \an(v)$ for any $z \in Z(u,v)$.  
\end{proof}	
		
\newpage

\subsection{Counterexamples: Pairwise LvLINGAM}
Pairwise LvLiNGAM will fail to discover any relationships in the simple ancestral graph shown in Figure~\ref{fig:pairwiseFail}. This is because all subsets of $V = \{1,2,3\}$ are confounded.
\begin{figure}[htb]
    \centering
\begin{tikzpicture}[->,>=triangle 45,shorten >=1pt,
		auto,
		main node/.style={ellipse,inner sep=0pt,fill=gray!20,draw,font=\sffamily,
			minimum width = .5cm, minimum height = .5cm}]
		
		\node[main node] (3)  {3};
		\node[main node] (1) [above left= .7cm of 3] {1};
		\node[main node] (2) [above right= .7cm of 3]  {2};

		\path[color=black!20!blue,style={->}]
		(1) edge node {} (3)
		(2) edge node {} (3)
		;
		
		\path[color=black!20!red,style={<->}]
		(1) edge node {} (2);
	\end{tikzpicture}
	\caption{\label{fig:pairwiseFail}An ancestral graph which Pairwise LvLiNGAM will fail to identify.}
\end{figure}
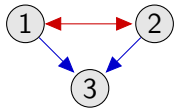

\subsection{Counterexamples: RCD}
We step through the RCD procedure when applied to the graph in Figure~\ref{fig:bapFail}. We follow the notation from Algorithm 1 in \citet{maeda2020causal}: $x_j$ is the observed data for variable $j$; at each step we consider a set $U \subseteq V$ where $|U| = l+1$ for some counter $l$; $M_i$ is the set of verified ancestors of $i$; $H_U = \bigcap_{j \in U} M_j$, and $y_j$ is the resulting residual when $x_j$ is regressed onto $H_U$. When regressing $x_j$ onto some set $H$, we let $d_{j,u.H}$ denote the population regression coefficient corresponding to $u \in H$. When performing a HSIC minimizing regression of $y_i$ onto $y_j$ for $j \in U \setminus i$ we let $\lambda$ be a potential solution and $S_i^U$ is the resulting residual.  Finally $S$ denotes a possible sink which might be certified at each step.

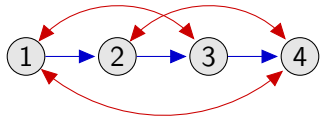
\begin{figure}[htb]
    \centering
\begin{tikzpicture}[->,>=triangle 45,shorten >=1pt,
		auto,
		main node/.style={ellipse,inner sep=0pt,fill=gray!20,draw,font=\sffamily,
			minimum width = .5cm, minimum height = .5cm}]
		
		\node[main node] (1) {1};
		\node[main node] (2) [right= .7cm of 1]  {2};
		\node[main node] (3) [right = .7cm of 2]  {3};
		\node[main node] (4) [right = .7cm of 3]  {4};
		
		\path[color=black!20!blue,style={->}]
		(1) edge node {} (2)
		(2) edge node {} (3)
		(3) edge node {} (4)
		;
		
		\path[color=black!20!red,style={<->}]
		(1) edge[bend left = 50] node {} (3)
		(2) edge[bend left = 50] node {} (4)
		(1) edge[bend right = 40] node {} (4);
	\end{tikzpicture}
	\caption{\label{fig:bapFail}A non-ancestral BAP.}
\end{figure}

\newpage

We walk through the RCD procedure assuming that line 9 is the intersection of sets; i.e., $H_U = \bigcap_{j \in U} M_j$.

\begin{itemize}
    \item $\bm{ U = \{1, 2\}: \qquad M = \{ \emptyset, \emptyset, \emptyset, \emptyset\};\, l = 1 \qquad H_U = \emptyset}$       
    \begin{itemize}
        \item $y_1 = \varepsilon_1$; $\;y_3 = \varepsilon_3 + \beta_{21}x_2$
        \item Let $i = 1$: $1 \in \pa(2)$ so there is no value of $\lambda$ such that $y_2. \independent y_1 - \lambda y_2$
        \item Let $i = 2$: Setting $\lambda = \beta_{21}$ yields 
        \begin{equation*}
            y_2 - \lambda y_1 = \varepsilon_2 \independent y_1.
        \end{equation*}
        \item $S = 2$. Update $M_2 = \{1\}$.
    \end{itemize}

    \item $\bm{ U = \{1, 3\}: \qquad M = \{ \emptyset, \{1\}, \emptyset, \emptyset\};\, l = 1 \qquad H_U = \emptyset}$       
    \begin{itemize}
        \item $1 \in \sib(3)$ so there is no updates to $M$.
    \end{itemize}  

    \item $\bm{ U = \{1, 4\}: \qquad M = \{ \emptyset, \{1\}, \emptyset, \emptyset\};\, l = 1 \qquad H_U = \emptyset}$       
    \begin{itemize}
        \item $1 \in \sib(4)$ so there is no update to $M$.
    \end{itemize} 
    
    \item $\bm{ U = \{2, 3\}: \qquad M = \{ \emptyset, \{1\}, \emptyset, \emptyset\};\, l = 1 \qquad H_U = \emptyset}$       
    \begin{itemize}
        \item $y_2 = \varepsilon_2 + \beta_{21}\varepsilon_1$; $\;y_3 = \varepsilon_3 + \beta_{32}(\varepsilon_2 + \beta_{21}\varepsilon_1)$
            \item Let $i = 2$: $2 \in \pa(3)$ so there is no value of $\lambda$ such that $y_3 \independent y_2 - \lambda y_3$.
            \item Let $i = 3$: In order for  $S_3^U \independent y_2$, it is necessary that $S_3^U$ not contain a $\varepsilon_2$ term. This implies that $\lambda$ must be $\beta_{32}$ so that $S_4^U = y_3 - \lambda y_2 = \varepsilon_3$. However, since $1 \in \sib(3)$, then
            \begin{equation*}
                y_3 - \lambda y_2 = \varepsilon_3 \not \independent y_2 = \varepsilon_2 + \beta_{21}\varepsilon_1.
            \end{equation*}
            so there is no update to $M$.
    \end{itemize}
    
    \item $\bm{ U = \{2, 4\}: \qquad M = \{ \emptyset, \{1\}, \emptyset, \emptyset\};\, l = 1 \qquad H_U = \emptyset}$       
    \begin{itemize}
        \item $2 \in \sib(4)$ so there is no update.  
    \end{itemize}
    
    \item $\bm{ U = \{3, 4\}: \qquad M = \{ \emptyset, \{1\}, \emptyset, \emptyset\};\, l = 1 \qquad H_U = \emptyset}$       
    \begin{itemize}
        \item $y_3 = \varepsilon_3 + \beta_{32}x_2$; $\;y_4 = \varepsilon_4 + \beta_{43}x_3$
            \item Let $i = 3$: $3 \in \pa(4)$ so there is no value of $\lambda$ such that $y_4 \independent y_3 - \lambda y_4$
            \item Let $i = 4$: For $S_4^U \independent y_3$, it is necessary that $S_4^U$ not contain a $\varepsilon_3$ term. This implies that $\lambda$ must be $\beta_{43}$, so that $S_4^U = y_4 - \lambda y_3 = \varepsilon_4$. However, since $1 \in \sib(4)$, then
            \begin{equation*}
                S_4^U = y_4 - \lambda y_3 = \varepsilon_4 \not \independent \varepsilon_3 + \beta_{32}(\varepsilon_2 + \beta_{21}\varepsilon_1)
            \end{equation*}
            so there is no update to $M$.
    \end{itemize}
\end{itemize}

This is all subsets of size $2$, but since an update has occurred, $l$ remains $1$, and the procedure will cycle through all subsets of size $2$ again. However, since $M_2$ is the only non-empty set, $H_U$ is the same for all pairs, so the outcomes are the same as before the second time through. $M$ is not updated so $l = 2$ and we test all sets of size $3$.
    
\begin{itemize}
    \item $\bm{ U = \{1,2,3\}: \qquad M = \{ \emptyset, \{1\}, \emptyset, \emptyset\};\, l = 1 \qquad H_U = \emptyset}$
    \begin{itemize}
        \item Let $i = 1$: $1 \in \pa(2)$ so no update can be made.  
        \item Let $i = 2$: $M_2 \cap U \neq \emptyset$ so it is not tested.
        \item Let $i = 3$: $y_1 = \varepsilon_1$ and $y_2 = x_2$ both do not contain a $\varepsilon_3$ term, so $S_3^U$ must contain a $\varepsilon_3$ term. Since $1 \in \sib(3)$, then $S_3^U$ cannot be independent of $y_1$, so no update is made.
        \item No update is made.
        
    \end{itemize} 
    
    \item $\bm{ U = \{1,2,4\}: \qquad M = \{ \emptyset, \{1\}, \emptyset, \emptyset\};\, l = 1 \qquad H_U = \emptyset}$
    \begin{itemize}
        \item Let $i = 1$: $1 \in \pa(2)$ so no update can be made.  
        \item Let $i = 2$: $M_2 \cap U \neq \emptyset$ so it is not tested.
        \item Let $i = 4$: $y_1 = \varepsilon_1$ and $y_2 = x_2$ both do not contain a $\varepsilon_4$ term, so $S_4^U$ must contain a $\varepsilon_4$ term. Since $1 \in \sib(4)$, then $S_4^U$ cannot be independent of $y_1$, so no update is made.
        \item No update is made.
    \end{itemize} 
    
    \item $\bm{ U = \{1,3,4\}: \qquad M = \{ \emptyset, \{1\}, \emptyset, \emptyset\};\, l = 1 \qquad H_U = \emptyset}$           
    \begin{itemize}
        \item Let $i = 1$: $1 \in \pa(2)$ so no update can be made.  
        \item Let $i = 3$: $3 \in \pa(4)$ so no update can be made.
        \item Let $i = 4$: $y_1 = \varepsilon_1$ and $y_3 = x_3$ both do not contain a $\varepsilon_4$ term, so $S_4^U$ must contain a $\varepsilon_4$ term. Since $1 \in \sib(4)$, then $S_4^U$ cannot be independent of $y_1$, so no update is made.
        \item No update is made.
    \end{itemize} 
    
    \item $\bm{ U = \{2,3,4\}: \qquad M = \{ \emptyset, \{1\}, \emptyset, \emptyset\};\, l = 1 \qquad H_U = \emptyset}$
    \begin{itemize}
        \item Let $i = 2$: $2 \in \pa(3)$ so no update can be made.  
        \item Let $i = 3$: $3 \in \pa(4)$ so no update can be made.
        \item Let $i = 4$: $y_2 = x_2$ and $y_3 = x_3$ both do not contain a $\varepsilon_4$ term, so $S_4^U$ must contain a $\varepsilon_4$ term. Since $2 \in \sib(4)$, then $S_4^U$ cannot be independent of $y_2$, so no update is made.
        \item No update is made.
    \end{itemize} 
\end{itemize}

Since no updates have been made, $l = 3$.

\begin{itemize}
    \item $\bm{ U = \{1, 2,3,4\}: \qquad M = \{ \emptyset, \{1\}, \emptyset, \emptyset\};\, l = 1 \qquad H_U = \emptyset}$
        \begin{itemize}
        \item Let $i = 1$: $1 \in \pa(2)$ so no update can be made.  
        \item Let $i = 2$: $2 \in \pa(3)$ so no update can be made.  
        \item Let $i = 3$: $3 \in \pa(4)$ so no update can be made.
        \item Let $i = 4$: $y_1 = \varepsilon_1$, $y_2 = x_2$, and $y_3 = x_3$ both do not contain a $\varepsilon_4$ term, so $S_4^U$ must contain a $\varepsilon_4$ term. Since $\sib(4) = \{1,2\}$, then $S_4^U$ cannot be independent of $y_1$ or $y_2$, so no update is made.
        \item No update is made.
    \end{itemize} 
\end{itemize}

The algorithm will terminate and has only discovered $1 \rightarrow 2$.

\newpage

\section{Proofs from Section~\lowercase{\ref{sec:bap}}}
\subsection{Lemma~\ref{thm:neccIndGeneral}} 

For $v \in V$ and sets $C \subseteq A \subseteq V \setminus \{v\}$, Suppose $D \in \mathbb{R}^{p \times p}$ such that $D_{ij} \neq 0$ only if $j \in \an(i)$. Then, for generic $B$ and error moments, if $\delta_v(C, A, \Sigma, D) \neq \tilde B(C\cup\{v\})_{v,C}$,
then $\E(\gamma_c^{K-1}\gamma_v) \neq 0$ for some $c \in C$.

\begin{proof}
Since $\E\left(\gamma_{c}^{K-1}\gamma_v\right)$ is a rational function of the model parameters, by \citet[Lemma 1]{okamoto1973distinctness}, showing that the quantity is non-zero for some parameters is sufficient for showing that it vanishes only over a null set. Without loss of generality, let $C$ be ordered such that $C = \{c_1, \ldots, c_{|C|}\}$ where $c_i$ is not a descendant of $c_j$ for any $j < i$. Note that
	\begin{equation}
	\begin{aligned}
	\gamma_v &= \varepsilon_v + \sum_{a \in \an(v)} \pi_{va}\varepsilon_a - \sum_{c \in C} \delta_{vc}Y_c\\
	 &=\varepsilon_v + \sum_{a \in \an(v)} \pi_{va}\varepsilon_a - \sum_{c \in C} \delta_{vc}(\varepsilon_c + \sum_{a \in \an(c)} \pi_{ca}\varepsilon_a).
	\end{aligned}
	\end{equation}
	Suppose $i$ is the minimum index for which $\delta_{c_i} \neq \tilde B_{v, c_i}$ so that $\delta_{c_j} = \tilde B_{v,c_j}$ for all $j < i$. Then, the coefficient of $\varepsilon_{c_i}$ in $Y_v - \sum_{j < i}\delta_{v,c_j}Y_{c_j}$ is

	\begin{align*}
	\pi_{v,c_i} - \sum_{j < i}\delta_{v, c_j} \pi_{c_j, c_i} &= \pi_{v,c_i} - \sum_{j < i}\tilde \beta_{v, c_j} \pi_{c_j, c_i}\\
	&= \sum_{l \in \mathcal{L}_{v,c_i}} W(l) -\sum_{j < i} \left[\left( \sum_{l \in \mathcal{L}^{(c_j)}_{v,c_j}(C)} W(l)\right) \left(\sum_{l \in \mathcal{L}_{c_j, c_i}} W(l)\right)\right] \\
	&= \sum_{l \in \mathcal{L}_{v,c_i}} W(l) -\sum_{j < i} \left[\sum_{ l \in \mathcal{L}^{(c_j)}_{v,c_i}(C)} W(l)\right]\\ 	\label{eq:marginalDirectDecomp} \numberthis
	&= \sum_{l \in \mathcal{L}^{(c_i)}_{v,c_i}} W(l) = \tilde B(C)_{v,c_i}.
	\end{align*}

	For all $j > i$, $c_j$ is not a descendant of $c_i$ so $Y_{c_j}$ does not include any terms of $\varepsilon_{c_i}$. By assumption, $\delta_{c_i} \neq \tilde B_{v, c_i}$, so let $\delta_{c_i} =\tilde B_{v, c_i} - \alpha$ for $\alpha \neq 0$ so that
	\begin{equation}
	\gamma_v = \alpha \varepsilon_{c_i} + \eta \qquad \text{ and } \qquad \gamma_{c_i} = \varepsilon_{c_i} + \zeta,
	\end{equation}
	where $\eta$ and $\zeta$ do not contain $\varepsilon_{c_i}$.
	Then,

	\begin{align*}
	\E\left(\gamma_{c_i}^{K-1}\gamma_v\right) & = \E\left(\left[\varepsilon_{c_i}  + \zeta\right]^{K-1}\left[ \alpha \varepsilon_{c_i} + \eta\right] \right) \\
	& = \E\left(\left[\varepsilon_{c_i}^{K-1}  + \sum_{k = 0}^{K-2} \varepsilon_{c_i}^{k}\zeta^{K-1 -k}\right]\left[ \alpha \varepsilon_{c_i} + \eta\right] \right) \\
	& = \alpha\E\left(\varepsilon_{c_i}^{K}\right) + \E\left(\varepsilon_{c_i}^{K-1}\eta\right) +  \E\left(\left[\sum_{k = 0}^{K-2} \varepsilon_{c_i}^{k}\zeta^{K-1 -k}\right]\left[ \alpha \varepsilon_{c_i} + \eta\right] \right).
	\end{align*}
Since the last two terms do not involve $\E(\varepsilon_{c_i}^K)$, we can always select some $\E(\varepsilon_{c_i}^K)$ such that
	\begin{align*}
	\E\left(\varepsilon_{c_i}^{K}\right) \neq -\frac{\left(\E\left(\varepsilon_{c_i}^{K-1}\eta\right) +  \E\left(\left[\sum_{k = 0}^{K-2} \varepsilon_{c_i}^{k}\zeta^{K-1 -k}\right]\left[ \alpha \varepsilon_{c_i} + \eta\right] \right)\right)} {\alpha}
	\end{align*}
which	ensures that $\E\left(\gamma_{c_i}^{K-1}\gamma_v\right) \neq 0$.
\end{proof}

\subsection{Proof of Lemma~\ref{thm:neccIndDes}}

Consider $v \in V$ and set $C \subseteq V\setminus \{v\}$. Let $D \in \mathbb{R}^{p \times p}$ such that $D_{ij} \neq 0$ only if $j \in \an(i)$.
Suppose $C \not \subseteq \an(v)$, but for generic $B$ and error moments, 
$\E(\gamma_c(D)^{K-1}\gamma_v(C, S,D)) = 0$
for all $c \in C$. Then for $C_1 = C\cap \an(v)\setminus \sib(v)$,
\begin{equation} \E(\gamma_c(D)^{K-1}\gamma_v(C_1, S,D)) = 0\end{equation} for all $c \in C$.

\begin{proof}
For convenience, let $A = \psAn(C)$, $A_1 = \psAn(C_1)$, $A_2 = A \setminus A_1$, and $\Lambda = I-D$ and $\Pi = (I- D)^{-1}$. Note that $A_2 \cap \de(A_1) = \emptyset$; this implies $D_{A_1, A_2} = 0$ and $\left[(I- D_{A, A})^{-1}\right]_{A_1, A_2} = 0$. So that

\begin{align*}
(I- D)_{{C_1}, A}S_{A,C} & = \begin{bmatrix} \Lambda_{{C_1}, A_1} & \Lambda_{{C_1}, A_2} \end{bmatrix}
\begin{bmatrix} S_{A_1, C_1} \\
S_{A_2, C_1} \\
 \end{bmatrix}
= \begin{bmatrix} \Lambda_{{C_1}, A_1} & 0 \end{bmatrix}
\begin{bmatrix} S_{A_1, C_1} \\
S_{A_2, C_1} \\
 \end{bmatrix}\\
& = \Lambda_{{C_1}, A_1} S_{A_1, C_1},
\end{align*}
and
\begin{align*}
(I-D)_{C_1, A} \Sigma_{A,v} &= \begin{bmatrix} \Lambda_{C_1, A_1} \Lambda_{C_1, A_2}\end{bmatrix} \begin{bmatrix} \Sigma_{A_1,v}\\
\Sigma_{A_2,v}
\end{bmatrix}
 = \begin{bmatrix} \Lambda_{C_1, A_1} 0\end{bmatrix} \begin{bmatrix} \Sigma_{A_1,v}\\
\Sigma_{A_2,v}
\end{bmatrix}\\
&= (I-D)_{C_1, A_1} \Sigma_{A_1,v}.
\end{align*}
Thus,
\begin{align*}
\delta_{v}(C_1, A, S,D) &= \left[(I- D)_{C_1, A}S_{A,C_1}\right]^{-1} (I-D)_{C_1, A} \Sigma_{A,v}\\
&= \left[(I- D)_{C_1,{A_1}}S_{{A_1},{C_1}}\right]^{-1} (I-D)_{C_1, A_1} \Sigma_{A_1,v}\\
& = \delta_{v}(C_1, A_1, S,D).
\end{align*}
By Lemma~\ref{thm:neccIndGeneral}, for generic $B$ and error moments, if \[\E(\gamma_c(D)^{K-1}\gamma_v(C, S,D)) = 0,\] then for every $q \not \in C_1$, $\delta_{vq}(C, A, S,D) = 0$.
\begin{align*}
\gamma_v(C, S,D) &= Y_v - Y_C\delta_{v}(C, A, S,D)\\
&=  Y_v - Y_{C_1}(C_1, A_1, S, D)\\
& =  \gamma_v(C_1, A_1, S,D).
\end{align*}
So if for all $c \in C$,
\begin{equation}
\begin{aligned}
\E\left(\gamma_c(D)^{K-1}\gamma_v(C,\Sigma,D)\right) = 0,
\end{aligned}
\end{equation}
then for all $c \in C$
\begin{equation}
\begin{aligned}
\E\left(\gamma_c(D)^{K-1}\gamma_v(C_1,S,D)\right) =0
\end{aligned}.
\end{equation}
\end{proof}

\subsection{Proof of Lemma~\ref{thm:noMissPruneParents}}
Suppose $D = H_\mathcal{C}(B)$ for some $H_\mathcal{C} \in
\mathcal{D}$ and for some $v \in V$, we have that\\ $\E(\gamma_c(D)^{K-1}\gamma_v(\pspa(v), S,D)) = 0$ for all $c \in \pspa(v)$.
If $q \in [\pa(v)\setminus \pspa(v)] \cup \sib(v)$, then for generic $B$ and error moments
\[\E \left(\gamma_q(D)^{K-1}\gamma_v(D)\right) \neq 0.\]

\begin{proof}
For notational convenience, let $C = \pspa(v)$. First consider $q \in \pa(v)\setminus \pspa(v)$. $\E(\gamma_c(D)^{K-1}\gamma_v(C, \psAn(C), S,D)) = 0$ for all $c \in \pspa(v)$ implies that
\begin{equation}
\begin{aligned}
\gamma_v(D) &= Y_v - Y_{\pspa(v)}(D_{v, \pspa(v)})^T\\
&= \left(\pi_{v,q} - \sum_{c \in C}\tilde B(C\cup\{v\})_{v,c}\pi_{c,q}\right) \epsilon_q + \eta\\
&= \left(\pi_{v,q} - \sum_{c \in C \cap \de(q)}\tilde B(C\cup\{v\})_{v,c}\pi_{c,q}\right) \epsilon_q + \eta\\
&= \alpha \epsilon_q + \eta
\end{aligned}
\end{equation}
where $\eta$ does not involve $\epsilon_q$. For any $c \in \de(q)$, $\tilde B(C\cup\{ v, q\})_{v,C} = \tilde B(C\cup\{v\})_{v,C}$ because there are no paths from $c$ to $v$ which pass through $q$, so marginalizing $q$ does not change the marginal direct effect. Thus, as shown in Lemma~\ref{thm:neccIndGeneral},
\begin{equation}
\begin{aligned}
\alpha &= \pi_{q,v} - \sum_{c \in C \cap \de(q)}\tilde B(C\cup\{q, v\})_{v,C}\pi_{c,q}\\
 &= \tilde B(C\cup\{q, v\})_{v,q}.
\end{aligned}
\end{equation}
The set of points, $B$ such that $q\in \pa(v)$, but the marginal direct effect $\tilde B(C\cup\{q, v\})_{vq} = 0$ have Lebesgue measure 0, so by the same argument as Lemma~\ref{thm:neccIndGeneral} when $\alpha \neq 0$, for generic error moments, $\E(\gamma_q^{K-1} \gamma_v) \neq 0$.

Now consider $q \in \sib(v)$. Since $\pspa(v) \subseteq \an(v)$ for all $v \in V$, then $\gamma_v = \epsilon_v + \eta$ where $\eta$ does not involve $\varepsilon_v$ and $\gamma_q = \epsilon_q + \zeta$ where $\zeta$ does not involve $\epsilon_q$. Then, using the same argument as the previous lemmas, selecting
\begin{equation}
\E(\epsilon_q^{K-1} \epsilon_v) \neq -\E\left(\sum_{t = 0}^{K-2}\binom{K-1}{t}\varepsilon_q^t\zeta^{K-1-t} (\epsilon_v + \eta) + \epsilon_q^{K-1}\eta\right)
\end{equation}
ensures that $\E(\gamma_q^{K-1} \gamma_v) \neq 0$
\end{proof}

\subsection{Proof of Lemma~\ref{thm:neccIndSib}}

Consider $v \in V$ and sets $A, C$ such that $C \subseteq A \subseteq V \setminus \{v\}$. Suppose $D = H_\mathcal{C}(B)$ for some $H_\mathcal{C} \in \mathcal{D}$ and $S = \Sigma$, but $C \cap \sib(v) \neq \emptyset$. Then for generic $B$ and error moments, there exists some $q \in C$ such that $\E\left(\gamma_q^{K-1}\gamma_v\right) \neq 0$.

\begin{proof}
	We again appeal to \citet[Lemma 1]{okamoto1973distinctness}, and show that the quantity is non-zero for generic $B$ and the error moments by constructing a single point (of B and the error moments) at which the quantity of interest is non-zero. In particular, select $q \in C \cap \sib(v)$. We then represent $\gamma_v$ as
	\begin{equation}
	\begin{aligned}
	\gamma_v &= \varepsilon_v + \sum_{a \in \an(v)} \pi_{va}\varepsilon_a - \sum_{c \in C} \delta_{vc} \sum_{z \in \An(c)} \pi_{cz} \varepsilon_z\\
	& = \alpha \varepsilon_q + \eta,
	\end{aligned}
	\end{equation}
	where
	\begin{align*}
	\alpha &= \pi_{vq} + \sum_{c \in C}\delta_{vc}\pi_{cq} \\
	\eta & = (1 - \sum_{c \in C}\delta_{vc} \pi_{cv})\epsilon_v + \sum_{a \in \an(v) \setminus q} \pi_{va}\varepsilon_a - \sum_{c \in C} \delta_{vc}\sum_{z \in \An(c) \setminus q} \pi_{cz} \varepsilon_z
	\end{align*}
and $\delta_{vc}$ is the $c$-th element of $\delta_v$ from \eqref{eq:deltaDef}. Similarly, we represent $\gamma_q$
	\begin{equation}
	\begin{aligned}
	\gamma_q &= \varepsilon_q + \sum_{a \in \an(q)} \pi_{va}\varepsilon_a - \sum_{s \in \pspa(q)} d_{qs} \sum_{t \in \An(s)} \pi_{st} \varepsilon_t\\
	& = \varepsilon_q + \zeta
	\end{aligned}
	\end{equation}
	where $\zeta$ does not involve $\varepsilon_q$. The coefficient on $\varepsilon_q$ is $1$ since $D = H_\mathcal{C}(B)$ implies that $d_{qs} \neq 0$ only if $s \in \an(q)$. For $S = \Sigma$ and any $H_\mathcal{C} \in \textbf{D}$, $\alpha$ is a rational function of $B$ and $\Omega$ because both $\Pi$ and $\delta$ only involve matrix inversions and multiplications of $D$ and $S$ which in turn are rational functions of $B$ and $\Omega$. We now show that for some point $B$ and $\Omega$, $\alpha \neq 0$. In particular, let $B = 0$ and $\omega_{qv} \neq 0$, but $\omega_{ij} = 0$ for all other $i \neq j$. At this point, $\pi_{vq} = \pi_{cq} = 0$ for all $c \in C \setminus q$ so that
	\begin{equation}
	\begin{aligned}
	\alpha = \delta_{vq}.
	\end{aligned}
	\end{equation}
	$B = 0$ implies that $D = 0$ for all $D \in \mathcal{D}$ and $S_{C,C} = \Omega_{C,C}$. In addition, $S_{C\setminus q,v} = 0$ since all treks between nodes in $C$ or between treks $C \setminus \{q\}$ and $v$ have path weights of 0. However, there is a single trek between $q$ and $v$, namely the bidirected edge, so $S_{qv} = \omega_{qv}$. Then,
	\begin{equation}
	\begin{aligned}
	\alpha = \delta_{vC} = \left[S_{C,C}\right]^{-1} S_{C,v} = \frac{\omega_{qv}}{\omega_{qq}} \neq 0.
	\end{aligned}
	\end{equation}

	Thus, for generic choice of $B$ and $\Omega$, $\alpha \neq 0$. Now, we finally examine the quantity of interest, which is a rational function of the error moments and $B$, and play the same game as before. In particular,

	\begin{align*}
	\E\left(\gamma_q^{K-1}\gamma_v\right) & = \E\left(\left[\varepsilon_q  + \zeta\right]^{K-1}\left[ \alpha \varepsilon_q + \eta\right] \right) \\
	& = \E\left(\left[\varepsilon_q^{K-1}  + \sum_{k = 0}^{K-2} \varepsilon_q^{k}\zeta^{K-1 -k}\right]\left[ \alpha \varepsilon_q + \eta\right] \right) \\
	& = \alpha\E\left(\varepsilon_q^{K}\right) + \E\left(\varepsilon_q^{K-1}\eta\right) +  \E\left(\left[\sum_{k = 0}^{K-2} \varepsilon_q^{k}\zeta^{K-1 -k}\right]\left[ \alpha \varepsilon_q + \eta\right] \right). \\
	\end{align*}

The last two terms do not involve $\E(\varepsilon_q^K)$ so we select $\E\left(\varepsilon_q^{K}\right)$ such that
	\begin{equation}
	\begin{aligned}
	\E\left(\varepsilon_q^{K}\right) \neq -\frac{\left(\E\left(\varepsilon_q^{K-1}\eta\right) +  \E\left(\left[\sum_{k = 0}^{K-2} \varepsilon_q^{k}\zeta^{K-1 -k}\right]\left[ \alpha \varepsilon_q + \eta\right] \right)\right)} {\alpha}
	\end{aligned}
	\end{equation}
to ensure that $\E\left(\gamma_q^{K-1}\gamma_v\right) \neq 0$. Thus, there exists some point such that $\E\left(\gamma_q^{K-1}\gamma_v\right) \neq 0$. This implies there is a null set of $B$ and error moments which we must avoid for each $H_\mathcal{C} \in \mathcal{D}$, but since $|\mathcal{D}|$ is finite, then the union of these null sets is again a null set.
\end{proof}

\subsection{Proof of Corollary~\ref{thm:noMissPruneParentsFinal}}

	Suppose $D = B$. For $v \in V$ and generic $B$ and error moments, suppose $\pa(v) \subseteq C \subseteq \an(v)\setminus \sib(v)$ and $\E(\gamma_c(D)^{K-1}\gamma_v(C, \Sigma,D)) = 0$ for all $c \in C$. If $q \in C \setminus \pa(v)$, then for all $c \in C$
	\begin{equation}
	\E(\gamma_c(D)^{K-1}\gamma_v(C \setminus \{q\}, \Sigma,D)) = 0.
	\end{equation}
	If $q \in \pa(v)$, then there exists some $c \in C$ such that
	\begin{equation}
	    \E(\gamma_c(D)^{K-1}\gamma_v(C \setminus \{q\}, \Sigma,D)) \neq 0.
	\end{equation}
	
\begin{proof}
Suppose $q \in C \setminus \pa(v)$ and without loss of generality, assume that $q$ is the last element in $C$. Then, Lemma~\ref{thm:debiasedEst} implies that  
\begin{equation}
\begin{aligned}
\delta_{v}(C, \psAn(C), \Sigma,D) &= B_{v,(C \setminus \{q\}, q)}
 = \begin{bmatrix} B_{v,(C \setminus \{q\})} & 0 \end{bmatrix}\\
 &= \begin{bmatrix} \delta_{v}(C\setminus \{q\}, \Sigma,D) & 0 \end{bmatrix}
\end{aligned}
\end{equation}
so that for all $c \in C$
\begin{equation}
\begin{aligned}
\E(\gamma_c(D)^{K-1}&\gamma_v(C \setminus \{q\}, \Sigma,D)) =  \E(\gamma_q(D)^{K-1}\gamma_v(C, \Sigma,D)) = 0.
\end{aligned}
\end{equation}
Now consider the second statement when $q \in \pa(v)$. If $\E(\gamma_c(D)^{K-1}\gamma_v(C \setminus \{q\}, \Sigma,D)) \neq 0$ for some $c \in C \setminus \{q\}$ then the statement trivially holds. Thus, it remains to be shown that $\E(\gamma_q(D)^{K-1}\gamma_v(C \setminus \{q\}, \Sigma,D)) \neq 0$ when $\E(\gamma_c(D)^{K-1}\gamma_v(C \setminus \{q\}, \Sigma,D)) = 0$ for all $c \in C \setminus \{q\}$. This is directly implied by Lemma~\ref{thm:noMissPruneParents}
\end{proof}

\newpage

\section{Proofs from Section~\ref{sec:modelMiss}}
\subsection{Proof of Corollary~\ref{thm:neccIndGeneralMiss}}
  Let $v \in V$, and consider any set $C \subseteq A \subseteq V
  \setminus \{v\}$. Suppose $D \in \mathbb{R}^{p \times p}$ with $D_{s,t} \neq 0$ only if $t \in \bar{\an}(s)$. Then, for generic $B$ and error moments, if $\delta_v(C, A, S, D) \neq \check B(C\cup v)_{v,C}$, then $\E(\gamma_c^{K-1}(D)\gamma_v(C, S, D)) \neq 0$ for some $c \in C$.
\begin{proof}
The proof exactly follows that of Lemma~\ref{thm:neccIndGeneral}; however, some of the quantities in $G$ are replaced with the corresponding quantities in $\bar G$.

Without loss of generality, let $C$ be ordered such that $C = \{c_1, \ldots, c_{|C|}\}$ where $c_i \not \in \de(c_j)$ (note that this is in the original graph $G$) for any $j < i$. Note that
	\begin{equation}
	\begin{aligned}
	\gamma_v(C,S,D) &= \bar \varepsilon_v + \sum_{a \in \bar{\an}(v)} \bar\pi_{va}\bar \varepsilon_a - \sum_{c \in C} \delta_{vc}Y_c\\
	 &=\bar \varepsilon_v + \sum_{a \in \bar{\an}(v)} \bar\pi_{va}\bar\varepsilon_a - \sum_{c \in C} \delta_{vc}(\bar\varepsilon_c + \sum_{a \in \bar{\an}(c)} \bar \pi_{ca} \bar \varepsilon_a).
	\end{aligned}
	\end{equation}
	Suppose $i$ is the minimum index for which $\delta_{c_i} \neq \check B_{v, c_i}$ so that $\delta_{c_j} = \tilde B_{v,c_j}$ for all $j < i$. Then, the coefficient of $\varepsilon_{c_i}$ in $Y_v - \sum_{j < i}\delta_{v,c_j}Y_{c_j}$ is

	\begin{align*}
	\bar \pi_{v,c_i} - \sum_{j < i}\delta_{v, c_j} \bar \pi_{c_j, c_i} &= \bar \pi_{v,c_i} - \sum_{j < i}\check \beta_{v, c_j}\bar \pi_{c_j, c_i}\\
	&= \sum_{l \in \mathcal{\bar L}_{v,c_i}} W(l) -\sum_{j < i} \left[\left( \sum_{l \in \mathcal{\bar L}^{(c_j)}_{v,c_j}(C)} W(l)\right) \left(\sum_{l \in \mathcal{\bar L}_{c_j, c_i}} W(l)\right)\right] \\
	&= \sum_{l \in \mathcal{\bar L}_{v,c_i}} W(l) -\sum_{j < i} \left[\sum_{ l \in \mathcal{\bar L}^{(c_j)}_{v,c_i}(C)} W(l)\right]\\ 
	&= \sum_{l \in \mathcal{\bar L}^{(c_i)}_{v,c_i}} W(l) = \check B(C \cup v)_{v,c_i}.
	\end{align*}

	For all $j > i$, $c_j \not \in \de(c_i)$ so $Y_{c_j}$ does not include any terms of $\varepsilon_{c_i}$ (note this is not $\bar \varepsilon_{c_i}$). By assumption, $\delta_{c_i} \neq \check B_{v, c_i}$, so let $\delta_{c_i} =\check B_{v, c_i} - \alpha$ for $\alpha \neq 0$ so that
	\begin{equation}
	\gamma_v = \alpha \varepsilon_{c_i} + \eta \qquad \text{ and } \qquad \gamma_{c_i} = \varepsilon_{c_i} + \zeta,
	\end{equation}
	where $\eta$ and $\zeta$ do not contain $\varepsilon_{c_i}$. Then,

	\begin{align*}
	\E\left(\gamma_{c_i}^{K-1}(D)\gamma_v(C,S,D)\right) & = \E\left(\left[\varepsilon_{c_i}  + \zeta\right]^{K-1}\left[ \alpha \varepsilon_{c_i} + \eta\right] \right) \\
	& = \E\left(\left[\varepsilon_{c_i}^{K-1}  + \sum_{k = 0}^{K-2} \varepsilon_{c_i}^{k}\zeta^{K-1 -k}\right]\left[ \alpha \varepsilon_{c_i} + \eta\right] \right) \\
	& = \alpha\E\left(\varepsilon_{c_i}^{K}\right) + \E\left(\varepsilon_{c_i}^{K-1}\eta\right) +  \E\left(\left[\sum_{k = 0}^{K-2} \varepsilon_{c_i}^{k}\zeta^{K-1 -k}\right]\left[ \alpha \varepsilon_{c_i} + \eta\right] \right).
	\end{align*}
Since the last two terms do not involve $\E(\varepsilon_{c_i}^K)$, we can always select some $\E(\varepsilon_{c_i}^K)$ such that
	\begin{align*}
	\E\left(\varepsilon_{c_i}^{K}\right) \neq -\frac{\left(\E\left(\varepsilon_{c_i}^{K-1}\eta\right) +  \E\left(\left[\sum_{k = 0}^{K-2} \varepsilon_{c_i}^{k}\zeta^{K-1 -k}\right]\left[ \alpha \varepsilon_{c_i} + \eta\right] \right)\right)} {\alpha}
	\end{align*}
which	ensures that $\E\left(\gamma_{c_i}^{K-1}\gamma_v\right) \neq 0$.
\end{proof}

\subsection{Proof of Corollary~\ref{thm:neccIndDesMiss}}
Consider $v \in V$ and set $C \subseteq V\setminus \{v\}$. Let $D \in \mathbb{R}^{p \times p}$ such that $D_{s,t} \neq 0$ only if $t \in \bar{\an}(s)$.
Suppose $C \not \subseteq \bar{\an}(v)$, but that 
$\E(\gamma_c(D)^{K-1}\gamma_v(C, S,D)) = 0$
for all $c \in C$. Then for generic $B$ and error moments, $C_1 = C\cap \left[\bar{\an}(v)\setminus \bar{\sib}(v)\right]$,
\begin{equation*} \E(\gamma_c(D)^{K-1}\gamma_v(C_1, S ,D)) = 0\end{equation*} for all $c \in C$.

\begin{proof}
The proof exactly follows the proof of Lemma~\ref{thm:neccIndDes}, except replaces all quantities in $G$ with quantities in $\bar G$.
\end{proof}

\subsection{Proof of Lemma~\ref{thm:noMissPruneParentsMiss}}
Suppose $D = H_\mathcal{C}(\bar B)$ for some $H_\mathcal{C} \in \mathcal{D}$ with $\mathcal{C}=(C_s)_{s\in V}$ such that $C_s \subseteq \bar{\an}(s) \setminus \bar{\sib}(s)$ for all $s \in V$.  Let $v \in V$ be such that we have 
$
\E(\gamma_c(D)^{K-1}\gamma_v(D)) = 0$ for all $c \in \pspa(v)$.
If $q \in \big(\bar{\pa(v)}\setminus \pspa(v)\big) \cup \bar{\sib}(v)$, then for generic $B$ and error moments,  $\E \left(\gamma_q(D)^{K-1}\gamma_v(D)\right) \neq 0$.

\begin{proof}
Suppose $q \in \bar{\pa}(v)\setminus \pspa(v)$. Then there exists a directed path $l$ from $q$ to $v$ such that $l \setminus \{q\} \subseteq \irr(v)$ from $q$ to $v$. Since $C_v \subseteq \bar{\an}(v) \setminus \bar{\sib}(v)$, then $\irr(v) \cap C_v = \emptyset$ because $\irr(v) \subseteq \bar{\sib}(v)$. Thus, $l \cap C_v = \emptyset$ so Lemma~\ref{lem:nonZeroZeta} implies that $\E \left(\gamma_q(D)^{K-1}\gamma_v(D)\right) \neq 0$ for generic parameters.

Suppose $q \in \bar{\sib}(v)$. Then either $q \in \sib(\irr(v))$ or $\irr(q) \cap \sib(\irr(v)) \neq \emptyset$. If $q \in \sib(\irr(v))$, then there exists some path $l$ such that $l \setminus \{q\} \subseteq \irr(v)$ from $s_1$ to $v$ where either $q = s_1$ (if $q \in \irr(v)$) or $q \in \sib(s_1)$ (if $q \in \sib(\irr(v)) \setminus \irr(v)$). In this case, $l \subseteq \irr(v) \cap C_v = \emptyset$, so Lemma~\ref{lem:nonZeroZeta} implies that $\E \left(\gamma_q(D)^{K-1}\gamma_v(D)\right) \neq 0$ for generic parameters. If $\irr(q) \cap \sib(\irr(v)) \neq \emptyset$, then Lemma~\ref{lem:cousins} implies the desired result.
\end{proof}

\subsection{Proof of Lemma~\ref{thm:neccIndSibMiss}}
	Consider $v \in V$ and sets $A, C$ such that $C \subseteq A \subseteq V \setminus \{v\}$. Suppose $D = H_\mathcal{C}(\bar B)$ for some $H_\mathcal{C} \in \mathcal{D}$ with $\mathcal{C}=(C_s)_{s\in V}$ such that $C_s \subseteq \bar{\an}(s) \bar{\sib}(s)$ for all $s \in V$. Suppose $u \in C$ and $u \in  \bar{\sib}(v)$, then for generic $B$ and error moments, there exists some $q \in C$ such that $\E\left(\gamma_q(D)^{K-1}\gamma_v(C, \Sigma, D)\right) \neq 0$.

\begin{proof}
If $q \in \bar{\sib}(v)$, then either $q \in \sib(\irr(v))$ or $\irr(q) \cap \sib(\irr(v)) \neq \emptyset$. If $\irr(q) \cap \sib(\irr(v)) \neq \emptyset$, then Lemma~\ref{lem:cousins} implies the desired result. Now, consider the first case where $q \in \sib(\irr(v))$. Then either $q \in \sib(v)$ or $q \in \sib(\irr(v) \setminus \sib(v)$. If $q \in \sib(v)$, then the desired result is implied by Lemma~\ref{thm:neccIndSib}. If $q \in \sib(\irr(v) \setminus \sib(v)$ then there exists some path $l \subseteq \irr(v)$ from $s_1$ to $v$ where $q \in \sib(s_1)$. Furthermore, $l \subseteq \irr(v) \cap C_v = \emptyset$, so Lemma~\ref{lem:nonZeroZeta} implies that there exists some $c \in C$ such that $\E \left(\gamma_c(D)^{K-1}\gamma_v(D)\right) \neq 0$ for generic parameters. 
\end{proof}

\subsection{Proof of Corollary~\ref{thm:noMissPruneParentsFinalMis}}
	Suppose $D = \bar B$. For $v \in V$ and generic $B$ and error moments, suppose $\bar{\pa}(v) \subseteq C \subseteq \bar{\an}(v)\setminus \bar{\sib}(v)$ and $\E(\gamma_c(D)^{K-1}\gamma_v(C, \Sigma,D)) = 0$ for all $c \in C$. If $q \in C \setminus \bar{\pa}(v)$, the for all $c \in C$
	\begin{equation}
	\E(\gamma_q(D)^{K-1}\gamma_v(C \setminus \{q\}, \Sigma,D)) = 0.
	\end{equation}
	If $q \in \pa(v)$, then there exists some $c \in C$ such that
	\begin{equation}
	    \E(\gamma_q(D)^{K-1}\gamma_v(C \setminus \{q\}, \Sigma,D)) \neq 0.
	\end{equation}

\begin{proof}
Corollary~\ref{lem:parentBarCertified} implies for any $q \in C \setminus \bar{\pa}(v)$,  \begin{equation}
\begin{aligned}
\delta_{v}(C, \psAn(C), \Sigma,D) &= \bar B_{v,(C \setminus \{q\}, q)}
= \begin{bmatrix} \bar B_{v,(C \setminus \{q\})} & 0 \end{bmatrix}
= \begin{bmatrix} \delta_{v}(C\setminus \{q\}, \psAn(C\setminus \{q\}), \Sigma ,D) & 0 \end{bmatrix}
\end{aligned}
\end{equation}
so that
\begin{equation}
\begin{aligned}
\E(\gamma_q(D)^{K-1}&\gamma_v(C \setminus \{q\}, \Sigma,D)) =  \E(\gamma_q(D)^{K-1}\gamma_v(C, \Sigma,D)) = 0.
\end{aligned}
\end{equation}

Now consider the second statement when $q \in \bar{\pa}(v)$. If $\E(\gamma_c(D)^{K-1}\gamma_v(C \setminus \{q\}, \Sigma,D)) \neq 0$ for some $c \in C \setminus \{q\}$ then the statement trivially holds. Thus, it remains to be shown that $\E(\gamma_q(D)^{K-1}\gamma_v(C \setminus \{q\}, \Sigma,D)) \neq 0$ when $\E(\gamma_c(D)^{K-1}\gamma_v(C \setminus \{q\}, \Sigma,D)) = 0$ for all $c \in C \setminus \{q\}$. This is directly implied by Lemma~\ref{thm:noMissPruneParentsMiss}.
\end{proof}

\subsection{Proof of Lemma~\ref{lem:nonZeroZeta}}
Let $D = H_\mathcal{C}(\bar B)$ for some $H_\mathcal{C} \in \mathcal{D}$ with $\mathcal{C} = (C_s)_{s \in V}$ such that $C_s \subseteq \bar{\an}(s) \setminus \bar{\sib}(s)$ for all $s \in V$. Suppose there exists some path $l = s_1 \rightarrow s_2 \rightarrow \ldots s_{|l|-1} \rightarrow v$ such that $ l \cap C_v = \emptyset$. Further suppose that $u \in \sib(s_1) \setminus l$ and $u \not \in C_v$. Then for generic parameters
\begin{equation}
    \E\left(\gamma_v(D)^{K-1} \gamma_{s_1}(D)\right) \neq 0 \qquad \text{ and } \qquad     \E\left(\gamma_v(D)^{K-1} \gamma_u(D)\right) \neq 0,
\end{equation}
so that $s_1$ and $u$ will not be pruned from $\widehat{\sib}(v)$ by Alg~\ref{alg:checkInd}. Furthermore, for any $C$ such that $u \in C$, for generic parameters there exists some $c \in C$ such that
\begin{equation}
\E\left(\gamma_c(D)^{K-1}\gamma_v(C, \Sigma, D)\right) \neq 0,
\end{equation}
so that $u$ will not be certified into $\widehat \pa(v)$.
\begin{proof}
Because $C_v \subseteq \an(v)$, we can write:
\begin{equation}
    \gamma_v(D) = \varepsilon_v + \sum_{a \in \an(v)} \zeta_{v,a} \varepsilon_a.
\end{equation}
We first show that $\zeta_{v,s} \neq 0$ for all $s \in l$. Note that $\zeta_{v,s} = \pi_{v,s} - \sum_{a \in \an(v)}d_{v,a} \pi_{a,s}$ is a rational function of the parameters because $d_{v,a}$, $\pi_{a,s}$, and $\pi_{v,s}$ are rational functions of the parameters. Thus, showing that that $\zeta_{v,s} \neq 0$ for specific parameter values implies that it is non-zero for generic parameter values.

Specifically, consider the set of edgeweights where all edges in the path $l$ are set to $1$, and all other edges are set to $0$. Then, $\pi_{v,s'} = \pi_{s', s''} =  1$, for any $ s',s'' \in l$ and $\pi_{v,w'} = \pi_{w',w''} = 0$ if either $w'$ or $w''$ is not in $l$. Then, for any $D = H_\mathcal{C}(\bar B)$ where $\pspa(v) \cap l = \emptyset$, $d_{v,w} = 0$ for all $w \not \in l$ since the only non-zero path to $v$ is $l$. Furthermore, let $\omega_{s,s} = 1$ for all $s \in V$ and let $\omega_{s_1, u} = 1$. 

Since $l \cap \pspa(v) = \emptyset$, we have
\begin{equation}
\begin{aligned}
    \gamma_v(D) &= Y_v - \sum_{w \in \pspa(v)} d_{v, w} Y_w\\
    &= \sum_{s \in \an(v)} \pi_{v,s} \varepsilon_s - \sum_{w \in \pspa(v)} d_{v, w} Y_w  \\
    &= \sum_{s \in l} \pi_{v,s}\varepsilon_s. 
\end{aligned}
\end{equation}
Thus, $\zeta_{v,s} = \pi_{v,s} - \sum_{w \in \pspa(v)} d_{v,w} \pi_{w,s} =  \pi_{v,s} = 1$ for all $s \in l$. 
This implies that $\zeta_{v,s} \neq 0$ for generic $B$.

Furthermore, since all directed edges pointing into $s_1$ are $0$, then $\gamma_{s_1}(D) = \varepsilon_{s_1}$.  For notational convenience, let $\phi = \sum_{s \in l \setminus s_1} \zeta_{v,s} \varepsilon_s$ so that $\gamma_v(D) = \zeta_{v,s_1}\varepsilon_{s_1}+ \phi$. Then
\begin{equation}
\begin{aligned}
    \E(\gamma_v(D)^{K-1} \gamma_{s_1}(D)) &= \E \left(\left(\zeta_{v,{s_1}}\varepsilon_{s_1} + \phi \right)^{K-1} \varepsilon_{s_1}\right)\\
    & = \E \left(\zeta_{v,{s_1}}^{K-1}\varepsilon_{s_1}^K + \sum_{k = 0}^{K-2} \zeta_{v,{s_1}}^{k}\varepsilon_{s_1}^{k+1} \phi^{K-1-k}\right).
\end{aligned}
\end{equation}
Since $\sum_{k = 0}^{K-2} \zeta_{v,{s_1}}^{k}\varepsilon_{s_1}^{k+1} \phi^{K-1-k}$ does not include any terms with $\varepsilon_{s_1}^K$, we can pick 
\begin{equation}
\E(\varepsilon_{s_1}^K) \neq -\frac{\E \left(\sum_{k = 0}^{K-2} \zeta_{v,{s_1}}^k\varepsilon_{s_1}^{k+1} \phi^{K-1-k} \right)}{\zeta_{v,{s_1}}^{K-1}}  
\end{equation}
so that $\E(\gamma_v(D)^{K-1} \gamma_{s_1}(D)) \neq 0 $. Because $\E(\gamma_v(D)^{K-1} \gamma_{s_1}(D))$ is a rational function of the model parameters, this implies that $\E(\gamma_v(D)^{K-1} \gamma_u(D))$ is non-zero for generic model parameters.

Similarly, $\gamma_{u}(D) = \varepsilon_{u}$ so that
\begin{equation}
\begin{aligned}
    \E(\gamma_v(D)^{K-1} \gamma_{u}(D)) &= \E \left(\left(\zeta_{v,{s_1}}\varepsilon_{s_1} + \phi \right)^{K-1} \varepsilon_{u} \right)\\
    & = \E\left(\zeta_{v,{s_1}}^{K-1}\varepsilon_{s_1}^{K-1}\varepsilon_u + \varepsilon_u\sum_{k = 0}^{K-2} \zeta_{v,{s_1}}^{k}\varepsilon_{s_1}^{k} \phi^{K-1-k}\right).
\end{aligned}
\end{equation}
Since $\varepsilon_u\sum_{k = 0}^{K-2} \zeta_{v,{s_1}}^{k}\varepsilon_{s_1}^{k} \phi^{K-1-k}$ does not include any terms with $\varepsilon_{s_1}^{K-1}\varepsilon_u$ we can pick 
\begin{equation}
\E(\varepsilon_{s_1}^{K-1}\varepsilon_u) \neq -\frac{\E \left(\varepsilon_u \sum_{k = 0}^{K-2} \zeta_{v,{s_1}}^k\varepsilon_{s_1}^{k} \phi^{K-1-k} \right)}{\zeta_{v,{s_1}}^{K-1}}  
\end{equation}
so that $\E(\gamma_v(D)^{K-1} \gamma_{u}(D)) \neq 0$. This implies that $\E(\gamma_v(D)^{K-1} \gamma_{u}(D)) \neq 0$ for generic parameters.

We now show that $u$ will not be certified into $\widehat{\pa}(v)$. Recall that 
\begin{equation}
    \delta_v(C, A, \Sigma, D) = \left\{\left[(I-D)_{C,A} \Sigma_{A,C}\right]^{-1}(I-D)_{C,A}\Sigma_{A,v}\right\}^T = \left(\Sigma_{C,C}\right)^{-1} \Sigma_{C,v}.
\end{equation}
There are no non-zero treks between $w$ and $v$ for any $w \not \in l \cup u$, so $\Sigma_{w, v} = 0$ for all $w \in C \setminus u$. Furthermore, $\Sigma_{u,v} = 1$ since there is a single trek between $u$ and $v$ and all edge weights on that trek are $1$. Furthermore, $\Sigma_{C,C}$ is a diagonal matrix with $\omega_{c,c} = 1$ on the diagonals (i.e., $\Sigma_{C,C}$ is the identity) since there are no non-zero treks between any nodes in $C$. Thus, $\delta_v(C,A,\Sigma, D)$ is $0$ except for the element corresponding to $u$ which is $ 1 / \omega_{u,u} = 1$.

Thus, 
\begin{equation}
\gamma_v(C, \Sigma, D) = \sum_{s \in l} \varepsilon_s - \frac{1}{\omega_{u,u}} \varepsilon_u.
\end{equation}
For notational convenience, let $\phi = \sum_{s \in l} \varepsilon_s$ and let $\zeta_{v,u} = -1 / \omega_{u,u}$. Using a similar argument as before, picking 
\begin{equation}
\E(\varepsilon_{u}^{K}) \neq -\frac{\E \left(\sum_{k = 0}^{K-2} \zeta_{v,u}^k\varepsilon_{u}^{k+1} \phi^{K-1-k} \right)}{\zeta_{v,u}^{K-1}}  
\end{equation}
implies that $\E(\gamma_v(C, \Sigma, D)^{K-1} \gamma_u(D)) \neq 0 $.
\end{proof}

\subsection{Proof of Lemma~\ref{lem:cousins}}
Suppose $\irr(u) \cap \sib(\irr(v)) \neq \emptyset$ and $D = H_{\mathcal{C}}(\bar{B})$ for some $\mathcal{C} = (C_s)_{s \in V}$ such that $C_s \subseteq \bar{\an}(s) \setminus \bar{\sib}(s)$ for all $s \in V$. Then, for generic parameters
\begin{equation}
     \E\left(\gamma_v(D)^{K-1} \gamma_u(D)\right) \neq 0
\end{equation}
so that $u$ will not be pruned from $\widehat{\sib}(v)$ by Alg~\ref{alg:checkInd}. Furthermore, for any $C \subseteq V \setminus v$ such that $u \in C$, for generic parameters, there exists some $c \in C$ such that
\begin{equation}
\E\left(\gamma_c(D)^{K-1}\gamma_v(C, \Sigma, D)\right) \neq 0,
\end{equation}
so that $u$ will not be certified into $\widehat \pa(v)$.

\begin{proof}
We first show that $u$ will not be pruned out of $\widehat{\sib}(v)$ by Alg~\ref{alg:checkInd}. We subsequently show that $u$ will not be certified into $\widehat{\pa}(v)$.

\vspace{1.5em}
\textbf{Not pruned from $\widehat{\sib}(v)$}:
Let $q \in \irr(u) \cap \sib(\irr(v))$ and $w \in \sib(q) \cap \irr(v)$. Then there exists a directed path $l_1$ from $w$ to $v$ such that $l_1 \subseteq \irr(v)$ so that $C_{v} \cap l_1 = \emptyset$. As shown in Lemma~\ref{lem:nonZeroZeta}, this implies that $\gamma_v(D) = \zeta_{v,w} \varepsilon_w + \phi_v$ where $\zeta_{v,w} \neq 0$ for generic parameters and some term $\phi_v$ which does not include $\varepsilon_w$. A similar statement can be made for $q$ and $u$ so that $\gamma_u(D) = \zeta_{u,q}\varepsilon_q + \phi_u$. Thus,
\begin{equation}
\begin{aligned}
    \E\left(\gamma_v(D)^{K-1} \gamma_u(D)\right) &= \E\left(\left[\zeta_{v,w}\varepsilon_w + \phi_v\right]^{K-1} \left[\zeta_{u,q}\varepsilon_q + \phi_u\right]  \right)\\
    &= \E\left(\zeta_{v,w}^{K-1}\varepsilon_w^{K-1}\left[\zeta_{u,q}\varepsilon_q + \phi_u\right] +\left[\zeta_{u,q}\varepsilon_q + \phi_u\right]\sum_{k = 0}^{K-2}\zeta_{v,w}^{k}\varepsilon_w^{k} \phi_v^{K-k-1}  \right)\\
    &= \E\left(\zeta_{v,w}^{K-1}\varepsilon_w^{K-1}\zeta_{u,q}\varepsilon_u +\zeta_{v,w}^{K-1}\varepsilon_v^{K-1}\phi_q \right. \\ 
    & \quad \qquad +\left. \left[\zeta_{u,q}\varepsilon_q + \phi_u\right]\sum_{k = 0}^{K-2}\zeta_{v,w}^{k}\varepsilon_w^{k} \phi_v^{K-k-1}  \right).\\
\end{aligned}
\end{equation}
Since $\zeta_{v,w}^{K-1}\varepsilon_w^{K-1}\phi_u +\left[\zeta_{u,q}\varepsilon_q + \phi_u\right]\sum_{k = 0}^{K-2}\zeta_{v,w}^{k}\varepsilon_w^{k} \phi_v^{K-k-1} $ does not involve any terms with $\varepsilon_w^{K-1}\varepsilon_q$, we can select
\begin{equation}
    \E\left(\varepsilon_w^{K-1}\varepsilon_q\right) \neq - \frac{\E\left(\zeta_{v,w}^{K-1}\varepsilon_w^{K-1}\phi_u +\left[\zeta_{u,q}\varepsilon_q + \phi_u\right]\sum_{k = 0}^{K-2}\zeta_{v,w}^{k}\varepsilon_w^{k} \phi_v^{K-k-1}  \right)}{\zeta_{v,w}^{K-1}\zeta_{u,q}}    
\end{equation}
so that $\E\left(\gamma_v(D)^{K-1} \gamma_u(D)\right) \neq 0$. This implies that $\E\left(\gamma_v(D)^{K-1} \gamma_u(D)\right) \neq 0$ for generic parameters.

\vspace{1.5em}
\textbf{Not certified into $\widehat{\pa}(v)$}:
We now show that $u$ will not be certified into $\widehat{\pa}(v)$ so that $u$ will remain in $\widehat{\sib}(v)$. Specifically, we show that any set $C$ will not be certified if $u \in C$. If $u\in \sib(\irr(v))$ or if $C \cap \sib(\irr(v)) \neq \emptyset$, then Lemma~\ref{lem:nonZeroZeta} directly completes the proof. Thus, it remains to be shown that $C$ will not be certified even if $u \not \in \sib(\irr(v))$ and $C \cap \sib(\irr(v)) = \emptyset $. Without loss of generality, assume that $1, \ldots, p$ is a valid causal ordering of $V$.
We consider two cases: (1) $\irr(u) \cap \irr(v) \neq \emptyset$ and (2) $\irr(u) \cap \irr(v) = \emptyset$.

For the first case, let $w = \max(\irr(u) \cap \irr(v))$. Then there exist two directed paths $l_1$ and $l_2$ such that $l_1$ is a directed path from $w$ to $v$ with $l_1 \subseteq \irr(v)$, $l_2$ is a directed path from $w$ to $u$ with $l_2 \subseteq \irr(u)$, and $\{w\} = l_1 \cap l_2$. Let $s_0 = \min\{s \in C : C_s \cap l_2 = \emptyset\}$ so that $s_0$ is the most upstream node in $l_2$ whose currently estimated parent set, $C_s$, does not contain any other nodes in $l_2$. The set over which the $\min$ is taken is non-empty because $l_2 \subset \irr(u)$ so that $C_u \cap l_2 = \emptyset$. For each $s \in C$, let $m_s = \max(l_2 \cap C_s)$ where $m_s = 0$ if $l_2 \cap C_s$ is empty. Now set the directed edges on $l_1$ and all directed edges on $l_2$ before $s_0$ to $1/p^3$. Set $\omega_{v,v} = 1$ for all $v \in V$, and set all other directed and bidirected edgeweights to $0$. Finally, for any node $s$ in $l_2$, let $L(s)$ denote the position of $s$ in $l_2$; i.e., if $l = s_1 \rightarrow s_2 \rightarrow s_3 \ldots$ then $L(s_i) = i$.

Suppose $s \in C \cap l_2$ and $s < s_0$. Since the only non-zero directed edges are within $l_2$ and $m_s = \max(C_s \cap l_2)$, then all non-zero directed paths from $C_s$ to $s$ must pass through $m_s$. Thus, for $C_s$ the marginal direct effect is $0$ for any $t \neq m_s$ and the marginal direct effect of $m_s$ is $\pi_{s, m_s} = p^{-3(L(s) - L(m_s))}$. Thus, for $s < s_0$  
\begin{equation}
\begin{aligned}
    \gamma_{s}(D) &=  Y_{s} - D_{s, C_{s}}Y_{C_{s}} = Y_{s} - \pi_{s, m_s} Y_{m_s}\\
&= \varepsilon_{s} + \sum_{\substack{ s' \in l_2\\ s > s' > m_s}}\pi_{s,s'}\varepsilon_{s'} + \pi_{s, m_s}Y_{m_s} - \pi_{s, m_s} Y_{m_s}\\
&= \varepsilon_{s} + \sum_{\substack{ s' \in l_2\\ s > s' > m_s}}\pi_{s,s'}\varepsilon_{s'}
\end{aligned}
\end{equation}
and
\begin{equation}
Y_{s} =  \varepsilon_{s} + \sum_{\substack{ s' \in l_2\\ s' < s }}\pi_{s, s'}\varepsilon_{s'}.
\end{equation}
Because $C_{s_0} \cap l_2 = \emptyset$, no node in $C_{s_0}$ has a non-zero directed path to $s$. Thus, $D_{s_0, V}$ is the zero vector and  
\begin{equation}
\begin{aligned}
    \gamma_{s_0}(D) = Y_{s_0} =  \varepsilon_{s} + \pi_{s_0, w} \varepsilon_w +  \sum_{ s \in l_2 \setminus w}\pi_{s_0,s}\varepsilon_{s}.
\end{aligned}
\end{equation}

For all other $s \in C$, then either $s \not \in l_2$ or $s \in l_2$ but $s > s_0$. Then all directed paths into $s$ have weight $0$, so $D_{s, V} = 0$. Thus,
\begin{equation}
\begin{aligned}
    \gamma_{s}(D) &=  Y_{s} = \varepsilon_{s}.
\end{aligned}
\end{equation}
Notably, $\gamma_{s}(D)$ only contains a term with $\varepsilon_w$ if $s = s_0$. We now show that under these parameters, when checking the certificate for $C$, $\delta_{v,s_0}(C, D, \Sigma) \neq 0$.

Note that $(I-D)_{C,A} \Sigma_{A, C} = \E\left((Y_C - D_{C,A}Y_A) Y_C^T\right) = \E\left(\gamma_C Y_C^T\right)$ and $(I-D)_{C,A} \Sigma_{A, v} = \E\left((Y_C - D_{C,A}Y_A) Y_v \right) = \E\left(\gamma_C Y_v \right)$. Then, letting $M = (I-D)_{C,A} \Sigma_{A, C}$ for convenience, we first show that $M$ is diagonally dominant so that it is non-singular and $M_{V\setminus s_0, V \setminus s_0}$ is also non-singular.

For any $s \in C$, since all edgeweights are positive and $\E(\varepsilon_s \varepsilon_r) = 0$ for $s,r \in C$, we have 
\begin{equation}
    M_{s,s} = \E\left([\varepsilon_{s} + \sum_{\substack{ s' \in l_2\\ s > s' > m_s}}\pi_{s,s'}\varepsilon_{s'}][\varepsilon_{s} + \sum_{\substack{ s' \in l_2\\ s' < s }}\pi_{s, s'}\varepsilon_{s'}]\right) > \E(\varepsilon_s^2) = 1.
\end{equation}

If $s > s_0$ then $\gamma_s = \varepsilon_s$, but since all directed edges downstream of $s$ are set to $0$, then $\varepsilon_s$ does not appear in any other $Y_{C \setminus s}$ so $M_{s,C\setminus s} = 0$. Similarly, since $Y_s = \varepsilon_s$, then then $\varepsilon_s$ does not appear in any other $\gamma_{C \setminus s}$ so $M_{C\setminus s, s} = 0$. It then remains to characterize $M_{r,s}$ when $r,s M s_0$. In this case,
\begin{equation}\begin{aligned}
    M_{s,r} &= \E\left([\varepsilon_{s} + \sum_{\substack{ s' \in l_2\\ s > s' > m_s}}\pi_{s,s'}\varepsilon_{s'}][\varepsilon_{r} + \sum_{\substack{ s' \in l_2\\ s' < r }}\pi_{r, s'}\varepsilon_{s'}]\right)\\
    &= \sum_{\substack{ s' \in l_2\\ s' < r\\ s > s' > m_s}} \pi_{s,s'} \pi_{r,s'}\E\left(\varepsilon_{s'}^2\right)\\
    & < p \max_{j,k \in V}(\pi_{j,k}^2) < \frac{1}{p^2}.
\end{aligned}
\end{equation}

Thus, $\vert M_{s,r} \vert < \frac{1}{p^2}$ for any $s \neq r$ and $M_{s,s} \geq 1$. This implies that $M$ and $M_{V \setminus s_0, V \setminus s_0}$ are both diagonally dominant so that
\begin{equation}
    \vert (M^{-1})_{s_0,s_0}\vert = \left \vert\frac{\text{det}(M_{V \setminus s_0, V \setminus s_0})}{\text{det}(M)} \right \vert \neq 0.
\end{equation}

Letting $F = (I-D)_{C,A} \Sigma_{A, v}$, we have:
\begin{equation}
(F)_{s} = \E(\gamma_{s} Y_{v}) = \E\left(\left[\varepsilon_{s} + \sum_{\substack{ s' \in l_2\\ s > s' > m_s}}\pi_{s,s'}\varepsilon_{s'}\right] \left[\varepsilon_v + \sum_{s'' \in l_1} \pi_{v,s''} \varepsilon_{s''}\right] \right).
\end{equation}
Since $l_1 \cap l_2 = w$ and all covariances are set to $0$, then
$F_s \neq 0$ only if $\gamma_s$ contains $\varepsilon_w$. Thus, $F_s =
0$ for all $s \neq s_0$ and $F_{s_0} = \E(\pi_{s_0,w} \pi_{v,w}
\varepsilon_w^2) \neq 0$. Combining all the results, we then have that
the $s_0$ element of $\delta_{v,C}(C,D,\Sigma)$ is $M_{s_0,s_0}
F_{s_0} \neq 0$. Since $Y_{s_0}$ is the only $Y$ which has an
$\varepsilon_{s_0}$ term, it holds that
\begin{equation}
    \gamma_v(C,\Sigma, D) = Y_v - \delta_{v,C} Y_C = \delta_{v,s_0} \varepsilon_{s_0} + \phi_v,
\end{equation}
where $\phi_v$ does not involve $\varepsilon_{s_0}$. Similarly, we can write
\begin{equation}
    \gamma_{s_0}(D) = Y_{s_0} - D_{s_0,V} Y = \varepsilon_{s_0} + \phi_{s_0}, 
\end{equation}
where $\phi_{s_0}$ does not involve $\varepsilon_{s_0}$. Then, 
\begin{equation}
\begin{aligned}
    \E\left(\gamma_{s_0}(D)^{K-1} \gamma_v(C, \Sigma, D)\right) &= \E\left(\left[\varepsilon_{s_0} + \phi_{s_0}\right]^{K-1} \left[\delta_{v,s_0}\varepsilon_{s_0} + \phi_v\right]  \right)\\
    &= \E\left(\varepsilon_{s_0}^{K-1}\left[\delta_{v,{s_0}}\varepsilon_{s_0} + \phi_v\right] + \left[\delta_{v,{s_0}}\varepsilon_{s_0} + \phi_v\right]\sum_{k = 0}^{K-2}\varepsilon_{s_0}^{k} \phi_{s_0}^{K-k-1}  \right)\\
&= \E\left(\varepsilon_{s_0}^{K}\delta_{v,{s_0}} + \varepsilon_{s_0}^{K-1}\phi_v + \left[\delta_{v,{s_0}}\varepsilon_{s_0} + \phi_v\right]\sum_{k = 0}^{K-2}\varepsilon_{s_0}^{k} \phi_{s_0}^{K-k-1}  \right).
\end{aligned}
\end{equation}
Therefore, selecting
\begin{equation}
\E\left(\varepsilon_{s_0}^{K}\right) \neq -\frac{ \E\left(\varepsilon_{s_0}^{K}\delta_{v,{s_0}} + \varepsilon_{s_0}^{K-1}\phi_v + \left[\delta_{v,{s_0}}\varepsilon_{s_0} + \phi_v\right]\sum_{k = 0}^{K-2}\varepsilon_{s_0}^{k} \phi_{s_0}^{K-k-1}  \right)}{\delta_{v,{s_0}}}     
\end{equation}
implies that $ \E\left(\gamma_{s_0}(D)^{K-1} \gamma_v(C, \Sigma, D)\right) \neq 0$ and thus $ \E\left(\gamma_{s_0}(D)^{K-1} \gamma_v(C, \Sigma, D)\right) \neq 0$ for generic parameters.

Now we slightly modify the argument above to the case where $\irr(u) \cap \irr(v) = \emptyset$. Let $w \in \sib(\irr(v)) \cap \irr(u)$, and let $q \in \sib(w) \cap \irr(v)$. Select two paths, $l_1$ and $l_2$, such that $l_1 \cap l_2 = \emptyset$ where $l_1$ is be a path from $q$ to $v$ which only passes through $\irr(v)$ and $l_2$ to be a path from $w$ to $u$ which only passes through $\irr(u)$. Similar to before, let $s_0 = \min(\{s \in C : C_s \cap l_2 = \emptyset\})$ so that $s_0$ is the most upstream node in $l_2$ whose currently estimated parent set, $C_s$, does not contain any other nodes in $l_2$. The set over which the $\min$ is taken is non-empty because $l_2 \subset \irr(u)$ so that $C_u \cap l_2 = \emptyset$. For all $s \in C \cap l_2$ such that $s < s_0$, let $m_s = \max(l_2 \cap C_2)$ where $C_s$ is the set of sets in $\mathcal{C}$. Now consider the set of parameters where all directed edges on $l_1$ and all directed edges before $s_0$ on $l_2$ are set to $1/p^3$. Let $\E(\varepsilon_w \varepsilon_q) = 1$, and set all other directed and bidirected edgeweights to $0$. In addition, set $\omega_{v,v} = 1$ for all $v \in V$. 

Using the same argument as before, $M = (I-D)_{C,A} \Sigma_{A, C}$ is diagonally dominant so $(M^{-1})_{s_0, s_0 } \neq 0$. Furthermore, letting $F = (I-D)_{C,A} \Sigma_{A, v}$, we have:
\begin{equation}
(F)_{s} = \E(\gamma_{s} Y_{v}) = \E\left(\left[\varepsilon_{s} + \sum_{\substack{ s' \in l_2\\ s > s' > m_s}}\pi_{s,s'}\varepsilon_{s'}\right] \left[\varepsilon_v + \sum_{s \in l_1} \pi_{v,s} \varepsilon_s\right] \right).
\end{equation}
Since $l_1 \cap l_2 = \emptyset$ and the only errors with non-zero covariance are $w$ and $q$, then $(F)_{s} \neq 0$ only if $\gamma_{s}$ contains $\varepsilon_w$. By construction, only $\gamma_{s_0}$ contains $\varepsilon_w$ so $F$ is zero except for the element corresponding to $s_0$ and $F_{s_0} =\E(\pi_{s_0,w}\varepsilon_w \pi_{v,q} \varepsilon_q) = \pi_{s_0,w}\pi_{v,q}\E(\varepsilon_w \varepsilon_q) \neq 0 $. This implies that $\delta(C, \Sigma, D)_{s_0} = (M^{-1})_{s_0,s_0} F_{s_0} \neq 0$.

Since $Y_{s_0}$ is the only $Y$ which has an $\varepsilon_{s_0}$ term,
we obtain that
\begin{equation}
    \gamma_v(C,\Sigma, D) = Y_v - \delta_{v,C} Y_c = \delta_{v,s_0} \varepsilon_{s_0} + \phi_v, 
\end{equation}
where $\phi_v$ does not include any terms with $\varepsilon_{s_0}$. Similarly,  $\gamma_{s_0}(D) = \varepsilon_{s_0} + \phi_{s_0}$ where $\phi_{s_0}$ does not involve $\varepsilon_{s_0}$. Then, 
\begin{equation}
\begin{aligned}
    \E\left(\gamma_{s_0}(D)^{K-1} \gamma_v(C, \Sigma, D)\right) &= \E\left(\left[\varepsilon_{s_0} + \phi_{s_0}\right]^{K-1} \left[\delta_{v,s_0}\varepsilon_{s_0} + \phi_v\right]  \right)\\
    &= \E\left(\varepsilon_{s_0}^{K-1}\left[\delta_{v,{s_0}}\varepsilon_{s_0} + \phi_v\right] + \left[\delta_{v,{s_0}}\varepsilon_{s_0} + \phi_v\right]\sum_{k = 0}^{K-2}\varepsilon_{s_0}^{k} \phi_{s_0}^{K-k-1}  \right)\\
&= \E\left(\varepsilon_{s_0}^{K}\delta_{v,{s_0}} + \varepsilon_{s_0}^{K-1}\phi_v + \left[\delta_{v,{s_0}}\varepsilon_{s_0} + \phi_v\right]\sum_{k = 0}^{K-2}\varepsilon_{s_0}^{k} \phi_{s_0}^{K-k-1}  \right)\\
\end{aligned}
\end{equation}
Thus, selecting
\begin{equation}
\E\left(\varepsilon_{s_0}^{K}\right) \neq -\frac{ \E\left(\varepsilon_{s_0}^{K}\delta_{v,{s_0}} + \varepsilon_{s_0}^{K-1}\phi_v + \left[\delta_{v,{s_0}}\varepsilon_{s_0} + \phi_v\right]\sum_{k = 0}^{K-2}\varepsilon_{s_0}^{k} \phi_{s_0}^{K-k-1}  \right)}{\delta_{v,{s_0}}}.     
\end{equation}
implies that $ \E\left(\gamma_{s_0}(D)^{K-1} \gamma_v(C, \Sigma, D)\right) \neq 0$ and thus $ \E\left(\gamma_{s_0}(D)^{K-1} \gamma_v(C, \Sigma, D)\right) \neq 0$ for generic parameters.
\end{proof}

\end{document}